\pdfoutput=1
\documentclass[11pt]{report}

\usepackage{amsmath}
\usepackage{amssymb}
\usepackage{amsthm}
\usepackage{fullpage}
\usepackage{graphicx}

\usepackage{tikz}

\usepackage[bookmarks,colorlinks,breaklinks]{hyperref}  
\hypersetup{linkcolor=blue,citecolor=blue,filecolor=blue,urlcolor=blue} 
\numberwithin{equation}{section}

\renewcommand{\chaptername}{Chapter}

\renewcommand{\chaptername}{Lecture}
\newcommand{\lecture}[5]{
\chapter[\texorpdfstring{#5 \\({\em Lecturer: #4, Scribe: #3})}{#5}]{{\texorpdfstring{{\LARGE #5}\\{\Large
      {\em #4}}\\\hspace*{\fill}{\large Scribe:
      #3}\\\vspace{-0.5cm}\hspace*{\fill}{\normalsize #2}}{#5}}}

}

\newcommand{\lref}[2][]{\hyperref[#2]{#1~\ref*{#2}}}
\renewcommand{\eqref}[2][]{\hyperref[#2]{(\ref*{#2})}}

\newtheorem{theorem}{Theorem}[section]
\newtheorem{corollary}[theorem]{Corollary}
\newtheorem{lemma}[theorem]{Lemma}
\newtheorem{observation}[theorem]{Observation}
\newtheorem{proposition}[theorem]{Proposition}
\newtheorem{definition}[theorem]{Definition}
\newtheorem{claim}[theorem]{Claim}
\newtheorem{fact}[theorem]{Fact}

\newcommand{\R}{{\mathbb R}}
\newcommand{\mathify}[1]{\ifmmode{#1}\else\mbox{$#1$}\fi}
\newcommand{\abs}[1]{\mathify{\left| #1 \right|}}
\newcommand{\ip}[1]{{\langle #1 \rangle}}
\newcommand{\E}{{\mathbb E}}
\newcommand{\prob}[2]{\underset{#1}{\rm Prob}\left[{#2}\right]}
\newcommand{\F}{{\mathbb F}}
\newcommand{\set}[1]{\mathify{\left\{ #1 \right\}}}
\newcommand{\poly}{{\rm poly}}
\newcommand{\acc}{{\sf acc}}
\newcommand{\bigO}{O}

\newcommand{\X}{{\mathcal X}}
\newcommand{\Y}{{\mathcal Y}}
\newcommand{\feas}{{\mathrm {feas}}}
\newcommand{\val}{\mathrm{value}}

\newcommand{\myopt}{\mathrm{OPT}} 
\newcommand{\mybits}{\ensuremath{\{0,1\}}} 
\newcommand{\mymaxthreesat}{MAX-$3$SAT} 
\newcommand{\YES}{\mathrm{YES}}
\newcommand{\NO}{\mathrm{NO}}
\newcommand{\calA}{{\mathcal A}}
\newcommand{\calB}{{\mathcal B}}

\newcommand{\NP}{\ensuremath{\mathsf{NP}}}
\newcommand{\NQP}{\ensuremath{\mathsf{NQP}}}
\newcommand{\RQP}{\ensuremath{\mathsf{RQP}}}

\newcommand{\agr}{{\rm agr}}

\newcommand{\circuitsat}{{CIRCUIT-SAT}}

\newcommand{\hastad}{H{\aa}stad}
\newcommand{\labelcover}{\mathsf{LABEL\ COVER}}
\newcommand{\maxcut}{\mathsf{MAX}\text{-}\mathsf{CUT}}
\newcommand{\maxsat}{\mathsf{MAX}\text{-}3\mathsf{SAT}}
\newcommand{\maxlin}{\mathsf{MAX}\text{-}3\mathsf{LIN}}
\newcommand{\eps}{\varepsilon}
\newcommand{\np}{\mathsf{NP}}
\newcommand{\p}{\mathsf{P}}
\newcommand{\lc}{\mathcal{L}}
\newcommand{\id}{\mathsf{id}}
\newcommand{\bits}{\{\pm 1\}}
\newcommand{\fc}[1]{\widehat{#1}}
\newcommand{\expect}{\operatorname{\mathbb{E}}}
\DeclareMathOperator{\opt}{opt}

\newcommand{\rvector}{{\bf r}}
\newcommand{\vectoru}{{\bf u}}
\newcommand{\vectorv}{{\bf v}}

\newcommand{\Exp}[2]{\underset{#1}{\mathbb E}\left[{#2}\right]}
\newcommand{\beq}{\begin{equation}}
\newcommand{\eeq}{\end{equation}}
\newcommand{\uginst}{{\mathcal L}}
\newcommand{\uginstfull}{{\mathcal L}(V,W,E,[M],\{\pi_{v,w}\})}
\DeclareMathOperator{\infl}{Inf_i}
\DeclareMathOperator{\inflk}{Inf^k_i}
\DeclareMathOperator{\ns}{NS_{\rho}}
\DeclareMathOperator{\mc}{mc}

\title{\texorpdfstring{\large \bfseries{Limits of Approximation Algorithms: PCPs and Unique Games\\
  (DIMACS Tutorial Lecture Notes)}\thanks{Jointly sponsored by the
  DIMACS Special Focus on Hardness of Approximation, the DIMACS
  Special Focus on Algorithmic Foundations of the Internet, and the
  Center for Computational Intractability with support from the
  National Security Agency and the National Science
  Foundation.}}{Limits of Approximation Algorithms}}
\author{\texorpdfstring{\textit{Organisers: Prahladh Harsha \& Moses
      Charikar}}{Harsha \and Charikar}}
\date{}

\begin{document}

\pagenumbering{empty}

\maketitle

\pagenumbering{roman}

 \chapter*{Preface}
 \addcontentsline{toc}{chapter}{Preface}

 These are the lecture notes for the DIMACS Tutorial {\em Limits of
     Approximation Algorithms: PCPs and Unique Games} held at the DIMACS
   Center, CoRE Building, Rutgers University on 20-21 July, 2009. This
   tutorial was jointly sponsored by the DIMACS Special Focus on
   Hardness of Approximation, the DIMACS Special Focus on Algorithmic
   Foundations of the Internet, and the Center for Computational
   Intractability with support from the National Security Agency and
   the National Science Foundation.

 The speakers at the tutorial were Matthew Andrews, Sanjeev Arora,
 Moses Charikar, Prahladh Harsha, Subhash Khot, Dana
 Moshkovitz and Lisa Zhang. We thank the scribes -- Ashkan Aazami, Dev
 Desai, Igor Gorodezky, Geetha
 Jagannathan, Alexander S. Kulikov, 
 Darakhshan J. Mir, Alantha Newman, Aleksandar Nikolov, David
 Pritchard and Gwen Spencer for their thorough and meticulous work.  

 Special thanks to Rebecca Wright and Tami Carpenter at DIMACS but for
 whose organizational support and help, this workshop would have been
 impossible. We thank Alantha Newman, a phone conversation with whom
 sparked the idea of this workshop. We thank the Imdadullah Khan and
 Aleksandar Nikolov for video recording the lectures. The video
 recordings of the lectures will be posted at the DIMACS tutorial webpage

 \href{http://dimacs.rutgers.edu/Workshops/Limits/}{\path{http://dimacs.rutgers.edu/Workshops/Limits/}}

 Any comments on these notes are
 always appreciated.\\

 \begin{flushright}
 \begin{minipage}{1.5in}
 Prahladh Harsha\\
 Moses Charikar

 30 Nov, 2009.
 \end{minipage}
 \end{flushright}

\cleardoublepage

\chapter*{Tutorial Announcement}
\addcontentsline{toc}{chapter}{Tutorial Announcement}

\noindent {\bf DIMACS Tutorial}\\
{\bf Limits of Approximation Algorithms: PCPs and Unique Games}\\
DIMACS Center, CoRE Building, Rutgers University, July 20 - 21, 2009\\

\noindent Organizers:\\
* Prahladh Harsha, University of Texas, Austin\\ 
* Moses Charikar, Princeton University 

This tutorial is jointly sponsored by the DIMACS Special Focus on Hardness of Approximation, the DIMACS Special Focus on Algorithmic Foundations of the Internet, and the Center for Computational Intractability with support from the National Security Agency and the National Science Foundation.

\vspace{0.25cm}
\hrule
\vspace{0.25cm}

The theory of NP-completeness is one of the cornerstones of complexity theory in theoretical computer science. Approximation algorithms offer an important strategy for attacking computationally intractable problems, and approximation algorithms with performance guarantees have been designed for a host of important problems such as balanced cut, network design, Euclidean TSP, facility location, and machine scheduling. Many simple and broadly-applicable approximation techniques have emerged for some provably hard problems, while in other cases, inapproximability results demonstrate that achieving a suitably good approximate solution is no easier than finding an optimal one. The celebrated PCP theorem established that several fundamental optimization problems are not only hard to solve exactly but also hard to approximate. This work shows that a broad class of problems is very unlikely to have constant factor approximations, and in effect, establishes a threshold for such problems such that approximation beyond this threshold would imply P= NP. More recently, the unique games conjecture of Khot has emerged as a powerful hypothesis that has served as the basis for a variety of optimal inapproximability results.

This tutorial targets graduate students and others who are new to the field. It will aim to give participants a general overview of approximability, introduce them to important results in inapproximability, such as the PCP theorem and the unique games conjecture, and illustrate connections with mathematical programming techniques.

List of speakers:  Matthew Andrews (Alcatel-Lucent Bell Laboratories),
     Sanjeev Arora (Princeton University),
     Moses Charikar (Princeton University),
     Prahladh Harsha (University of Texas, Austin),
     Subhash Khot (New York University),
     Dana Moshkovitz (Princeton University) and 
     Lisa Zhang (Alcatel-Lucent Bell Laboratories)

\cleardoublepage
\tableofcontents

\cleardoublepage
\pagenumbering{arabic}

\lecture{1}{20 July, 2009}{Darakhshan J. Mir}{Sanjeev Arora}{An Introduction to Approximation Algorithms}


In this lecture, we will introduce the notion of approximation algorithms and see examples of approximation algorithms for a variety of NP-hard optimization
problems.

\section{Introduction}
Let $Q$ be an optimization problem\footnote{Formally, a (maximization) optimization problem is specified by two domains $\X,\Y$, a feasibility function $\feas:\X\times\Y \to \{ 0,1\}$ and an evaluation function $\val:\X\times\Y\to \R$. An input instance to the problem is an element $x \in \X$. For each such $x$, the optimization problem is as follows:
$$OPT(x) = \max\left\{\val(x,y) | y \in \Y, \feas(x,y) =1\right\}.$$
$OPT(x)$ is also called the optimal value.}. An optimal solution for an instance of this optimization problem is a feasible solution that achieves the best value for the objective function. Let $OPT(I)$ denote the value of the objective function for an optimal solution to an instance $I$.
\begin{definition}[Approximation ratio]
An algorithm  for $Q$ has an approximation ratio $\alpha$ if for instances $I$, the algorithm produces a solution of cost $\leq \alpha \cdot OPT(I)$ ($\alpha \geq 1$), if $Q$ is a minimization problem and of cost $\geq \alpha \cdot OPT(I)$ if $Q$ is a maximization problem.
 
\end{definition}
 We are interested in polynomial-time approximation algorithms for NP-hard problems. How does a polynomial-time approximation algorithm know what the cost of the optimal solution is, which is NP-hard to compute? How does one guarantee that the output of the algorithm is within $\alpha$ of the optimal solution when it is NP-hard to compute the optimal solution. 
In various examples below, we see techniques of handling this dilemma.

\subsection{Examples}
\begin{enumerate}
 \item \textbf{2-approximation for metric Travelling Salesman Problem (metric-TSP):}
Consider a complete graph $G$ formed by $n$ points in a metric space. Let $d_{ij}$ be the distance between point $i$ and $j$. The metric TSP problem is to find a minimum cost cycle that visits every point exactly once.

The following observation relating the cost of the minimum spanning tree (MST) to the optimal TSP will be crucial in bounding the approximation ratio.
\begin{observation}
 The cost of the Minimum spanning Tree (MST) is at most the optimal cost of TSP.
\end{observation}
\textbf{Algorithm $A$:}
\begin{enumerate}
\item Find the MST
\item Double each edge
\item Do an ``Eulerian transversal'' and output its cost
\end{enumerate}
Observe that $TSP\leq cost(A) \leq 2\cdot MST \leq 2\cdot TSP$.

\item \textbf{A 1.5-approximation to metric-TSP:}
The approximation ratio can be improved to 1.5 by modifying the above using an idea due to Christofides~\cite{Christofides1976}. Instead of doubling each edge of the MST as in the above algorithm, a minimum cost matching is added among all odd degree nodes. Observe that cost of matching $\leq \frac{1}{2} TSP$.
So,
\[Cost(\textit{Appx-algo}) \leq MST+ \frac{1}{2}TSP \leq 1.5\cdot TSP\]

It is to be noted that since 1976, there has been no further improvement on this approximation ratio.

\end{enumerate}

The above examples are examples of approximation algorithms that attain a constant approximation ratio.
In the next section, we will see how to get arbitrarily close to the optimal solution when designing an approximation algorithm, ie., approximation ratios arbitrarily close to 1.

\section{Polynomial-time Approximation Scheme (PTAS)}
A PTAS is a family of polynomial-time algorithms, such that for every $\epsilon >0$, there is an algorithm in this family that is an $(1+\epsilon)$ approximation to the NP-hard problem $Q$, if it is a minimization problem and an $(1-\epsilon)$-approximation if $Q$ is a maximization problem.

The above definition allows the running time to arbitrarily depend on $\epsilon$ but for each $\epsilon$ it should be polynomial in the input size e.g. $n^{\frac{1}{\epsilon}}$ or $n^{2^{\frac{1}{\epsilon}}}$.

\subsection{Type-1 PTAS}
Various type of number problems typically have type-1 PTAS. The usual strategy is to try to round down the numbers involved , so the choice of numbers is small and then use Dynamic Programming. The classic example of such an approach is the Knapsack problem.\\

\noindent \textbf{Knapsack problem} Given a set of $n$ items, of sizes $s_1, s_2 \ldots s_n$ such that $s_i \leq 1~ \forall i$, and profits $c_1, c_2, \ldots, c_n$, associated with these items, and a knapsack of capacity 1, find a subset $I$ of items whose
total size is bounded by 1 such that the total profit is maximized.

The knapsack problem is NP-hard in general, however if the profits fall in a small-sized set, then there exists an efficient polynomial time algorithm.
\begin{observation}
 If the values $c_1, c_2, \ldots c_n$ are in  $ [1, \ldots, w]$, then the problem can be solved in $poly(n,w)$-time using dynamic programming.
\end{observation}
 
This naturally leads to the following approximation algorithm for knapsack.\\

\noindent \textbf{$(1+\epsilon)$-Approximation Algorithm}
\begin{enumerate}
\item Let $c= \max_i c_i$.
\item Round down each $c_i$ to the nearest multiple of $\frac{\epsilon c}{n}$. Let this quantity be $r_i\cdot \left(\frac{\epsilon c}{n}\right)$, i.e., $r_i = \lfloor c_i/\frac{\epsilon c}{n}\rfloor$.
\item With these new quantities ($r_i$) as profits of items, use the standard Dynamic Programming algorithm, to find the most profitable set $I'$.
\end{enumerate}

The number of $r_i$'s is at most $n/\epsilon$. Thus, the running time of this algorithm is at most $poly(n,n/\epsilon)= poly(n,1/\epsilon)$. We now show that the above algorithm obtains a $(1-\epsilon)$-approximation ratio
\begin{claim}
 $\sum_{ i \in I'} c_i$ is an $(1-\epsilon)$-approximation to OPT.
\end{claim}
\begin{proof}
 Let $O$ be the optimal set. For each item, rounding down of $c_i$ causes a loss in profit of at most $\frac{\epsilon c}{n}$. Hence the total loss due to rounding down is at most $n$ times $\frac{\epsilon c}{n}$. In other words,
\[\sum_{i \in O}c_i - \frac{\epsilon c}{n} \cdot \sum_{i \in O} r_i \leq n \frac{\epsilon c}{n} =c\epsilon\]
Hence, $\frac{\epsilon c}{n}\sum_{i\in O}r_i \geq OPT - c\epsilon$. 
Now, 
\[ \sum_{i \in I'} c_i \geq \frac{\epsilon c}{n} . \sum_{i \in I'} r_i  \geq \frac{\epsilon c}{n} \sum_{i \in O} r_i \geq OPT -\epsilon c  \geq (1- \epsilon) OPT\]
 The first inequality follows from the definition of $r_i$, the second from the fact that $I'$ is an optimal solution with costs $r_i$'s, the third from the above observation and the last from the fact that $OPT \geq c$.
\end{proof}

\subsection{Type-2 PTAS}
In these kinds of problems we define a set of ``simple'' solutions and find the minimum cost simple solution in polynomial time. Next, we show that an arbitrary solution may be modified to a simple solution without greatly affecting the cost.\\

\noindent \textbf{Euclidean TSP} A Euclidean TSP is a TSP instance where the points are in $\R^2$ and the distances are the corresponding Euclidean distances.

A trivial solution can be found in $n!$. Dynamic Programming finds a solution in $n^22^n$.

We now give a high-level description of a $n^{1/\epsilon}$-time algorithm that achieves a $(1+\epsilon)$-approximation ratio.
Consider the smallest square that contains all $n$ points. Use quad-tree partitioning to recursively partition each square into four subsquares until unit squares are obtained. We consider the number of times the tour path crosses a cell in the quad-tree. We construct the ``simple solution'' to the problem by restricting the tour to cross each dividing line $\leq \frac{6}{\epsilon}$ times. We can then discretize the lines at these crossing points. Each square has $\leq \frac{24}{\epsilon}$ number of crossing points. A tour may use each of these crossing points either 0, 1 or 2 times. So for the entire quadtree there are $\leq 3^{\frac{24}{\epsilon}}=\exp(\frac{24}{\epsilon})$ number of possibilities. For details see Arora's 2003 survey~\cite{Arora2003}.

In the next section, we will see examples of approximation algorithms which use linear programming and semi-definite programming.

\section{Approximation Algorithms for MAXCUT}

The MAX-CUT problem is as follows: Given a graph $G(V,E)$ with $\abs{V}=n$, find $\max_{S \subset V} \lvert E(S,\overline{S}) \rvert$.

The notation $E(S, \overline{S}) $ refers to the set of all edges $(i,j)$ such that vertex $i \in S$ and  vertex $ j \in \overline{S} $.
\subsection{Integer Program Version}
Define variable $x_i$, such that $x_i =0$, if vertex $i \in S$  and $x_i = 1$, if $i \in \overline{S}$. We have the following integer program:\\
\begin{eqnarray*}
\text{Maximize} \sum_{(ij) \in E} e_{ij}&&\\
\text{subject to}\\
e_{ij}&\leq&\min \lbrace x_i+x_j, 2-(x_i+x_j) \rbrace,~\forall (i,j) \in E\\
 x_1,\dots,x_n &\in& \{0,1\}
\end{eqnarray*}

Notice that $e_{ij} \ne 0 \iff x_i \ne x_j$.
\subsection{Linear Program Relaxation and Randomized Rounding}

This can be converted to a \textbf{Linear Program} as follows:\\
\begin{eqnarray*}
\text{Maximize} \sum_{(ij) \in E} e_{ij}&&\\
\text{subject to}\\
e_{ij}&\leq&\min \lbrace x_i+x_j, 2-(x_i+x_j) \rbrace,~\forall (i,j) \in E\\
 x_1,\dots,x_n &\in& [0,1]
\end{eqnarray*}

 Every solution to the Integer Program is also a solution to the Linear Program. So the objective function will only rise.
If ${\rm OPT_{LP}}$ is the optimal solution to the LP, then:
\[{\rm OPT_{LP}} \geq \text{MAX-CUT}\]

\subsubsection{Randomized Rounding}
We now round the LP-solution to obtain an integral solution as follows: form a set $S$ by putting $i$ in $S$ with probability $x_i$. The expected number of edges in such a cut, ${\mathbb E}[|E(S,\bar{S})|]$ can be then calculated as follows:
\begin{eqnarray*}
\mathbb{E}[|E(S,\bar{S})|] &=& \sum_{(i,j) \in E} Pr[(i,j) \text{ is in the cut}]\\
&=&\sum_{(i,j) \in E } x_i(1-x_j) +x_j(1-x_i)
\end{eqnarray*}
The above calculates only an expected value of the cut, however if we repeat the above algorithm several times, it can be seen by Markov's inequality that we can get we can get very close to this value. We now show that this expected value is at least half the LP-optimal, which in turns means that it is at least half the MAX-CUT

\begin{claim}
$$ {\mathbb E}[|E(S,\bar{S})|]=\sum_{(ij) \in E } x_i(1-x_j) +x_j(1-x_i) \geq \frac{1}{2} {\rm OPT_{LP}} \geq \frac{1}{2} \text{\rm MAX-CUT} $$
\end{claim}
\begin{proof}
We have \[{\rm OPT_{LP}} = \sum_{(ij) \in E } e_{ij} = \sum_{(i,j) \in E} \min \lbrace (x_i+x_j), 2-(x_i+x_j)\rbrace\]
It can easily be checked that for any $x_i,x_j \in [0,1]$, we have $$x_i(1-x_j) + x_j(1-x_i) \geq \frac12\cdot\min\{(x_i+x_j), 2-(x_i+x_j)\}.$$
Thus, a term by term comparison of the LHS of the inequality with ${\rm OPT_{LP}}$ reveals that ${\mathbb E}[|E(S,\bar{S})|] \geq \frac12 OPT_{LP} \geq \frac12\text{MAX-CUT}$.
\end{proof}
We thus, have a $1/2$-approximation algorithm for MAX-CUT using randomized rounding of the LP-relaxation of the problem. Actually, it is to be noted that the LP-relaxation is pretty stupid, the optimal to the LP is the trivial solution $x_i = 1/2$ for all $i$, which in turn leads to $OPT_{LP}=|E|$. But we do mention this example as it naturally leads to the following more powerful SDP relaxation.

\subsection{Semi Definite Programming (SDP) Based Method}
We will now sketch a 0.878-approximation to MAX-CUT due to Goemans and Williamson~\cite{GoemansW1995}. The main idea is to relax the integer problem defined above using vector valued variables. THE SDP relaxation is as follows:
\begin{eqnarray*}
\text{Maximize} &&\sum_{(i,j) \in E} \frac{(1-\ip{\vec v_i,\vec v_j})}{2}\\
\text{subject to} && \ip{\vec v_i,\vec v_i}=1,~ \forall i 
\end{eqnarray*}
Denote the optimal to the above SDP by $OPT_{SDP}$. We first observe that the SDP is in fact a relaxation of the integral problem. Let $\vec v_0$ be any vector of unit length, i.e., $\ip{\vec v_0,\vec v_0} =1$. Consider the optimal cut $S$ that achieves MAX-CUT. Now define, $$\vec v_i = \bigg\{
\begin{matrix}
\vec v_0&{\rm if}& i \in S \\ -\vec v_0&{\rm if}& i\notin S 
\end{matrix}, \forall i.$$ Consider the quantity $\frac{(1-\ip{\vec v_i,\vec v_j})}{2}$. This is $0$ if the vectors $\vec v_i$ and $\vec v_j$ lie on the same side, and equals $1$ if they lie on opposite sides. Thus, $OPT_{SDP} \geq \text{MAX-CUT}$. 

How do we round the SDP solution to obtain an integral solution. The novel rounding due to Goemans and Williamson is as follows: The SDP solution produces $n$ vectors $\vec v_1,\dots,\vec v_n$. Now pick a random hyperplane passing through the origin of the sphere and partition vectors according to which side tof the hyperplane they lie. Let $(S,\bar{S})$ be the cut obtained by the above rounding scheme. It is easy to see that
\begin{eqnarray*}\E[|E(S,\bar{S})|]&=& \sum_{(i,j) \in E} \Pr[(i,j)\in \text{cut}]\\
&=& \sum_{(i,j)\in E} \Pr[\vec v_i, \vec v_j \text{ lie on opposite sides of the hyperplane}]
\end{eqnarray*}
Let $\theta_{ij}$ be the angle between vectors $\vec v_i$ and $\vec v_j$. Then the probability that they are cut is proportional to $\theta_{ij}$, in fact exactly $\theta_{ij}/\pi$. Thus, 
\[\E[|E(S,\bar{S})|]=\sum_{(ij) \in E}\frac{\theta_{ij}}{\pi} \]
Let us know express $OPT_{SDP}$ in terms of the $\theta_{ij}$'s. Since $\theta_{ij} = cos^{-1}(\ip{\vec v_i,\vec v_j})$, we have 
\[OPT_{SDP}= \sum_{(i,j) \in E}\frac{(1- cos\theta_{ij})}{2}\]
By a ``miracle of nature''(Mathematica?) Goemans and Williamson observed that 
\[\displaystyle \frac{\theta}{\pi} \geq (0.878\ldots)\times \displaystyle \frac{1-cos \theta}{2},~ \forall \theta \in [0, \pi]\]
Hence,
\[\frac{\E[|E(S,\bar{S})|]}{{\rm OPT_{SDP}}} \geq 0.8788.\]
Thus, we have a 0.878-approximation algorithm for MAX-CUT.

\cleardoublepage

\lecture{2}{20 Jul, 2009}{Alexander S. Kulikov}{Dana Moshkovitz}
{The PCP Theorem: An Introduction} 

Complementing the first introduction lecture on approximation algorithms, this lecture will be an introduction to the limits of approximation algorithms. This will in turn naturally lead to the PCP Theorem, a ground-breaking discovery from the early 90's.

\section{Optimization Problems and Gap Problems}

The topic of this lecture is the hardness of approximation. But to talk about hardness of approximation, we first need to talk
about optimization problems. Recall the definition of optimization problems from the earlier lecture. 
%
Let us begin by giving an example of a canonical optimization problem.

\begin{definition}[\mymaxthreesat]
  The \emph{maximum $3$-satisfiability problem} ({\mymaxthreesat}) is: Given a $3$-CNF formula $\phi$ (each clause 
	contains \emph{exactly} three literals) with $m$ clauses, what is the maximum fraction
  of the clauses that can be satisfied simultaneously by any assignment to the variables? 
\end{definition}

We first prove the following important claim.
\begin{claim}\label{claim:maxsat}
	There exists an assignment that satisfies at least $7/8$ fraction of clauses.
\end{claim}
\begin{proof}
  The proof is a classical example of the probabilistic method. Take a random assignment (each variable of
  a given formula is assigned either $0$ or $1$ randomly and independently). Let $Y_i$ be a random
	variable indicating whether the $i$-th clause is satisfied. For any $1\le i \le m$ 
  (where $m$ is the number of clauses),
    \[ \E{Y_i} = 0 \cdot \frac{1}{8} + 1 \cdot \frac{7}{8} \, , \]
  as exactly one of eight possible assignments of Boolean constants to the variables of the $i$-th clause
  falsifies this clause. Here we use the fact that each clause contains exactly three literals.

  Now, let $Y$ be a random variable equal to the number of satisfied clauses: $Y=\sum_{i=1}^{m}Y_i$.
  Then, by linearity of expectation,
    \[ \E{Y}=\E{\sum_{i=1}^{m}Y_i}=\sum_{i=1}^{m}\E{Y_i}=\frac{7m}{8} \, . \]
  Since a random assignment satisfies a fraction $7/8$ of all clauses, there must exist
  an assignment satisfying at least as many clauses.
\end{proof}

The natural question to ask is if we can do better? Can we find an assignment that satisfies more clauses. Let us phrase this question more formally. For this, we first recall the definition of approximation algorithms from the previous lecture.



\begin{definition}
  An algorithm $C$ for a maximization optimization problem is called $\alpha$-approximation (where $0 \le \alpha \le 1$), 
  if for every input $x$, the algorithm $C$ outputs a value which is at least $\alpha$ times the optimal value, i.e.,
  \[ \alpha \cdot \myopt(x) \le C(x) \le \myopt(x) \, . \]
\end{definition}
\lref[Claim]{claim:maxsat} implies immediately that there exists an efficient (i.e., polynomial time) 
$7/8$-approximation algorithm for {\mymaxthreesat}. The natural question is whether there exists an approximation algorithm that attains a better approximation ratio.
The answer is that such an algorithm is not known. The question that we are going to consider in this lecture
is whether we can prove that such an algorithm does not exist. Of course, if we want to prove this, we have to assume
that P$\neq$NP, because otherwise there is an efficient $1$-approximation algorithm.

There is some technical barrier here. We are talking about optimization problems, i.e., problems where 
our goal is to compute something. It is however much more convenient to consider decision problems (or languages), where we have only two possible answers: yes or no. So, we are going to transform an optimization problem
to a decision problem. Namely, we show that hardness of a certain decision problem implies 
hardness of approximation of the corresponding optimization problem.

\begin{definition}
  For a maximization problem $I$ and $A < B \in \mathbb{R}^+$, the corresponding $[A,B]$-gap problem is 
  the following promise decision problem\footnote{A promise problem $\Pi$ is specified by a pair $(\YES,\NO)$ where $\YES,\NO \in \{0,1\}^*$ and $\YES$ and $\NO$ are disjoint sets. Note there is no requirement that $\YES \cup \NO = \{0,1\}^*$. This is the only difference between promise problems and languages.}: 
  \begin{eqnarray*}
  \YES & = & \{x | \myopt(x) \ge B\}\\
  \NO & = & \{x | \myopt(x) < A \}
\end{eqnarray*}
\end{definition}
We now relate the hardness of the maximization problem to the hardness of the gap problem.
\begin{theorem}
  If the $[A,B]$-gap version of a maximization problem is NP-hard, then it is NP-hard 
  to approximate the maximization problem to within a factor $A/B$.
\end{theorem}
\begin{proof}
  Assume, for the sake of contradiction, that there is a polynomial time $A/B$-approximation algorithm $C$ 
  for a maximization problem under consideration. We are going to show that this algorithm can be used 
  in order to solve the gap problem in polynomial time. 

The algorithm for the gap problem is:
  for a given input $x$, if $C(x) \ge A$, return ``yes'', otherwise return ``no''. 

  Indeed, if $x$ is a yes-instance for the gap problem, i.e., $\myopt(x) \ge B$, then 
    \[ C(x) \ge A/B \cdot \myopt(x) \ge A/B \cdot B = A \]
  and we answer ``yes'' correctly. If, on the other hand, $\myopt(x) < A$, then 
    \[ C(x) \le \myopt(x) < A \]
  and we give the correct ``no'' answer.
\end{proof}
Thus, to show hardness of approximation to within a particular factor, it suffices to show hardness of the corresponding gap problem. Hence from now onwards, we focus on gap problems.

\section{Probabilistic Checking of Proofs}

We will now see a surprising alternate description of the hardness of gap problems. The alternate description is in terms of probabilistically checkable proofs, called PCPs for short. 

\subsection{Checking of Proofs}

Let us first recall the classical notion of proof checking. NP is the class of languages that have a deterministic polynomial-time verifier. For every input $x$ in the language, there exists a proof that convinces the verifier that $x$ is in the language. For every input $x$ not in the language, there is no proof that convinces the verifier that $x$ is in the language. 

For example, when the language is 3SAT, the input is a 3CNF formula $\varphi$. A proof for the satisfiability of $\varphi$ is an assignment to the variables that satisfies $\varphi$.

A verifier that checks such a proof may need to go over the entire proof before it can know whether $\varphi$ is satisfiable: the assignment can satisfy all the clauses in $\varphi$, but the last one to be checked.

\subsection{Local Checking of Proofs}

Can we find some other proof for the satisfiability of $\varphi$ that 
can be checked {\em locally}, by querying only a {\em constant} number of symbols from the proof? 

For this to be possible,
we allow the queries to be chosen in a randomized manner (otherwise, effectively the proof is only of constant size, and a language that can be decided using a polynomial-time verifier with access to such a proof can be decided using a polynomial-time algorithm).
The queries should be chosen using at most a logarithmic number of random bits. The logarithmic bound ensures that the verifier can be, in particular, transformed into a standard, deterministic polynomial time, verifier. The deterministic verifier would just perform {\em all} possible checks, instead of one chosen at random. Since the number of random bits is logarithmic, the total number of possible checks is polynomial. The number of queries the deterministic verifier makes to the proof is polynomial as well.

To summarize, we want a verifier that given the input and a proof, tosses a logarithmic number of random coins and uses them to make a constant number of queries to the proof. If the input is in the language, there should exist a proof that the verifier accepts with probability at least $B$. If the input is not in the language, for any proof, the verifier should accept with probability at most $A$. 
The {\em error} probability is the probability that the verifier does not decide correctly, i.e., $1-B+A$. If $B=1$, we say that the verifier has {\em perfect completeness}, i.e., it never errs on inputs in the language.



\subsection{The Connection to The Hardness of Gap Problems}

The $NP$-hardness of approximation of 3SAT is in fact {\em equivalent} to the existence of local verifiers for $NP$:


\subsubsection{Hardness $\Rightarrow$ Local Verifier}

For $[A,1]$-gap-{\mymaxthreesat}, there is a verifier that makes only $3$ queries to the proof, has perfect completeness, and errs with probability at most $A$!

On input formula $\varphi$,
the proof is a satisfying assignment for $\phi$.
The verifier chooses a random clause of $\phi$, reads the assignment to the three variables of the clause, and checks if the clause is satisfied. The verifier uses $\log m$ random bits, where $m$ is the number of clauses.
If $\varphi$ is satisfiable, the verifier accepts the proof with probability 1.
If not, at most $A$ fraction of all clauses of $\phi$ can be satisfied simultaneously, so the verifier accepts with probability at most $A$.

Moreover, the NP-hardness of $[A,1]$-gap-{\mymaxthreesat} yields local verifiers for all $NP$ languages! More precisely,

\begin{claim}
  If $[A,1]$-gap-{\mymaxthreesat} is NP-hard, then every NP language $L$ has a \emph{probabilistically checkable proof} (PCP).
  That is, there is an efficient randomized verifier that uses only logarithmic number of coin tosses and 
  queries $3$ proof symbols, such that
  \begin{itemize}
    \item if $x \in L$, then there exists a proof that is always accepted;
    \item if $x \not \in L$, then for any proof the probability to err and accept is at most $A$.
  \end{itemize}
\end{claim}
Note that the probability of error can be reduced from $A$ to $\epsilon$ by repeating the action of the verifier $k = O(\frac{\log 1/\epsilon}{\log 1/A})$ times, thus making $O(k)$ queries.

\subsubsection{Local Verifier $\Rightarrow$ Hardness}

What about the other direction? Do local verifiers for NP imply the NP-hardness of gap-{\mymaxthreesat}, which would in turn imply inapproximability of {\mymaxthreesat}?

For starters, let us assume that every language in NP has a verifier that makes three bit queries and whose acceptance predicate is the OR of the three variables (or their negations). Assume that for inputs in the language the verfier always accepts a vaild proof, while for inputs not in the language, for any proof, the verifier accepts with probability at most $A$.
From the above correspondence between local verifiers and gap problems, we get that $[A,1]$-gap-{\mymaxthreesat} is NP-hard. 

What if the verifier instead reads a constant number of bits (not $3$) and its acceptance predicate is some other Boolean function on these constant number of bits? We will now use the fact that every Boolean predicate $f \colon \mybits^{O(1)} \rightarrow \mybits$ can be written
as a $3$-CNF formula $\phi$ with clauses (and some additional variables $z$) such that
for any assignment $x$ of Boolean values to variables of $f$, 
$f(x)=1$ iff $\phi(x,z)$ is satisfiable for some $z$. We have thus transformed the $O(1)$-query verfier with an arbitrary Boolean acceptance predicate to a 3-query verifier with acceptance predicate an OR of the 3 variables (or their negations). We thus have.
 \begin{claim}
  If every NP language $L$ has a constant query verifier that uses only logarithmic number of coin tosses and 
  queries $Q$ proof symbols, such that
  \begin{itemize}
    \item if $x \in L$, then there exists a proof that is always accepted;
    \item if $x \not \in L$, then for any proof the probability to err and accept is at most $A$.
  \end{itemize}
then $[A',1]$-gap-{\mymaxthreesat} is NP-hard for some $A'<1$.
\end{claim}

Thus, we have shown that the problem of proving NP-hardness of gap-{\mymaxthreesat} is equivalent to the problem of constructing constant query verifiers for NP. But do such verifiers exist? 

\subsection{The PCP Theorem}

Following a long sequence of work, Arora and Safra and Arora, Lund, Motwani, Sudan and Szegedy in the early 90's constructed local verifiers for NP:

\begin{theorem}[PCP Theorem (\dots,\cite{AroraS1998},\cite{AroraLMSS1998})]
  Every NP language $L$ has a probabilistically checkable proof (PCP).
  More precisely, there is an efficient randomized verifier that uses only logarithmic number of coin tosses and 
  queries $O(1)$ proof symbols, such that
  \begin{itemize}
    \item if $x \in L$, then there exists a proof that is always accepted;
    \item if $x \not \in L$, then for any proof the probability to accept it is at most $1/2$.
  \end{itemize}
\end{theorem}

The proof in \cite{AroraS1998,AroraLMSS1998} is algebraic and uses properties of low-degree polynomials. There is a more recent alternate combinatorial proof for the theorem due to Dinur~\cite{Dinur2007}. We will later in the workshop see some of the elements that go into the construction.

The PCP Theorem shows that it is NP-hard to approximate {\mymaxthreesat} to within 
\emph{some constant factor}. The natural further question is: can we improve this constant to $7/8$ (to match the trivial approximation algorithm from \lref[Claim]{claim:maxsat})? A positive answer to this question would yield a \emph{tight}
$7/8$-hardness for approximation of {\mymaxthreesat}.

\section{Projection Games}

In $1995$, Bellare, Goldreich, and Sudan~\cite{BellareGS1998} introduced a paradigm for proving inapproximability results. 
Following this paradigm, H\r{a}stad~\cite{Hastad2001} established {\em tight} hardness for approximating {\mymaxthreesat}, as well as many other problems. 

The paradigm is based on the hardness of a particular gap problem, called {\em Label-Cover}. 

\begin{definition}
  An instance of a \emph{projection game} (also called {\em label-cover}) is specified by bipartite graph $G=(A,B,E)$,
  two alphabets $\Sigma_A$, $\Sigma_B$, and \emph{projections} $\pi_e \colon \Sigma_A \rightarrow \Sigma_B$
  (for every edge $e \in E$). 

Given assignments $\calA:A\to \Sigma_A$ and $\calB:B \to \Sigma_B$, an edge $e=(a,b) \in E$ is said to be satisfied iff $\pi_e(\calA(a))=\calB(b)$. The \emph{value}
  of this game is 
  \[ \max_{\calA,\calB}\left(\Pr_{e\in E}\left[\text{$e$ is satisfied} \right]\right)\, .\]
  In a \emph{label-cover problem}: given a projection game instance, compute the value of the game.
\end{definition}

\begin{center}
  \begin{tikzpicture}[auto,sloped]
	  \foreach \a in {1.0,1.5,...,5.0}
  	  \draw (1,\a) rectangle (1.5,\a+0.5);
  	\draw (1.25,6.0) node {$A$};
  	\draw (6.25,5.0) node {$B$};
  	\draw (1.25,3.25) node {$?$};
  	\draw (6.25,2.25) node {$?$};
  	\foreach \b in {1.5,2.0,...,4.0}
  	  \draw (6,\b) rectangle (6.5,\b+0.5);
  	\path (1.5,3.25) edge[->] node[] {$\pi_e$} (6.0,2.25);
  	\path (1.5,3.75) edge[->] (6.0,3.75);
  	\path (1.5,4.75) edge[->] (6.0,3.75);
  \end{tikzpicture}
\end{center}

The PCP-theorem could be formulated as a theorem about Label-Cover.

\begin{theorem}[\dots, \cite{AroraS1998}, \cite{AroraLMSS1998}]\label{thm:weakpcp}
  Label-Cover is NP-hard to approximate within \emph{some constant}.
\end{theorem}
\begin{proof}
  The proof is by reduction to Label-Cover. Recall that we have a PCP verifier from the previous formulation. 
  We construct a bipartite graph as follows. For each possible random string 
  of the verifier, we have a vertex on the left-side. Since the verifier uses only 
  a logarithmic number of bits, the number of such strings is polynomial.
  For each proof location, we have a vertex on the right-side. We add an edge between a left vertex (i.e., a random string)
  and a right vertex (i.e., a proof location), if the verifier queries this proof location on this random string.
  Thus, we defined a bipartite graph. We are now going to describe the labels and projections.

  The labels for the left-side vertices are accepting verifier views and for the right-side are proof
  symbols. A projection is just a consistency check. For example, in case of {\mymaxthreesat}, we have satisfying 
  assignments of a clause on the left-side and values of variables on the right-side.
\end{proof}

However, this theorem is not strong enough to get tight hardness of approximation results. For this reason, we call this the weak projection games theorem.
What we actually need is a low error version of this theorem, which improves the hardness in \lref[Theorem]{thm:weakpcp} from ``some constant'' to ``any arbitrarily small constant''.

\begin{theorem}[(strong) Projection Games Theorem (aka Raz's verifier)~\cite{Raz1998}]
  For every constant $\epsilon>0$, there exists a $k=k(\epsilon)$ such that is is NP-hard
  to decide if a given projection game with labels of size at most $k$ (ie., $|\Sigma_A|, |\Sigma_B| \leq k$) has value $1$ or at most $\epsilon$.
\end{theorem}

The PCP construction in the strong projection games theorem is commonly refered to as Raz's verifier as the theorem follows from Raz's {\em parallel repetition theorem} applied to the construction in \lref[Theorem]{thm:weakpcp}~\cite{Raz1998}. Recently, Moshkovitz and Raz~\cite{MoshkovitzR2008b} gave an alternate proof of this theorem that allows the error $\epsilon$ to be sub-constant.

Why does the label size $k$ depend on $\epsilon$ in the above theorem? This is explained by the following claim, which implies that 
$k$ must be at least $1/\epsilon$.

\begin{claim}
  There is an efficient $1/k$-approximation algorithm for projection games on labels  of size $k$ 
  (i.e., $|\Sigma_A|, |\Sigma_B| \le k$).
\end{claim}

We will later in this tutorial see how the projection games theorem implies tight hardness of approximation for \textsc{3Sat}.

We remark that for many other problems, like \textsc{Vertex-Cover} or \textsc{Max-Cut}, we do not know of tight hardness of approximation results based on the projection games theorem. To handle such problems, Khot formulated the {\em unique games conjecture}~\cite{Khot2002}. This conjecture postulates the hardness of {\em unique label cover}, where the projections $\pi_e$ on the edges are permutations. More on that -- later in the tutorial.

\cleardoublepage

\lecture{3}{20 July, 09}{Gwen Spencer }{Matthew Andrews}{Approximation
Algorithms for Network Problems}


Lectures~3 and 4 will be on network/flow problems, the known approximation algorithms and inapproximability results. In particular, this lecture will serve as an introduction to the different types of network/flow problems and a survey of the known results, while the follow-up lecture by Lisa Zhang will deal with some of the techniques that go into proving hardness of approximation of network/flow problems.

\section{Network Flow Problems}
Network/Flow problems are often motivated by industrial applications. We
are given a communication or transportation network and our goal is to
move/route objects/information though these networks.

The basic problem that we shall be considering is defined by a graph $G=(V,E)$ and a set of
(source,destination) pairs of nodes which we'll denote $(s_1,t_1)$,
$(s_2, t_2)$, etc.  We will sometimes call these pairs ``demand
pairs.''  There are many variants of the problem:
\begin{itemize}
\item \textbf{Only Connectivity is required.}  The question is one of feasibility: ``Is it possible to select a
subset of the edge set of $G$ that connects every $(s_i, t_i)$
pair?''
\item \textbf{Capacities must be respected.} Each edge has a
capacity, and each $(s_i, t_i)$ pair has some amount of demand that
must be routed from $s_i$ to $t_i$.  Observe that this problem is
infeasible if there exists a cut in the graph which has less
capacity than the demand which must cross it.  Imagine variations on
this problem in which more capacity can be purchased on an edge at
some cost (that is, the capacities are not \textit{strict}): the
question becomes: ``What is the minimum amount of capacity that must
be purchased to feasibly route all demand pairs?''
\item \textbf{What solutions are ``good'' depends on the objective
function.}  Consider the difference between the objective of trying
to minimize the maximum congestion (where congestion is the total
demand routed along an edge) and the objective of trying to minimize
the total capacity purchased:  it is not hard to find examples where
a good solution with respect to the first objective is a bad
solution with respect to the second objective and vice versa. The
maximum congestion objective is often used to describe delay/quality
of service.
\item \textbf{Splittable vs. unsplittable flow.} In the unsplittable
flow case all demand routed between $s_i$ and $t_i$ must travel on a
single $(s_i, t_i)$ path.  In the splittable flow case each demand
can be split so that it is routed on a set of paths between $s_i$
and $t_i$.

\item \textbf{Directed vs. Undirected.}  Is the graph directed or
undirected?  As a rule of thumb, problems in which the graph is
directed are more difficult.
\end{itemize}
Next we'll consider some specific problems and describe what positive and negative
results exist for each of them:
\subsection{Minimum Cost Steiner Forest}
In this problem we are interested in simple connectivity.  The input
to the problem is a graph with edge costs and a set of $(s_i, t_i)$
pairs.  The goal is to connect each $(s_i, t_i)$ pair via a set of
edges which has the minimum possible total cost (the cost of a set
of edges is just the sum of the costs of all edges in the set).

Notice that any feasible solution to this problem is a set of
trees.

Both positive and negative results exist for this problem:
\begin{itemize}
\item \textbf{Positive:} 2-approximation
(Agrawal-Klein-Ravi~\cite{AgrawalKR1995},
Goemans-Williamson~\cite{GoemansW1995})

\item \textbf{Negative:} APX-hard, there exists $\varepsilon$ such that no
$1+\varepsilon$ approximation algorithm exists for the problem unless
P=NP.
\end{itemize}

\subsection{Congestion Minimization (Fractional)}
The input to this problem is a graph with edge capacities and a set
of $(s_i, t_i)$ pairs.  The goal is to connect all $(s_i, t_i)$
pairs fractionally (that is, for all $i$, to route a total of one
unit of demand from $s_i$ to $t_i$ along some set of paths in the
graph) in a way that minimizes the maximum congestion.  The
congestion on an edge is simply the total demand routed on that edge
divided by the capacity of the edge. The maximum congestion is the
maximum congestion taken over all edges in the graph.

This problem can be solved exactly in polynomial time via a linear
program.  We write the linear program as follows: let $u_e$ denote
the capacity of edge $e$, and have a decision variable $x_{p,i}$
which is the amount of demand $i$ that is routed on path $p$:

\[ \begin{array}{rl}
    \min        & z \\
    \mbox{s.t.} & \sum_p x_{p,i}  = 1   \quad \forall i  \\
                & \sum_i \sum_{p:e\in p}x_{p,i}\leq zu_e  \quad \forall
                e.\\
                & x_{p,i}\geq 0 \quad \forall p, i.
    \end{array}
\]

The first set of constraints says that for each demand pair $i$, one
unit of demand must be routed from $s_i$ to $t_i$.  The second set
of constraints says that for each edge $e$, the sum of all demand
routed on $e$ must be less than $z$ times the capacity of $e$. Since
the objective is to minimize $z$, the optimal LP solution finds the
minimum multiplicative factor $z$ required so that the capacity of
each edge is at least $1/z$ times the total demand routed on that
edge (that is, the optimal $z$ is the minimum possible maximum
congestion).

Though this LP is not of polynomial size (the number of paths may be
exponentially large) it can be solved in polynomial time, using an
equivalent edge-based formulation whose variables $y_{e,i}$ represent
the amount of flow from demand $i$ routed through edge $e$. Hence we
can obtain an exact solution to the problem.

\subsection{Congestion Minimization (Integral)}
Now consider the Congestion Minimization problem when we require the
routing be \textbf{integral} (all demand routed from $s_i$ to $t_i$
must be routed on a single path).  We can no longer solve this
problem using the linear program above.  The following results are
known:

\begin{itemize}
\item \textbf{Positive:} A $O(\log n/\log \log n)$-approximation algorithm where $n$ is the number of
vertices due to Raghavan-Thompson~\cite{RaghavanT1987}. This algorithm is based on the
technique of {\em randomized rounding} which we describe below.

\item \textbf{Negative:}
Andrews-Zhang~\cite{AndrewsZ2007}) show that
there is no $(\log n)^{1-\varepsilon}$-approximation unless NP has
\textit{efficient} algorithms.  More formally our result holds unless
NP$\subseteq$ZPTIME($n^{{polylog}(n)}$), where
ZPTIME($n^{{polylog}(n)}$) is the class of languages that
can be recognized by randomized algorithms that always give the
correct answer and whose expected running time is $n^{{polylog}(n)}=n^{\log^k n}$ for some constant $k$.  The assumption that
NP$\not\subseteq$ZPTIME($n^{{polylog}(n)}$) is not quite
as strong an assumption as $NP\neq P$ but is still widely believed to
be true.
\end{itemize}

Note that the gap between the positive and negative results here is
large. We comment that for the directed version of the problem, a
negative result has been proved that no 
$\Omega(\log n/ \log \log n)$-approximation exists unless NP has efficient
algorithms~\cite{ChuzhoyGKT2007,AndrewsZ2008}.\\

We now describe the Raghavan-Thompson randomized rounding method for
approximating the Integral Congestion Minimization Problem:

\begin{itemize}
\item Notice that the LP for the fractional problem is a linear
relaxation of the IP we would write for the integral case.  Thus,
the optimal solution to the fractional version is a lower bound on
the optimal value of the integral version:  $OPT_{frac} \leq
OPT_{integral}$.
\item Note that $\sum_p x_{p,i} = 1 $ for all $i$.  Treat the $x_{p,i}
$ as a probability distribution: demand $i$ is routed on path $p$ with
probability $x_{p,i}$.  By linearity of expectation, the expected
congestion of the resulting ranodmly rounded solution on each edge is
at most $OPT_{frac}$.\\

In any given rounding though, some edges will have more than their
expected congestion.  It is possible to show that for any {\em fixed} edge $e$,
with large probability ($\geq 1-\frac{1}{2n^2}$) the congestion on
edge $e$ is $O(\log n/\log \log n) OPT_{frac}=O(\log n/\log \log n)
OPT_{integral}$. By a union bound this implies that with probability
at least $\frac{1}{2}$ the maximum congestion of the randomly rounded
solution {\em on any edge} is $O(\log n/\log \log n) OPT_{integral}$. 

For the directed case this gives the best achievable approximation.
Whether something better exists for the undirected case is an open
question.

\end{itemize}

\subsection{Edge Disjoint Paths}

The input is a graph with edge capacities and a set of $(s_i, t_i)$
pairs.  The goal is to connect every $(s_i, t_i)$ pair integrally
using disjoint paths (that is, to find a set of paths, one
connecting each pair, such that the paths for two distinct pairs
share no edges).  The goal is to connect the maximum possible number
of pairs.

The following results are known:
\begin{itemize}

\item \textbf{Positive:} A $O(m^{1/2})$-approximation where $m$ is the number of edges, due to Kleinberg~\cite{Kleinberg1996}.

\item \textbf{Negative:} (undirected) – No $(\log
n)^{1/2-\varepsilon}$ -approximation exists unless NP has efficient algorithms (Andrews-Zhang~\cite{AndrewsZ2006}).

\item \textbf{Negative:}(directed) – No $O(m^{1/2-\varepsilon})$-approximation exists unless P=NP
(Guruswami-Khanna-Rajaraman-Shepherd-Yannakakis~\cite{GuruswamiKRSY2003}).

\end{itemize}
 We'll look at Kleinberg's $m^{1/2}$ approximation algorithm for this
 problem ($m$ is the number of edges). Consider a greedy algorithm as follows:
\begin{enumerate}
\item Find the shortest path that connects two terminals.
\item Remove all the edges on that path from the graph.
\item Repeat until we cannot connect any more terminals. 

\end{enumerate}

\noindent \textbf{Analysis.} At all times we let $G'$ be the subgraph
of $G$ that contains the remaining edges (i.e.\ the edges that have
not been removed in Step 2). There are two cases in our analysis:
either the shortest path $p$ linking two terminals in the remaining
graph $G'$ has length at most $m^{1/2}$, or not:

\begin{itemize}
\item Suppose the shortest path $p$ in $G'$ has length
$\leq m^{1/2}$.  Each edge in $p$ intersects at most one path from
the optimal solution (since the paths in the optimal solution are
disjoint), so $p$ intersects at most $m^{1/2}$ paths from the
optimal solution.  Thus, when the algorithm removes $p$, at most
$m^{1/2}$ paths are removed from the optimal solution.\\

Thus, the algorithm produces at least one path for every $m^{1/2}$
paths in the optimal solution.

\item Suppose the shortest path $p$ in $G'$ has length
strictly greater than $ m^{1/2}$. Since the paths in the optimal
solution are disjoint, and all must have length at least as long as
$p$, the optimal solution has at most $m/((m^{1/2)})= m^{1/2}$ paths
in $G'$.  Thus, to get a $m^{1/2}$ approximation for $G'$ the
algorithm need only produce one path (so the algorithm can just use
$p$).
\end{itemize}

We mention that for this problem we can't hope to do better with a
linear programming relaxation method because the gap between the
optimal IP solution and the optimal LP solution can be $m^{1/2}$.

\subsection{Minimum Cost Network Design}
The input is a graph, a set of $(s_i, t_i)$ pairs and a cost
function $f(c)$ for placing capacity on an edge.  The goal is to
route one unit of demand between each pair in a way that requires
the minimum cost expenditure for capacity.  \\

Commonly considered cost functions include:

\begin{enumerate}
\item Linear: shortest paths are optimal.
\item Constant: this results in the Steiner forest problem.
\item Subadditive: economies of scale and buy-at-bulk problems. This
is a nice way of modelling how aggregating demand onto a core
network is beneficial and arises in many industrial network design
problems. For subadditive cost functions the following results are
known:
\begin{itemize}
\item \textbf{Positive:} a $O(\log n)$-approximation
(Awerbuch-Azar~\cite{AwerbuchA1997}, Bartal~\cite{Bartal1998}, 
Fakcharoenphol-Rao-Talwar~\cite{FakcharoenpholRT2004}).

\item \textbf{Negative:} No $(\log n)^{1/4-\varepsilon}$ approximation exists unless
NP has efficient algorithms (Andrews~\cite{Andrews2004}).
\end{itemize}
\end{enumerate}

\section*{Summary}
Approximation ratios vary widely for different types of network flow
problems:
\begin{itemize}
\item Constant approximation: Steiner forest.
\item $O(\log n)$-approximation: Congestion minimization, Buy-at-Bulk network
design.

\item $m^{1/2}$-approximation: Edge Disjoint paths.\\
\end{itemize}

\section*{Questions}
\textbf{Q: Are these algorithms actually what is used in practice?}\\

\textbf{Answer:} Not exactly. Take the case of the randomized rounding that we covered: in
practice this technique may not give the best congestion due to the
Birthday Paradox. It is quite likely that there is some edge that gets
higher congestion than the average by a logarithmic factor. 
Hence, practical algorithms typically apply heuristics to try and
reduce the congestion. One technique that often works well is to sort
the demands based on the distance between the terminals (from closest
to farthest). We then go through the demands in order and try to
greedily reroute them.

{We remark that industrial networks often cost a huge amount of
money and so tweaking
a solution a little to save even a single percent can generate
meaningful cost savings. In addition, a lot of these real applications are
huge: cutting-edge computing power together with CPLEX are not even close to
being able to solve these problems exactly. Approximation really
is necessary. }

\cleardoublepage

\lecture{4}{20 July, 2009}{David Pritchard}{Lisa Zhang}{Hardness of the Edge-Disjoint Paths Problem}

\section{Overview}
The \emph{edge-disjoint paths} problem (EDP) is the combinatorial optimization problem with inputs
\begin{itemize}
\item a (directed or undirected) graph $G$ with $n$ nodes and $m$ edges
\item a list of $k$ pairs of (not necessarily distinct) nodes of $G$, denoted $(s_i, t_i)_{i=1}^k$
\end{itemize}
and whose output is
\begin{itemize}
\item a subset $X$ of $\{1, \dotsc, k\}$ representing a choice of paths to route
\item $s_i$-$t_i$ paths $\{P_i\}_{i \in X}$ so that the $P_i$ are pairwise edge-disjoint
\end{itemize}
and
\begin{itemize}
\item the objective is to maximize $|X|$.
\end{itemize}
In these notes, the main results are: a simple proof that for any $\epsilon>0$ it is \NP-hard to approximate the \emph{directed} edge-disjoint paths problem to ratio $m^{1/2-\epsilon}$ (\lref[Section]{sec:diredp}); and a more complex proof that for any $\epsilon>0$, if we could approximate the \emph{undirected} edge-disjoint paths problem to ratio $\log^{1/3-\epsilon}n$, then there would be randomized quasi-polynomial time algorithms for \NP\ (\lref[Section]{sec:undiredp}).

\section{Literature}
For directed EDP, there is a simple $\sqrt{m}$-approximation algorithm due to Kleinberg~\cite{Kleinberg1996} (see also Erlebach's survey \cite{Erlebach2006}), which nearly matches the $m^{1/2-\epsilon}$-hardness result we will present (which is due to Guruswami et al.~\cite{GuruswamiKRSY2003}). A $O(n^{2/3}\log^{2/3} n)$ approximation is also known~\cite{VaradarajanV2004}.

For undirected EDP, Kleinberg's simple algorithm ~\cite{Kleinberg1996} still gives a $\sqrt{m}$-approximation, but an improved $\sqrt{n}$-approximation was recently obtained by Chekuri et al.\ \cite{ChekuriKS2006}. The main technical ingredient in the proof we will present is the \emph{high girth argument}, which was used first in 2004 by Andrews~\cite{Andrews2004} and subsequently in a variety of papers~\cite{AndrewsCGKTZ2007,AndrewsZ2006,AndrewsZ2007,AndrewsZ2008,AndrewsZ2009,ChuzhoyK2006,GuruswamiT2006}, some of which have closed the approximability and inapproximability gaps of various problems up to constant factors. Many of these papers deal with \emph{congestion minimization}, where all demands must be routed and the objective is to minimize the maximum load on any edge. Focusing on undirected graphs, the papers most closely related to what we will show are:
\begin{itemize}
\item \cite{Andrews2004}, which introduced the high girth argument and gave a polylog-hardness for buy-at-bulk undirected network design. For this problem, all demands must be routed, and the cost minimized. The types of ``buy-at-bulk" edges used in the hardness construction were fixed-cost edges (which once bought, can be used to any capacity) and linear-cost edges (where you pay proportional to the capacity used). The paper reduced from a type of 2-prover interactive system, similar to PCPs.
\item \cite{AndrewsZ2006}, which gave the $\log^{1/3-\epsilon}n$-hardness proof we will describe in these notes. The paper reduced from Trevisan's inapproximability results~\cite{Trevisan2001} on the \emph{bounded degree independent set} problem. In turn, those resuls rely on advanced PCP technology~\cite{SamorodnitskyT2000}.
\item \cite{AndrewsCGKTZ2007} --- a paper which was the culmination of merging several lines of work --- which resulted in an improved $\log^{1/2-\epsilon}n$-hardness proof for undirected EDP. This paper uses the hardness of \emph{constraint satisfaction problems}, while the preliminary versions use the Raz verifier (parallel repetition) and directly-PCP based methods. This is so far the best inapproximability result known for undirected EDP, although it is very far from the best known approximation ratio of $\sqrt{n}$ \cite{ChekuriKS2006}.
\end{itemize}

In more detail, the table below summarizes some results in the literature (lower bounds assume $\NP \not\subset\mathsf{ZPTIME}(n^{\mathrm{polylog} n})$, and some constant factors are omitted). Stars ($\star$) denote results in which the high girth method is used.\\

\noindent{\footnotesize\begin{tabular}{lll}
{\bf Problem}&{\bf Upper bound}&{\bf Lower Bound}\\
Undirected EDP&$m^{1/2}$ \cite{Kleinberg1996}, $n^{1/2}$ \cite{ChekuriKS2006}&$\star\log^{1/3-\epsilon} m$ \cite{AndrewsZ2006}, $\star \log^{1/2-\epsilon} m$ \cite{AndrewsCGKTZ2007}\\
Directed EDP&$m^{1/2}$ \cite{Kleinberg1996}, $n^{2/3}\log^{2/3}n$ \cite{VaradarajanV2004} & $m^{1/2-\epsilon}$ \cite{GuruswamiKRSY2003}\\
Undirected Congestion Minimization & $\log m/ \log \log m$ \cite{RaghavanT1987} & $\star\log^{1-\epsilon} \log m$ \cite{AndrewsZ2007},
$\star\frac{\log \log m}{\log \log \log m}$ \cite{RaoZunpub} \\
Directed Congestion Minimization & $\log m/ \log \log m$ \cite{RaghavanT1987} & $\star\log^{1-\epsilon} m$ \cite{AndrewsZ2008}, $\star\frac{\log m}{\log \log m}$ \cite{ChuzhoyGKT2007}\\
Undirected Uniform Buy-at-Bulk & $\log m$ \cite{AwerbuchA1997,FakcharoenpholRT2004} & $\star\log^{1/4-\epsilon} m$ \cite{Andrews2004} \\
Undirected Nonuniform Buy-at-Bulk & $\log^5 m$ \cite{ChekuriHKS2006} & $\star\log^{1/2-\epsilon} m$ \cite{Andrews2004} \\
Undirected EDP with Congestion $c$ \\
(with some restrictions on $c$) & $n^{1/c}$ \cite{AzarR2006, BavejaS2000, KolliopoulosS2001} & $\star \log^{(1-\epsilon)/(c+1)} m$ \cite{AndrewsCGKTZ2007}\\
Directed EDP with Congestion $c$ \\
(with some restrictions on $c$) & $n^{1/c}$ \cite{AzarR2006, BavejaS2000, KolliopoulosS2001} & $\star n^{\Omega(1/c)}$ \cite{ChuzhoyGKT2007}
\end{tabular}
}

If the number $k$ of terminal pairs is fixed, the undirected EDP problem is exactly solvable in polynomial-time, using results from the theory of graph minors~\cite{RobertsonS1995}. (As we will see in \lref[Theorem]{theorem:fhw}, the directed case behaves differently.)

\section{Hardness of Directed EDP}\label{sec:diredp}
In this section we prove the following theorem.
\begin{theorem}[\cite{GuruswamiKRSY2003}]\label{theorem:foo}
For any $\epsilon>0$ it is \NP-hard to approximate the directed edge-disjoint paths problem \textsc{(Dir-EDP)} to within ratio $m^{1/2-\epsilon}$.
\end{theorem}
Although it is very common for inapproximability proofs in the literature to reduce one approximation problem to another, this proof has the cute property that it reduces an exact problem to an approximation problem. Phrased differently, the complete proof does not rely on any PCP-like technology. Specifically, our starting point is the following theorem.
\begin{theorem}[\cite{FortuneHW1980}]\label{theorem:fhw}
The following decision problem \textsc{(Dir-2EDP)} is \NP-hard: given a directed graph $G$ and four designated vertices $s, s', t, t'$ in the graph, determine whether there are two edge-disjoint directed paths, one from $s$ to $t$, and another from $s'$ to $t'$.
\end{theorem}
(Note, this immediately shows that it is hard to \NP-hard to approximate \textsc{(Dir-EDP)} to a factor better than 2.)

The key to proving \lref[Theorem]{theorem:foo} is a construction which maps $(G, s, s', t, t')$ to an instance $(H, (s_i, t_i)_{i=1}^k)$ of \textsc{Dir-EDP} where $k$ is a parameter we will tune later.
The construction is illustrated in \lref[Figure]{fig:erlebach}.
\begin{figure}
\begin{center}\parbox[c]{4.5cm}{\includegraphics[width=4cm]{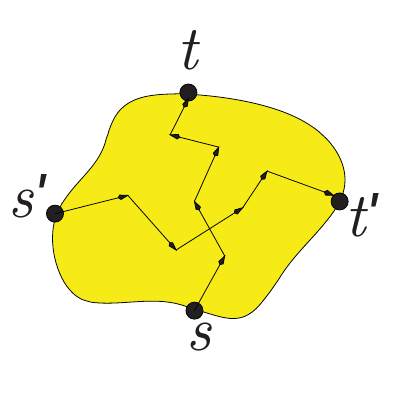}} \parbox[c]{10.5cm}{\includegraphics[width=10cm]{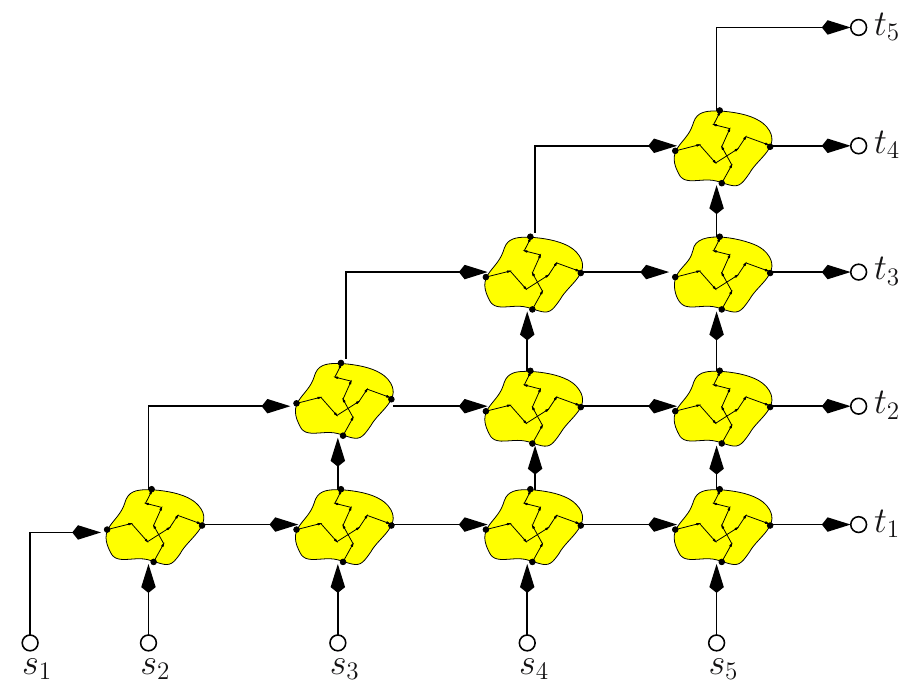}}
\end{center}
\caption[\textsc{Dir-2EDP} and construction of $H$]{Left: the graph $G$ in an instance of \textsc{Dir-2EDP}. Right: The construction of $H$ for $k=5$. Each small yellow object is a copy of $G$, and each other line is a directed edge. This illustration is adapted from Erlebach~\cite{Erlebach2006}.}\label{fig:erlebach}
\end{figure}
The two important properties of this construction are the following:
\begin{itemize}
\item[(a)] If $G$ admits edge-disjoint $s$-$t$ and $s'$-$t'$ paths (say $P$ and $P'$), then $H$ has a solution of value $k$ (i.e.\ all pairs $s_i$-$t_i$ for $1 \leq i \leq k$ can be simultaneously linked by edge-disjoint paths). To see this, we utilize the copies of $P$ and $P'$ within each copy of $G$; then it's easy to see there are mutually disjoint paths $s_i$-$t_i$ paths (the path leaving $s_i$ goes up through $i-1$ copies of $P$, then right through $k-i$ copies of $P'$, to $t_i$).
\item[(b)] If $G$ does not admit edge-disjoint $s$-$t$ and $s'$-$t'$ paths, then there is no solution for $H$ with value greater than 1. To see this, suppose for the sake of contradiction that there is an $s_i$-$t_i$ path $P_i$ and a $s_j$-$t_j$ path $P_j$ in $H$ such that $P_i$, $P_j$ are edge-disjoint. Without loss of generality $i<j$. Then a topological argument shows that there must be some copy of $H$ such that $P_i$ uses the copies of $s'$ and $t'$ and $P_j$ uses the copies of $s$ and $t$. This contradicts our assumption about the \textsc{Dir-2EDP} instance.
\end{itemize}
Facts (a) and (b) show that any algorithm that has approximation ratio better than $k/1$ on the \textsc{Dir-EDP} instance $H$ also solves the \textsc{Dir-2EDP} instance.

Without loss of generality we assume $G$ is (weakly) connected, then the encoding size of the \textsc{Dir-2EDP} instance is proportional to $|E(G)|$ and the encoding size of the \textsc{Dir-EDP} instance is $|E(H)| = O(k^2|E(G)|)$. In order to conclude that it is \NP-hard to approximate the \textsc{Dir-EDP} instance to a factor better than $k$, we need $|E(H)|$ to be polynomial in $|E(G)|$. Thus we may take $k = |E(G)|^\alpha$ for any constant $\alpha$. Going back to the analysis, we get $k = |E(H)|^\frac{\alpha}{2\alpha+1}$; hence by taking $\alpha \to \infty$, we get the desired result (that it is \NP-hard to approximate \textsc{Dir-EDP} to a factor $k = m^{1/2-\epsilon}$).

\section{Hardness of Undirected EDP}\label{sec:undiredp}
In this section we sketch the proof of the following theorem.
\begin{theorem}[\cite{AndrewsZ2006}]\label{theorem:bar}
For any $\epsilon>0$, if we can approximate the directed edge-disjoint paths problem \textsc{(Undir-EDP)} to within ratio $O(\log^{1/3-\epsilon} m)$, then every problem in \NP\ has a probabilistic always-correct algorithm with expected running time $\mathrm{exp}(\mathrm{polylog}(n))$, i.e.\ $\NP \subset \mathsf{ZPTIME}(\mathrm{exp}(\mathrm{polylog}(n)))$.
\end{theorem}
We fully describe the construction of the proof and give intuition for the analysis, but skip some of the detailed parts and precise setting of parameters. The construction creates a simple graph (i.e.\ one with no parallel edges) so the theorem also holds with $\log m$ replaced by $\log n$ since these are the same up to a factor of 2. The proof shows more precisely that $\NP \subset \mathsf{coRPTIME}(\mathrm{exp}(\mathrm{polylog}(n)))$, i.e.\ it gives a quasi-polynomial size, randomized reduction with one-sided error that is right at least (say) 2/3 of the time; then standard arguments\footnote{We insert $\mathsf{Q}$ into standard complexity class names to denote quasi-polynomial time. Suppose $\NP \subset \mathsf{coRQP}$. Then there is a $f(n)$-time algorithm for SAT with $f$ quasi-polynomial. This also implies $\NQP \subset \mathsf{coRQP}$ since every $\NQP$ language with quasi-polynomial running time $g(n)$ is equivalent (by the Cook-Levin construction) to satisfiability of a formula of size $g(n)$, and it can be decided in time $f(g(n))$ which is quasi-polynomial. The definition of $\RQP$ immediately implies $\mathsf{RQP} \subset \mathsf{NQP}$ hence $\mathsf{RQP} \subset \mathsf{coRQP}$. Taking complements we deduce $\mathsf{RQP} = \mathsf{coRQP}$ and it is easy to show that $\mathsf{RQP} \cap \mathsf{coRQP} = \mathsf{ZPQP}$, hence $\NP \subset \mathsf{coRQP} = \mathsf{ZPQP}$. Alternatively, see Lemma 5.8 in \cite{EngebretsenH2003} for a more efficient construction.} allow us to move to $\mathsf{ZPTIME}$. Here is the proof overview.
\begin{itemize}
\item The starting point is the inapproximability of the \emph{independent set} problem (\textsc{IS}) in bounded-degree graphs: find a set of mutually non-adjacent vertices (an \emph{independent set}) with as large cardinality as possible. We denote the degree upper bound by $\Delta$.
\item As usual, our goal is to find a transformation from \textsc{IS} instances to \textsc{Undir-EDP} instances which preserves the ``\NP-hard-to-distinguish gap" in the objective function.
\item We will create a new graph $G$ in the following way. Roughly we ``define a path" $P_i$ for each vertex $v_i$ of the \textsc{IS} instance so that $P_i \cap P_j \neq \varnothing$ iff $v_i, v_j$ are adjacent. (It is easy to see this is possible, with the length of $P_i$ proportional to the degree of $v_i$.) Then we define $G$ to be the union of all $P_i$. (See \lref[Figure]{fig:G}.)
\item To get some intuition for the rest of the proof, define the terminals $s_i$, $t_i$ of the \textsc{Undir-EDP} instance to be the endpoints of $P_i$. Then it is \emph{almost} true that the \textsc{Undir-EDP} instance is isomorphic (in terms of feasible solutions) to the \textsc{IS} instance. The significant problem is that $G$ necessarily also contains $s_i$-$t_i$ paths other than $P_i$, which may be used for routing. (Such paths are obtained by using a combination of edges taken from different $P_j$'s; see \lref[Figure]{fig:G}.)
\item To get around this problem, we transform $G$ into a different graph $H$ defined by two parameters $x, c$. Each intersection of two paths $P_i, P_j$ is replaced by $c$ consecutive intersections; and we replace each $P_i$ with $x$ images $\{P_{i, \alpha}\}_{\alpha=1}^x$.
    The construction of $H$ has a lot of independent randomness, two consequences of which are that (i) when $v_i, v_j$ are adjacent, we can lower-bound the probability that $P_{i, \alpha} \cap P_{j, \beta} \neq \varnothing$ and (ii) $H$ has few short cycles.
    We call each $P_{i, \alpha}$ a \emph{canonical path}; for each canonical path its endpoints define a terminal pair for the new \textsc{Undir-EDP} instance.
\item The optimum of the \textsc{Undir-EDP} solution is at least $x$ times the optimum of the \textsc{IS} instance. To get our hardness-of-approximation result, we also need that when the \textsc{Undir-EDP} optimum is ``large", so is the \textsc{IS} optimum. This is done via a map from \textsc{Undir-EDP} solutions $R$ on $H$ to \textsc{IS} solutions $S$ on $G$. The map is parameterized by a number $a \le x$. We (1) throw away all non-canonical paths in $R$ and (2) take $v_i$ in $S$ iff at least $a$ out of the $x$ paths $\{P_{i, \alpha}\}_{\alpha=1}^x$ are routed by $Y$. The final analysis uses the fact that the canonical paths have length $O(c \Delta)$ while most non-canonical paths are long; the latter depends on the fact that $H$ has few short cycles.
\end{itemize}
\begin{figure}
\begin{center}
\includegraphics{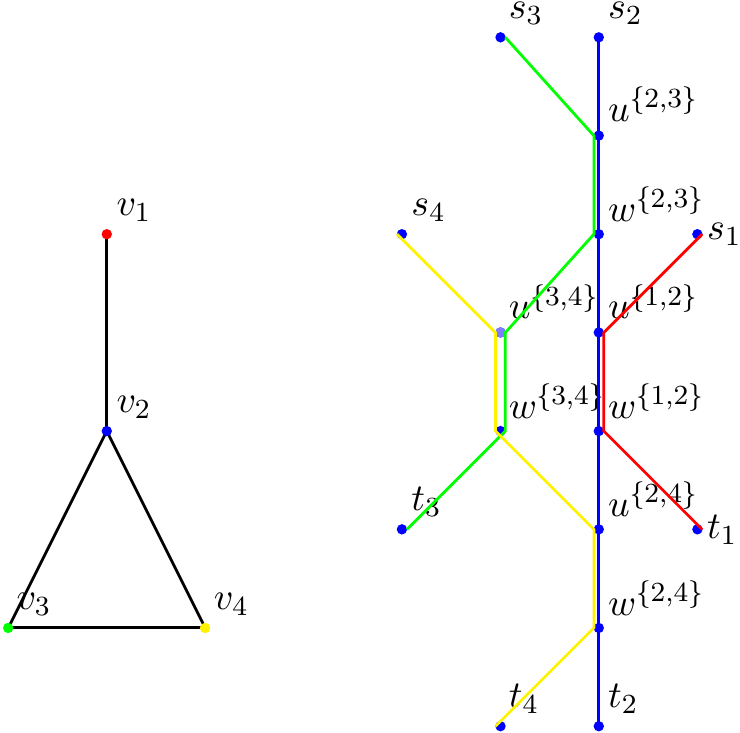}
\end{center}
\caption[\textsc{Undir-EDP} instance]{An illustration of how the independent set instance, $G_0$ (left), yields a \textsc{Undir-EDP} instance on graph $G$ (right). The
colour of $v_i$ in $G_0$ corresponds to the colour of canonical path $P_i$ in $G$. Note that in this example, even though $\{v_1, v_2\}$
is not an independent set, there are edge-disjoint $s_1$-$t_1$ and $s_2$-$t_2$ paths ($P_1$, and a non-canonical $s_2$-$t_2$
path).}\label{fig:G}
\end{figure}

\subsection{Hardness of Bounded-Degree Independent Set}
Trevisan~\cite{Trevisan2001} showed that for any constant $\Delta$, it is \NP-hard to approximate the independent set problem within ratio $\Delta/2^{\sqrt{\log \Delta}}$ on graphs with degrees bounded by $\Delta$. For our purposes, we need a version of this that works for super-constant $\Delta$. By extending the framework in Trevisan~\cite{Trevisan2001}, Andrews and Zhang~\cite{AndrewsZ2006} proved the following:
\begin{theorem}\label{theorem:is}Consider the family of graphs with upper bound $\Delta = \log^b n$ on degree, where $n$ is the number of nodes and $b$ is a constant. If there is a $(\log^{b-\epsilon}n)$-approximation algorithm for \textsc{IS} on these graphs for any $\epsilon>0$, then $\NP \subset \mathsf{coRPTIME}(n^{O(\log \log n)}).$
\end{theorem}

\subsection{\texorpdfstring{The Graphs $G$ and $H$}{The Graphs G and H}}
First we give a formal description of the graph $G$ we sketched earlier. Let $G_0$ denote the \textsc{IS} instance, without loss of generality $G_0$ is connected.
Each edge $e = v_iv_j$ of $G_0$ yields two vertices $u^{\{i, j\}}, w^{\{i, j\}} \in V(G)$ and each vertex $v_i$ of $G_0$ yields two vertices $s_i, t_i$; these are all the vertices of $G$. Let the neighbours of $v_i$ in $G_0$ in any order be $v_p, v_q, \dotsc, v_r$, then we define the path $P_i := (s_i, u^{\{i, p\}}, w^{\{i, p\}}, u^{\{i, q\}}, w^{\{i, q\}}, \dotsc, u^{\{i, r\}}, w^{\{i, r\}}, t_i)$. More precisely, for each adjacent pair of vertices in this list we define an edge of $G$; this constitutes all of the edges of $G$. Every edge of the form $u^{\{i, j\}}w^{\{i, j\}}$ appears in both $P_i$ and $P_j$, while every other edge of $G$ appears in exactly one $P_i$. The number of vertices and edges of $G$ is $O(|E(G_0)|)$ and the number $k$ of terminal pairs is $|V(G_0)|$.

There are two additional ideas needed to define $H$, one whose effect is to randomly replace $P_i$ by $x$ images $\{P_{i, \alpha}\}_{\alpha=1}^x$ and another whose effect is to increase the probability that two paths $P_{i, \alpha}$, $P_{j, \beta}$ intersect when $v_iv_j \in E(G_0)$.

Consider the following probabilistic operation $f_x$ on graphs: replace every vertex $v$ by $x$ ``copies" $\{v_\alpha\}_{\alpha=1}^x$, and replace every edge $vv'$ with a random bipartite matching of $\{v_\alpha\}_{\alpha=1}^x$ to $\{v'_\alpha\}_{\alpha=1}^x$, where these matchings are chosen independently for all input edges $vv'$. Thus $f_x$ multiplies the total number of vertices and edges by $x$. Note that if $vv'$ is an edge of some graph $K$ and $1 \leq \alpha,\beta \leq x$ then the probability that $u_{\alpha}v_{\beta} \in f_x(K)$ is exactly $1/x$; we will later use this fact, as well as the independence of the different random matchings, to show that $f_x(K)$ behaves like a random graph in terms of short cycles. We define the image of terminal pairs under $f_x$ as follows. Define $P_{i, \alpha}$ as the unique path in $f_x(G)$ obtained by starting at $s_{i,\alpha}$ (which denotes ${(s_i)}_\alpha$) and following the images of edges of $P_i$. $P_{i, \alpha}$ does not necessarily end at $t_{i,\alpha}$, rather it ends at $t_{i,\beta} =: t'_{i, \alpha}$ for some uniformly random $\beta$. The terminal pairs of $f_x(G)$ are all pairs $(s_{i, \alpha}, t'_{i, \alpha})$. We call the $P_{i, \alpha}$ \emph{canonical paths}.

At this point it is straightforward to compute the following: if $v_iv_j \in E(G_0)$ and $\alpha, \beta$ are fixed, the probability that the paths $P_{i, \alpha}$, $P_{j, \beta}$ intersect in $f_x(G)$ is exactly $1/x$. More generally, for subsets $A, B$ of $\{1, \dotsc, x\}$, the probability $\Pr[\{P_{i, \alpha}\}_{\alpha \in A} \cup \{P_{j, \beta}\}_{\beta \in B}$ are mutually edge-disjoint in $f_x(G)]$ can be expressed as some function $\delta(x, |A|, |B|)$\footnote{Explicitly, $\delta(x, |A|, |B|) = (x-|A|)!(x-|B|)!/(x-|A|-|B|)!x!$.}. We would like to decrease this (i.e.\ increase the probability some $P_{i, \alpha}$ intersects some $P_{j, \beta}$), and to do so, we consider a graph $G'$ obtained similarly to $G$ except, for $v_iv_j \in E(G_0)$, we force $P_i$ and $P_j$ to intersect $c$ times. We give an informal but precise definition since the formal definition is lengthy. To construct $G'$ from $G$, we perform the following for all edges $v_iv_j \in E(G_0)$: replace the intersection edge $u^{\{i, j\}}w^{\{i, j\}}$ with the gadget pictured in \lref[Figure]{fig:Gprime}, and simultaneously redefine $P_i, P_j$ to follow the indicated paths. Not only will this cause the new $P_i$ and $P_j$ to intersect $c$ times, but the images of these intersections under $f_x$ will be \emph{independent} in the sense that, for all subsets $A, B$ of $\{1, \dotsc, x\}$, the probability $\Pr[\{P_{i, \alpha}\}_{\alpha \in A} \cup \{P_{j, \beta}\}_{\beta \in B}$ are mutually edge-disjoint in $f_x(G')]$ decreases to $\delta^c(x, |A|, |B|)$.

Finally, $H$ is defined to be $f_x(G')$, with canonical paths $P_{i, \alpha}$ and terminal pairs defined as for $f_x(G)$.

\begin{figure}
\begin{center}
\includegraphics{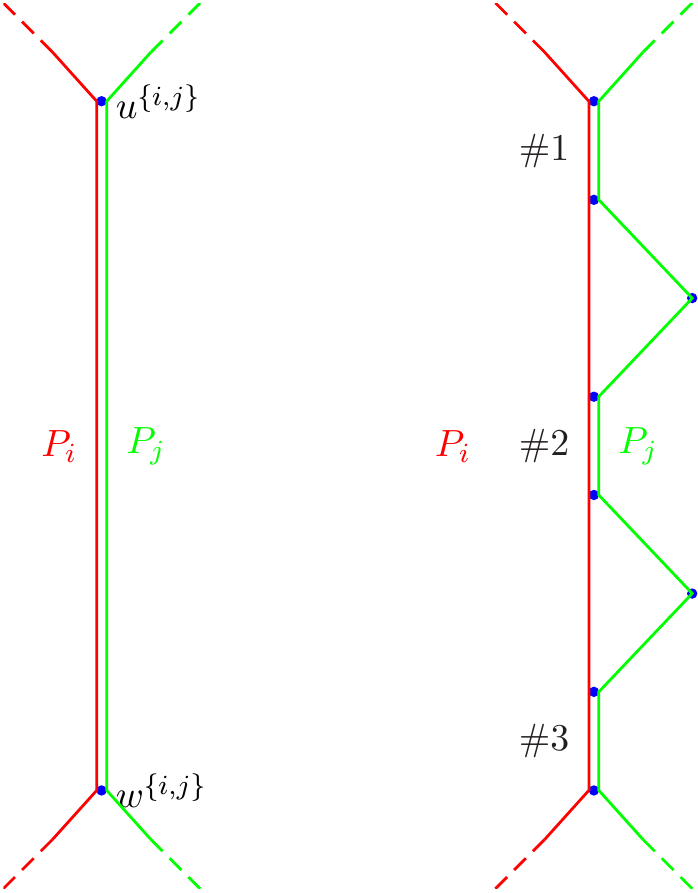}
\end{center}
\caption[Transforming $G$ to $G'$]{Illustrating, for $c=3$, the operation used to transform $G$ (left) into $G'$ (right).}\label{fig:Gprime}
\end{figure}

\subsection{Small Cycles}
The graph $H = f_x(G')$ we defined is not quite a ``random graph" in the usual (Erd\H{o}s-Renyi) sense, but it has enough randomness that it has a typical property of random graphs, namely that the number of small cycles can be bounded. This is done using the first moment method (Markov's inequality), analogous to the 1963 Erd\H{o}s-Sachs theorem (e.g.\ that in the Erd\H{o}s-Renyi model $G(n, d/n)$ the number of cycles of length $g$ is at most $d^g$ in expectation).

If $C$ is any simple cycle in $G'$, then it is not hard to see that the expected number of simple cycles in $H$ that are ``images" of $C$ is $1$. This is a good bound but it is not quite sufficient for our purposes, since cycles may exist in $H$ whose inverse image in $G'$ is not simple. To be precise about getting a bound, for each edge $vv'$ of $G'$, we say that it has $x^2$ corresponding \emph{potential edges} $\{v_\alpha v'_\beta \mid 1 \le \alpha,\beta \le x\}$ in $H = f_x(G')$. (In any realization of $H$, exactly $x$ of these edges are actually present.)
Then it is not hard to see that we get the following: conditioned on the existence of any $\kappa<x$ potential edges in $H$, the probability that any other potential edge is in $H$ is at most $1/(x-\kappa)$. Then the first moment method allows us to show:
\begin{equation}\textrm{if }g < x/2, \E[\#\textrm{ cycles in $H$ of length $\le g$}] \le O(n_0 c\Delta)^{g+1}\label{eq:girthbound}\end{equation}
where we define $n_0 = |V(G_0)|$. Note that this bound is independent of $x$; roughly speaking this is because the factor of $x$ in $|V(H)|$ cancels with the factor $1/x$ in the probability of $H$ containing any given potential edge. Eventually, we will set $g$ to be poly-logarithmic in $n_0$, and the right-hand side of Equation~\eqref{eq:girthbound} will be quasi-polynomial in $n_0$.

It is not hard to argue that $H$ has maximum degree 3, so each vertex has $O(2^g)$ vertices within distance $g$; combining this fact with Equation~\eqref{eq:girthbound} gives us the following form of the ``high-girth argument" that we use in the final proof:
\begin{equation}\Pr[\textrm{$O(n_0c\Delta)^{g+1}$ vertices in $H$ have distance $\le g$ to a cycle of length $\le g$}] \ge 9/10.\label{eq:within}\end{equation}
\subsection{Analysis Sketch}
As mentioned earlier, our reduction uses the following map $R \mapsto S$ from \textsc{Undir-EDP} routings on $H$ to independent sets on $G_0$, parameterized by a number $a$: put $v_i$ into $S$ if at least $a$ out of the $x$ paths $P_{i, \alpha}$ for $i$ are routed by $R$. To be exact, $S$ is only an independent set with some probability, which we would like to make large. By applying simple bounds to $\delta$, for any $A, B \subset \{1, \dotsc, x\}$ with $|A|, |B| \ge a$ and for $i, j$ adjacent in $G_0$, we have that
\begin{equation}
\Pr[\{P_{i, \alpha}\}_{\alpha \in A} \cup \{P_{j, \beta}\}_{\beta \in B} \textrm{ mutually edge-disjoint in }f_x(G')] \le \exp(-ca^2/x).\label{eq:fail}\end{equation}
For any two fixed adjacent vertices $v_i, v_j$ in $G_0$, by a union bound, the probability that \emph{any subsets} $A, B$ with $|A|, |B| \ge a$ exist, such that $A$ and $B$ fail the event in \eqref{eq:fail} is at most $\tbinom{x}{a}\tbinom{x}{a}\exp(-ca^2/x)$. Therefore
using another union bound, \begin{equation}\label{eq:fooo}
\Pr[S\textrm{ independent}] \ge 1-|E(G_0)|\binom{x}{a}\binom{x}{a}\exp(-ca^2/x).\end{equation}
The setting of parameters in the proof is then chosen so that $\Pr[S\textrm{ independent}]$ is at least 9/10. This, along with \lref[Theorem]{theorem:is} and Equation~\eqref{eq:within}, are the three sources of error in the proof.

At a high level the analysis breaks the paths in $R$ into four types,
\begin{itemize}
\item[(a)] a canonical path $P_{i, \alpha}$ so that $\#\{\beta | P_{i, \beta} \in R\} \ge a$.
\item[(b)] a canonical path $P_{i, \alpha}$ so that $\#\{\beta | P_{i, \beta} \in R\} < a$
\item[(c)] a non-canonical $s_{i, \alpha}$-$t'_{i, \alpha}$ path where $s_{i, \alpha}$ has distance $\le g$ to a cycle of length $\le g$
\item[(d)] any other non-canonical $s_{i, \alpha}$-$t'_{i, \alpha}$ path
\end{itemize}
and applies the following analysis (recall $n=|E(G_0)|$):
\begin{itemize}
\item The number of paths of type (a) is at most $|S|x$.
\item The number of paths of type (b) is at most $(n_0-|S|)a$.
\item The number of paths of type (c) is at most $O(n_0c\Delta)^{g+1}$ by \eqref{eq:within}.
\item To upper bound the number of paths of type (d), let $P'$ denote one such path. The union of $P'$ and $P_{i, \alpha}$ contains a simple cycle, but the length of that cycle is at least $g$. The length of $P_{i, \alpha}$ is fixed at $O(c \Delta)$ and hence the length of $P'$ is at least $g - O(c \Delta)$. Since the type-(d) paths are disjoint, there are at most $|E(H)|/(g - O(c \Delta))$ of them.
\end{itemize}

This gives a lower bound on $|S|$ in terms of $|R|$. The proof is then completed by setting the parameters carefully. In detail, using the fact that the greedy algorithm for independent set on $G_0$ always gives a solution of value at least $n_0/(\Delta+1)$, we can show that $|S|$ is an independent set with size at least a constant times $|R|/x$ provided that $g {>} \Delta^2 c, \,\, x {>} \Delta a, \,\, c {>} \frac{x}{a} \ln \frac{x}{a} + \frac{x \ln n_0}{a^2}, \,\, x {>} \Delta^2 c \cdot O(n_0 \Delta c)^g$ hold, where we have omitted some constant factors. (These conditions come from the relative contributions of the different types of paths, as well as the error bounds.) The ratio of inapproximability for \textsc{Undir-EDP} is then roughly $\Delta$ as a function of the input size $m = |E(H)| = O(n_0 c \Delta x)$, and it is not hard to show that $\Delta$ is roughly $\log^{1/3} m$ at maximum. (In \cite{AndrewsZ2006}, a precise setting of parameters is given.)

\cleardoublepage

\lecture{5}{20 July, 2009}{Ashkan Aazami}{Prahladh Harsha}{Proof of the PCP Theorem (Part I)} 

In this lecture and the follow-up lecture tomorow, we will see a sketch of the proof of the PCP theorem. Recall the statement of the PCP theorem from Dana Moshkovitz's lecture earlier today. Dana had mentioned both a weak form (the original PCP Theorem) and a strong form (Raz's verifier or hardness of projective games). We need the strong form as it is the starting point of most tight inapproximability results. ``Standard proofs'' of the strong form proceed as follows: first prove the PCP Theorem~\cite{AroraS1998,AroraLMSS1998} either using the original proof  or the new proof of Dinur~\cite{Dinur2007} and then apply Raz's parallel repetition~\cite{Raz1998} theorem to it to obtain the strong form. However, since the work of Moshkovitz and Raz~\cite{MoshkovitzR2008b}, we can alternatively obtain the strong form directly using the proof techniques in the orginal proof of the PCP Theorem along with the composition technique of Dinur and Harsha~\cite{DinurH2009}. We will follow the latter approach in this tutorial.

\section{Probabilistically Checkable Proofs (PCPs)}

We first introduce the \emph{probabilistically checkable proof} (PCP) 
and some variants of it. 

Our goal is to construct a PCP for some NP-complete problem. We will work with the NP-complete problem \textsc{Circuit-SAT}. Let $C$ be an instance of the \textsc{Circuit-SAT} problem.
A PCP consists of a verifier $V$ that is provided with a proof $\pi$ of acceptance 
of the input instance $C$. The goal of the verifier is to check if the given proof is ``valid''.
Given the input $C$, the verifier $V$ generates a random string $R$ and based on the input instance $C$ and the
random bits of $R$ it generates a list $Q$ of queries from the proof $\pi$.  Next, the verifier $V$ queries the proof $\pi$
at the locations of $Q$ and based on the content of the proof in these locations the verifier either accepts the input $C$
as an acceptable instance or rejects it. The content of $\pi$ at the locations $Q$ is called the \emph{local view} of $\pi$ 
and it is denoted by $\pi_Q$. We denote the \emph{local predicate} that the verifier checks by $\varphi$; the verifier accepts
if $\varphi(\pi_Q)=1$ and it rejects otherwise. The verifier has the following properties:
\begin{description}
\item[Completeness:] If $C$ is satisfiable then there is a proof $\pi$ such that the verifier always  accepts with probability $1$; i.e., 
$$\exists \pi: \prob{}{\varphi(\pi_Q)=1}=1.$$
\item[Soundness:] If $C$ is not satisfiable then for every proof $\pi$ the verifier accepts with probability at most $\delta$ (say $\delta=\frac{1}{3}$);
$$\forall \pi: \prob{}{\varphi(\pi_Q)=1}\leq \frac{1}{3}.$$
\end{description}

\begin{figure}[h]
\begin{center}
\includegraphics[scale=.8]{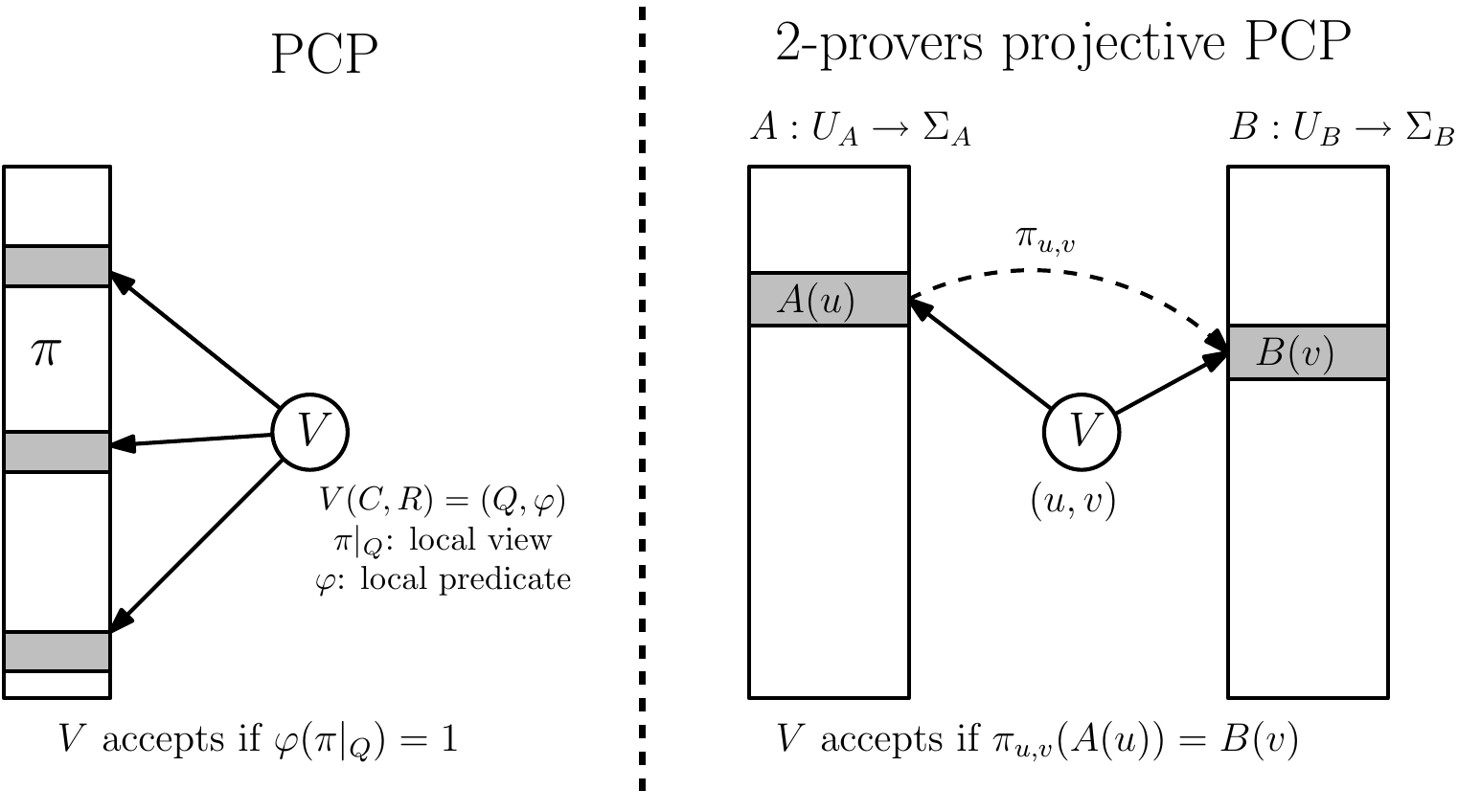}
\end{center}
\caption{PCP and $2$-queries projective PCP}
\end{figure}

The original PCP Theorem now can be stated formally as follows.
\begin{theorem}[PCP Theorem~\cite{AroraS1998,AroraLMSS1998}] \textsc{Circuit-SAT} has a PCP (with the above completeness and soundness properties) that uses $\abs{R}=O(\log{n})$ random bits and queries $\abs{Q}=O(1)$  locations of the proof where $n$ is the size of the input circuit.
\end{theorem}

Note that the length of the proof $\pi$ is polynomial in $n$, the size of the input instance $C$, since the length of the random string is 
$O(\log{n})$ so the verifier can make $2^{O(\log{n})}$ number of queries.

\subsection{Strong Form of the PCP Theorem and Robust PCPs}
Now we introduce a strong form of the PCP theorem, this is also called the $2$-prover \emph{projection game} theorem.
In this type of PCPs, there are two non-communicating provers $A:U_A\rightarrow \Sigma_A$ and $B:U_B\rightarrow \Sigma_B$ 
and a verifier $V.$ 
Given the input instance $C$, the verifier first generates a random string $R$ of length logarithmic in the input size and then using 
the random string, it determines two locations $u$ and $v$ and generates a projection function $\pi_{u,v}:\Sigma_A\rightarrow \Sigma_B$.   The verifier then queries the two provers $U_A$ and $U_B$ on  locations $u$ and $v$ respectively and accepts the provers' answers if they are consistent with the projection function ie., $\pi_{u,v}(A(u))=B(v)$.

We can now state the strong from of the PCP theorem as follows.
\begin{theorem}[Strong form of PCP (aka Raz's verifier, hardness of projection games)~\cite{Raz1998}]\label{thm:razverifier} For any constant $\delta>0$, there exist alphabets $\Sigma_A, \Sigma_B$ such that
the \textsc{Circuit-SAT} has a 2-prover projection game with a verifier $V$ such that
\begin{description}
\item[Completeness:] If $C$ is satsifiable ($C\in \mbox{\textsc{Circuit-SAT}}$) then there exist two provers  $A:U_A\rightarrow\Sigma_A, B:U_B\rightarrow\Sigma_B$ such that $$\prob{}{\pi_{u,v}(A(u))=B(v)}=~1$$ 
\item[Soundness:] If $C$ is not satisiable ($C\notin \mbox{\textsc{Circuit-SAT}}$) then for all  pairs of provers $A:U_A\rightarrow\Sigma_A, B:U_B\rightarrow\Sigma_B$, we have
$$\prob{}{\pi_{u,v}(A(u))=B(v)}\leq~\delta.$$ 
\end{description}
\end{theorem}

Now we introduce the notion of \emph{robust} PCP. These PCPs have a stronger soundness property. In the ordinary PCPs, the 
soundness property says that if the input instance $C$ is not an acceptable input, then the local predicate that the verifier 
checks is not satisfied with high  probability. In the robust PCPs the local view  is far from any satisfying assignment with high 
probability.

First for some notation. Given two codewords $x$ and $y$, the agreement between $x$ and $y$ is defined as $agr(x,y)=\prob{i}{x_i=y_i}$. For a given set $S$ 
of code-words, we define the agreement of $S$ and $x$ by $\agr(x,S)=\max_{y\in S}{\agr(x,y)}$. 
Let us denote the set of all satisfiable assignments to the local predicate $\varphi$ by $SAT(Q)$. 

The robust PCPs have the same completeness property as in the ordinary PCPs, but they have a stronger soundness property. More precisely, the following soundness property of regular PCPs is replaced by the stronger ``robust soundness'' property.
\begin{description}
\item[Soundness:]
$$C\notin \mbox{\textsc{Circuit-SAT}}\Rightarrow \prob{}{\pi_Q\in SAT(\varphi)}\leq \delta$$
\item[Robust Soundness:]
$$C\notin \mbox{\textsc{Circuit-SAT}}\Rightarrow \E[\agr(\pi_Q,SAT(\varphi))]\leq \delta$$
\end{description}

We call PCPs with the robust soundness property, robust PCPs.

\subsection{Equivalence of Robust PCPs and 2-Provers Projection PCPs}
Note that robust PCPs are just regular PCPs with a stronger soundness requirement. We now show that robust PCPs are equivalent to $2$-provers projection PCPs.
Given a robust PCP with the verifier $V$ and the prover $\pi$, we construct a 2-prover projective verifier $V'$ and two provers $A, B$ as follows. The prover $B$ is the same
prover as $\pi$. For each possible random string $R$ and the corresponding queries $Q$ of the verifier $V$, the prover $A$ has the local 
view $\pi_Q$ at the location indexed by $R$; i.e., the prover $A$ has all  possible local views of the prover $\pi$. The verifier $V'$
of the $2$-prover projection PCP is as follows.
\begin{enumerate}
\item Generate a random string $R$ and compute a set $Q$ of queries as in the verifier $V.$
\item Query $1$: Asks the prover $A$ for the entire ``accepting'' local view (i.e., $\pi_Q$).
\item Query $2$: Ask the prover $B$ for a random location within the  local view (i.e., $(\pi_Q)_i$). 
\item Accept if the answer of the prover $B$ is consistent with the answer of the prover $A$. 
\end{enumerate}
It is an easy exercise to check the following two facts. The constructed $2$-provers PCP has the completeness property. Tthe robust soundness 
of the robust PCP translates into the soundness of the $2$-provers PCP. A closer look at this transformation reveals that it is in fact, invertible. This demonstrates a syntactic equivalence between robust PCPs and $2$-prover projection PCPs. Note that in this equivalence, the alphabet size of the left prover $|\Sigma_A|$ translates to query complexity of the robust PCP verifier (to be precise, free-bit complexity of robust PCP verifier).
Given this equivalence, our goal to prove \lref[Theorem]{thm:razverifier} can be equivalently stated as constructing for every constant $\delta$, robust PCPs for \textsc{Circuit-SAT} with robust soundness $\delta$ and query complexity some function of $\delta$ (but independent of $n$).

\section{Locally Checkable Codes}
Our goal is to construct a robust PCP for the \textsc{Circuit-SAT} over a constant size alphabet with constant number of queries for arbitrarily
small soundness error. To achieve this goal, we need to transform a NP-proof (or a certificate for an NP problem) 
to a proof that can be locally checked. To do this, we use locally checkable codes. There are two potential candidates for 
locally checkable codes.
\begin{enumerate}
\item  The first one is the \emph{Directed Product} code; the new proof of the PCP theorem by Dinur  and the proof of
the parallel repetition theorem of Raz are based on this encoding.
\item The second one is the \emph{Reed-Muller} code which is based
on the low-degree polynomials over a finite field $\F$, and the original proof of the PCP theorem is based on this encoding.
\end{enumerate}
We use the Reed-Muller code in construction of the robust PCP. 

A PCP, by definition, is a locally checkable encoding of the NP witness. In the rest of today's lecture, we shall construct locally checkable encodings of two very specific properties, namely ``low-degreeness'' and ``being zero on a sub-cube''. We will define these properties formally shortly, however it is worth noting that neither of these properties is a NP-complete property. In the next lecture, we will show how despite this, we can use the local checkability of these two properties to construct PCPs for all of NP.

\subsection{Reed-Muller Code}

Let $\F$ be a finite field, and let $\mathcal{P}_d^m$ be the set of all $m$-variate polynomials of degree at most $d$ over $\F$. 
The natural way of specifying a function $f\in \mathcal{P}_d^m$ is to list the coefficients of $f$. It is easy to check that 
a $m$-variate polynomial of degree $d$ has $\binom{m+d}{m}$ coefficients. The Reed-Muller encoding  of $f$ is the 
list of the evaluations of $f$ on all $x\in\F^m$; the codeword at the position indexed by $x\in\F^m$ has value $f(x)$. The 
length of this codeword is $\abs{\F^m}$.

This encoding is inefficient but there is an efficient ``local test'' to find out if a given codeword is close to a correct encoding of a 
low degree polynomial. 

\begin{itemize}
\item Question: Given a function $f:\F^m\rightarrow \F$, how does one check if $f$ is a Reed-Muller encoding:
The straightforward way to do this is to interpolate the polynomial and check if it has degree at most $d$.
\item Question: Given a function $f:\F^m\rightarrow \F$, how does one {\em locally} check if $f$ is close to a Reed-Muller encoding. A test for this purpose was first suggested by Rubinfeld and Sudan~\cite{RubinfeldS1996}
This test is based on the fact that a restriction of a low-degree polynomial (over $\F^m$) to a line (or any space with small dimension)
is also a low-degree polynomial.
\end{itemize}
\subsection{Low Degree Test (Line-Point Test)}
Given the evaluations of function $f$ on all points in $\F^m$. Our goal is to check if $f$ is close to a $m$-variate polynomial 
of degree at most $d$; we do this by checking the values of the function $f$ on  a random line. A set $\set{x+ty|t\in \F}$,
for some $x,y\in\F^m$, is called a line in $\F^m$.

\noindent {\bf Low Degree Test (LDT):}
\begin{enumerate}
\item Pick a random line $\ell$ in $\F^m$; this can be done by picking two random points $x,y\in\F^m.$
\item Query the function $f$ on all points of the line $\ell$. Let $f|_{\ell}$ denote the restriction of $f$ on the line $\ell$ 
(i.e., $f|_{\ell}(t)=f(x+ty)$).
\item Accept if $f|_{\ell}$ is an univariate low-degree polynomial (i.e., $f|_{\ell}\in \mathcal{P}_d^1$).
\end{enumerate}

Clearly, if $f\in \mathcal{P}_d^m$, then $f|_{\ell}$ is an univariate polynomial of degree at most $d$. Hence, we have the perfect
completeness.
\begin{description}
\item[Completeness:] $f\in\mathcal{P}_d^m \Rightarrow \prob{}{\mbox{LDT accepts}}=1$
\end{description}

Rubinfeld and Sudan~\cite{RubinfeldS1996} proved the following form of soundness for this test.
\begin{description}
\item[Soundness:] $\forall \delta, \exists \delta': \prob{}{\mbox{LDT accepts}}\geq 1-\delta \Rightarrow $ 
$f$ is $(1-\delta')$-close to some low-degree polynomial (i.e., $\agr (f,\mathcal{P}_d^m)\geq 1-\delta' $).
\end{description}

We will actually need the following stronger soundness that was proven by Arora and Sudan~\cite{AroraS2003}.
\begin{description}
\item[Stronger Soundness:] $\E[\agr(f|_{\ell},\mathcal{P}_d^1)]\geq \delta \Rightarrow \agr(f,\mathcal{P}_d^m)\geq \delta-m\epsilon$,
where $\epsilon={\rm poly}(m,d,\frac{1}{\abs{\F}})$.
\end{description}
Raz and Safra~\cite{RazS1997} proved an equivalent statement (with better dependence of $|\F|$ on $d$) for the plane-point test as opposed to the line-point test.

\subsection{Zero Sub-Cube Test}
In this section, we introduce another test that is used in the construction of robust PCPs.
Let $f$ be a polynomial over $\F^m$ and let $H$ be a subset of $\F$. We want to test if $f$ is a low degree polynomial
(i.e., $f\in \mathcal{P}_d^m$) and if it is zero on the sub-cube $H^m$ (i.e., $f|_{H^m}\equiv 0$).
Using the low degree test (LDT) we can check if $f\in\mathcal{P}_d^m$, but to test if $f$ is zero on $H^m$ it is not enough to
pick few random points from $H^m$ and test if $f$ is zero on those points.

Before describing the correct test, we
present two results about the polynomials.

\begin{lemma}[Schwartz-Zippel]
Let $f$ be a $m$-variate polynomial of degree $d$ over $\F^m$. If $f$ is not a zero polynomial (i.e., $f\not\equiv 0$), 
then $$\prob{x}{f(x)=0}\leq \frac{d}{\abs{\F}}.$$
\end{lemma}
The above lemma shows that if a low degree polynomial over a sufficiently large field is not zero at every point, then it can only be
zero on small fraction of points.

\begin{proposition}
Let $f$ be a polynomial of degree at most $d$ over $\F^m$. The restriction of $f$ to $H^m$ is a zero polynomial (i.e., 
$f|_{H^m}\equiv 0$)  if and only if there exist polynomials $q_1,\ldots, q_m$ of degree at most $d-\abs{H}$ such that
\begin{equation}
f(x)=\sum_{i=1}^{m}{g_{H}(x_i)q_i(x)},
\label{Eq}
\end{equation} 
where $g_H(x)=\prod_{h\in H}{(x-h)}$ is an univariate polynomial (of degree $\abs{H}$).
\end{proposition}

Now we describe the Zero Sub-cube Test. In the LDT we assumed that the evaluation of $f$ on all points are given in the proof table.
By the above Proposition if the polynomial $f$ of degree at most $d$ is zero on $H^m$ then there are polynomials $q_1,\ldots, q_m$ 
of degree at most $d-\abs{H}$ that satisfy Equation \eqref{Eq}. In the Zero Sub-Cube Test, we require that the proof table also
contains the evaluations of $q_1,\ldots, q_m$ (in addition to the evaluation of $f$) on all points in $\F^m$.

\noindent{\bf Zero Sub-cube Test}:
\begin{enumerate}
\item Choose a random line $\ell$ in $\F^m.$
\item For $f,q_1,\ldots, q_m$ check if $f|_{\ell}, q_1|_{\ell},\ldots, q_m|_{\ell}$ is a low degree polynomial. In more detail,
check if $f|_{\ell}$ has degree at most $d$, and for each $i=1,\ldots, m$ check if $q_i|_{\ell}$ has degree at most $d-\abs{H}$.
\item For each $x\in\ell$, check if $f(x)=\sum_{i=1}^{m}{g_{H}(x_i)q_{i}(x)}$.
\item Accept if each of the above tests passes, and reject otherwise.
\end{enumerate}

Combining the soundness of the low-degree test and the above properties of polynomials, we can prove the following completeness and soundness of the {\bf Zero Sub-cube Test}. Let $\mathcal{Z}_d^m$ denote the set of $m$-variate polynomials $P$ of degree $d$ such that $P|_{H^m} =0$. Also for any line $\ell$, let $\mathrm{acc}(\ell)$ denote the set of accepting local views of the {\bf Zero Sub-cube Test} for the random line $\ell$.

\begin{description}
\item[Completeness:] If $f\in\mathcal{Z}_d^m$, then $\prob{}{\mbox{Zero Sub-cube Test accepts}}=1$ or equivalently $\prob{}{(f|_\ell, q_1|_\ell,\ldots,q_m|_\ell) \in \mathrm{acc}(\ell)}= 1$.
\item[Soundness:] $\E[\agr((f|_{\ell},q_1|_\ell,\dots,q_m|_\ell),\mathrm{acc}(\ell))]\geq \delta \Rightarrow \agr(f,\mathcal{Z}_d^m)\geq \delta-m\epsilon-d/|\F|$,
where $\epsilon={\rm poly}(m,d,\frac{1}{\abs{\F}})$.

\end{description} 

\cleardoublepage

\lecture{6}{21 July, 2009}{Geetha Jagannathan  \& Aleksandar Nikolov}{Prahladh Harsha}{Proof
   of the PCP Theorem (Part II)} 

 \section{Recap from Part 1}

 Recall that we want to construct a robust PCP for the NP-Complete
 problem. I.e. for every $n$-sized instance $x$ of the NP-complete problem $L$ we want to construct a proof $\Pi$, which can be checked by a
 verifier using a random string $R$ of length $\log n$ and a
 constant-size query $Q$. The verifier computes a local predicate
 $\phi$ of the local view $\Pi|_Q$ and accepts iff $\phi(\Pi|_Q) =
 1$. We want the construction to satisfy the following properties.
 \begin{description}
 \item[Completeness:] If $x \in L$ then there exists a proof $\Pi$ such that $$\prob{R}{\phi(\Pi|_Q) = 1} = 1.$$
 \item[Soundness:]  If $x\notin L$ then for all proofs $\Pi$, $$\E[agr(\Pi|_Q, sat(\phi))] \leq \delta.$$
 \end{description}

 Recall further that in Part I we constructed a PCP with the parameters
 above not for any NP-complete property but for the specific ``Zero on a Subcube'' property. We say that a function
 $f:\F^m \rightarrow \F$ satisfies the ``Zero on Subcube'' property iff:
 \begin{itemize}
 \item $f$ is a low-degree polynomial $P$.
 \item $P$ vanishes on $H^m$, where $H \subseteq \F$.
 \end{itemize}

 \section{Robust PCP for \circuitsat }

 In this part of the proof we will show how to use the local test for
 Zero on a Subcube to construct a PCP for the \circuitsat~problem.

 \subsection{Problem Definition}
 \begin{figure}[htp]
   \centering
   \includegraphics[scale=0.7]{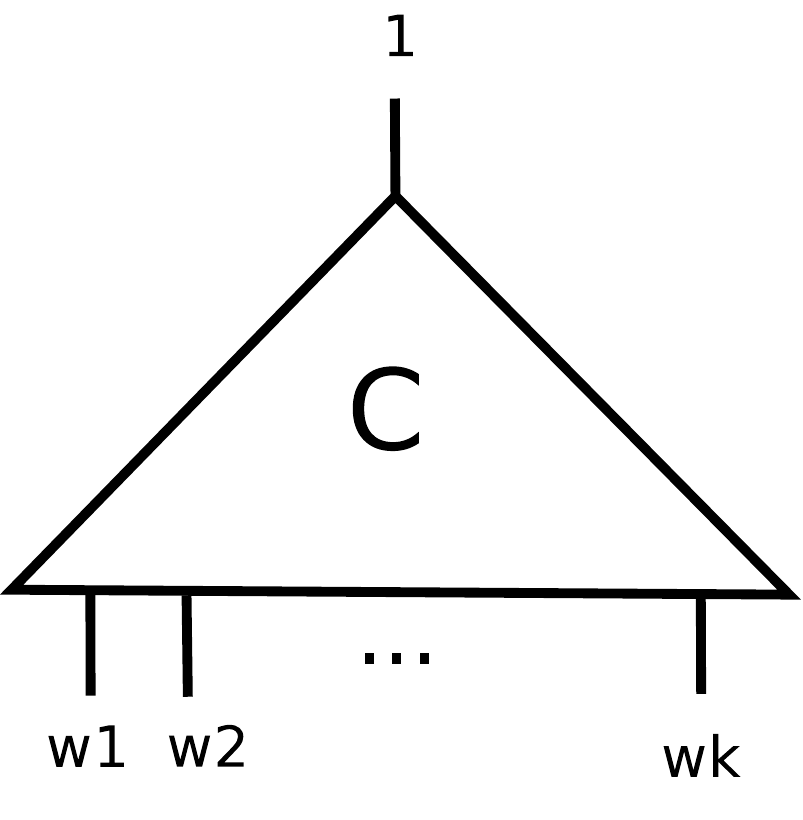}
   \caption{\circuitsat 's input}
   \label{fig:circuit}
 \end{figure}
 \circuitsat~is the following decision problem:
 \begin{itemize}
 \item \textbf{Input}: A circuit $C$ with $n$ gates (\lref[Figure]{fig:circuit}); $k$ of them are the input gates $w_1, \ldots,
   w_k$, and the rest are OR and NOT gates with fan in at most 2 and
   fanout at most 1. Let's associate variables $z_1, \ldots, z_n$ with
   each gate (including the input gates). Variable $z_i$ is the output
   of gate $i$.  The output gate outputs 1.
 \item \textbf{Output}: 1 iff there exists an assignment to $z_1, \ldots,
   z_n$ that respects the gate functionality, and 0 otherwise.
 \end{itemize}

 Note that a proof for this problem is an assignment to $z_1, \ldots,
 z_n$, and verifying the proof amounts to checking that the assignment
 respects gate functionality at each gate. To use our local Zero on a
 Subcube test for \circuitsat~we need to encode the assignment
 and the circuit $C$ algebraically, so that an assignment satisfies $C$
 iff a related function is a low-degree polynomial that vanishes on a
 small subcube. Representing the assignment and the circuit algebraically
 is performed by a process known as arithmetization.

 \subsection{Arithmetization of the Assignment}

 First we will map an assignment to the gate variables $z_1, \ldots,
 z_n$ to a low-degree polynomial over an arbitrary field $\F^m$ so that
 the assignment is encoded by the polynomial. 

 Let $|H^m| = n$ and choose an arbitrary bijection $H^m \leftrightarrow [n]$. The
 assignment maps each gate to either 0 or 1, so it is equivalent to a
 function $A:H^m \rightarrow \{0, 1\}$. We choose $H$ so that $\{0, 1\}
 \subseteq H \subseteq \F$, and we can write $A:H^m \rightarrow \F$.

 The following (easy-to-prove) algebraic fact will be used in the arithmetization of the circuit.

 \begin{fact}[Low-Degree Extension (LDE)]
   For any function $S:H^m \rightarrow \F$, there exists a polynomial
   $\hat{S}:\F^m \rightarrow \F$ such that $\hat{S}|_{H^m} \equiv S$
   and the degree of $\hat{S}$ for each variable is at most
   $|H|$. Therefore the total degree of $\hat{S}$ is at most $m|H|$.
 \end{fact}

 Then $A:H^m \rightarrow \F$ is mapped by the low-degree extension to a
 polynomial $\hat{A}:\F^m \rightarrow \F$ and the degree of $\hat{A}$
 is at most $m|H|$.

 \subsection{Arithmetization of the Circuit}\label{arithmet}

 Our goal is to derive a rule from the circuit $C$ which maps any
 polynomial $\hat{A}:\F^m \rightarrow \F$ to a different low-degree
 polynomial $P_{\hat{A}}:\F^{3m+3} \rightarrow \F$, such that
 $P_{\hat{A}}|_{H^{3m+3}} \equiv 0$ if and only if $\hat{A}$ encodes a
 satisfying assignment. Note that the existence of such a rule is all
 we need to construct a PCP for \circuitsat, as it reduces verifying a
 satisfying assignment to testing the Zero on a Subcube property.

 We will specify the circuit in a slightly different fashion to enable the arithmetization. Consider a function $\bar{C}:
 H^{3m} \times H^3 \rightarrow \{0, 1\}$ that takes three indexes $i_1,
 i_2, i_3 \in [n]=H^m$ and three bits $b_1, b_2, b_3 \in \{0, 1\} \subseteq
 H$ and outputs a bit as follows based on the functionality of the gate whose input variables are $z_{i_1}$ and $z_{i_2}$ and output variable is $z_{i_3}$.
 \begin{equation*}
   \bar{C}(i_1, i_2, i_3, b_1, b_2, b_3) =
   \begin{cases}
     1,& \text{iff the assignment $z_{i_1} = \bar{b}_1$, $z_{i_2} = \bar{b}_2$ $z_{i_3} = \bar{b}_3$, where $i_1$ and $i_2$}\\
         &\text{ are input values to gate $i_3$ and $i_3$ is the output value is an}\\
         & \text{INVALID configuration for the gate $i_3$}\\
 \\
     0 & \text{otherwise}.
   \end{cases}
 \end{equation*}
 \lref[Figure]{fig:gate} illustrates the meaning of the arguments of $\bar{C}$.
 \begin{figure}[htp]
   \centering
   \includegraphics[scale=0.8]{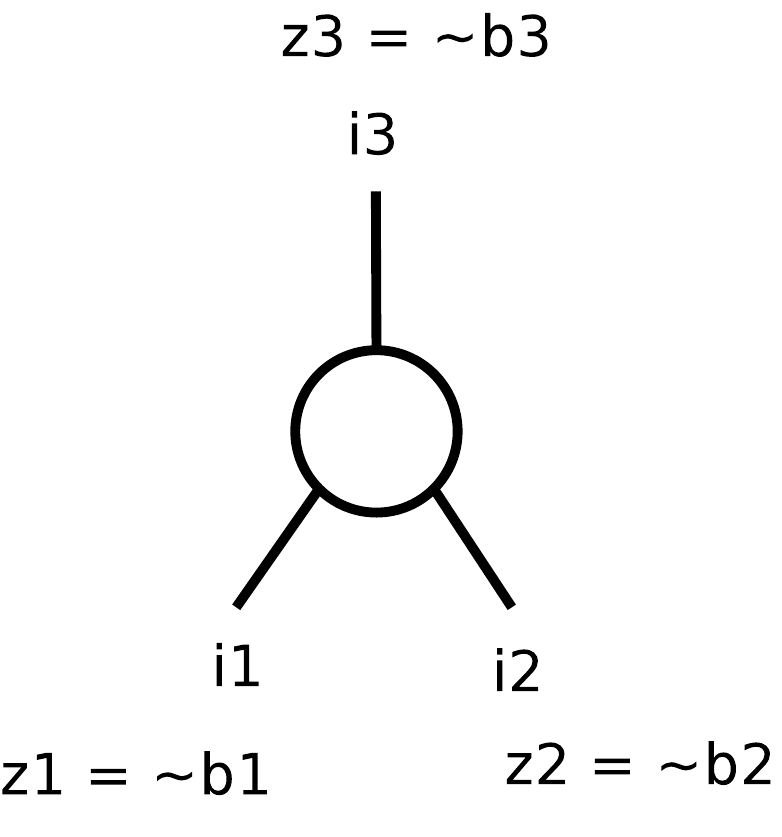}
   \caption{Setting of gate variables for $\bar{C}$}
   \label{fig:gate}
 \end{figure}

 Now once again we can use the LDE to map $\bar{C}:H^{3m+3} \rightarrow
 \F$ to a low-degree polynomial $\hat{C}:\F^{3m+3} \rightarrow \F$.

 We are ready to construct the rule we need. Given any $q: \F^m
 \rightarrow \F$ we define $P_{(q)}:\F^{3m+3} \rightarrow \F$ such that
 \begin{align*}
   P_{(q)}(\underbrace{x_1, \ldots, x_m}_{\mathbf{x}_1}, &\underbrace{x_{m+1}, \ldots, x_{2m}}_{\mathbf{x}_2}, \underbrace{x_{2m+1}, \ldots, x_{3m}}_{\mathbf{x}_3}, z_1, z_2, z_3)\\
   &= \hat{C}(\mathbf{x}_1, \mathbf{x}_2, \mathbf{x}_3, z_1, z_2,
   z_3)(q(\mathbf{x}_1) - z_1)(q(\mathbf{x}_2) - z_2)(q(\mathbf{x}_3) - z_3).
 \end{align*}
 Note that if $q$ is low-degree, $P_{(q)}$ is also low-degree.

 The motivation for defining $P_{(q)}$ in this way will become clear
 when we apply the definition to $\hat{A}$:
 \begin{equation}\label{phat}
   P_{(\hat{A})}(\mathbf{i}_1, \mathbf{i}_2, \mathbf{i}_3, b_1, b_2,
   b_3) = \hat{C}(\mathbf{i}_1, \mathbf{i}_2, \mathbf{i}_3, b_1, b_2,
   b_3)(\hat{A}(\mathbf{i}_1) - b_1)(\hat{A}(\mathbf{i}_2) - b_2)(\hat{A}(\mathbf{i}_3) - b_3).
 \end{equation}
 It is now an easy case-analysis to observe the following.
 \begin{observation}
    $P_{(\hat{A})}|_{H^{3m+3}} \equiv 0 \Leftrightarrow \text{$\hat{A}$
     is a satisfying assignment.}$
  \end{observation}

 \subsection{The PCP Verifier}  
 Given a circuit $C$, the PCP proof consists of the oracles $\hat{A}:\F^{m} \rightarrow \F$ and 
 $P_{\hat{A}}:\F^{3m+3} \rightarrow \F$.

 The PCP verifier needs to make the following checks:
 \begin{itemize}
 \item $\hat{A}$ satisfies the low-degree test
 \item $P_{\hat{A}}$ satisfies the low-degree test
 \item $(P_{\hat{A}}, \hat{A})$ satisfies the rule described in \eqref{phat}.
 \item $P_{\hat{A}}$ is zero on the subcube $H^m$.
 \end{itemize}

 Given the low-degree test and zero-on-subcube test, it is straightforward to design a PCP that performs the above tests. 
 The PCP verifier expects as proofs the oracles $\hat{A}:\F^m \to \F, P_{\hat{A}}:\F^{3m+3}\to \F, q_1:\F^{3m+3}\to \F, \dots,q_{3m+3}:\F^{3m+3}\to \F$. The oracles $q_1,\dots,q_{3m+3}$ are the auxiliary oracles for performing the zero-on-subcube test. The verifier first picks a random line $\ell$ in $\F^{3m+3}$. It reads the value of all the oracles along the line $\ell$. It checks that the restrictions of all the oracles to the line is low-degree. It then checks that for each point $x$ on the line $l$, the ``zero-on-subcube'' test is satisified, namely $$P_{\hat{A}} (x) = \sum_{i=1}^{3m+3} q_i(x) g_H(x_i).$$
 It finally checks for each point on that \eqref{phat} is satisified. This completes the description of the PCP verifier.

 For want of time, we will skip the analysis of the Robust PCP (see \cite{BenSassonGHSV2006} and \cite{Harsha2004} for details). 

 Let us now compute the parameters of the PCP verifier. Here $n = H^m$ is the input length. 
 Let us assume $m$ = $\log(n)/\log\log(n)$. We can choose $|\F| = \poly(m |H|)$. 
 The PCP verifier makes $O(|\F|) =\poly\log n$ queries and the amount of randomness used is $O(m\log(|\F|))=O(\log n)$. The above construction yields a robust PCP of the following form
 \begin{theorem}\label{thm:polylog}\circuitsat\ has a robust PCP that uses $O(\log n)$ randomness, makes $\poly\log n$ queries and has $(1/\poly\log n$) robust soundness parameter.
 \end{theorem}

 Observe that the robust PCP constructed in the above theorem has polylog query complexity and not constant, as we had originally claimed. In the next section, we give a high level outline of how to reduce number of queries (from polylog to constant) using a composition technique
 originally designed by Arora and Safra~\cite{AroraS1998}.

 \subsection{PCP Composition}
 In this section, we describe briefly the PCP composition (due to Arora and Safra~\cite{AroraS1998})
 that helps in reducing the number of queries made by the verifier form $\poly \log(n)$ to constant without affecting the other parameters too much.  Recall that the PCP verifier on input $x$ and an oracle access to a proof $\phi$,
 tosses some random coins and based on the randomness queries some locations. Denote the query locations 
 as the set $I$.  The PCP verifier evaluates a \circuitsat~
 predicate $\phi(\Pi|_I)$ and accepts or rejects based on the outcome of the predicate.  
 The PCP constructed in the previous section achieves $O(\log(n))$ randomness and $O(\poly \log(n))$ 
 number of queries.   We would like to reduce the query complexity from $\poly\log n$ to $O(1)$. How do we do this? How does one check that $\Pi|_I$ satisfies $\phi$ without reading all of $\Pi|_I$ and only reading a constant number of locations in $\Pi|_I$. Arora and Safra suggested that use another PCP to recursively perform this check: $\phi|_I = 1$.
 Let us denote the original PCP verifier as the {\it outer} verifier and the one that checks the predicate $\phi$
 without reading the entire set $I$ as the {\it inner} verifier.  The {\it inner} verifier gets the circuit $\phi$
 and $(\Pi|_I)$ as inputs.  But if it reads all the input bits then the query complexity is not reduced.  
 Instead, the {\it inner} gets as input the circuit and an oracle to the proof of $(\Pi|_I)$ and it now needs to check that the proof in this location satisfies $\phi$. This requires some care as a simple recursion will not do the job. A composition in the context of robust PCPs (or equivalently 2-query projective PCPs) was first shown by Moshkovitz and Raz~\cite{MoshkovitzR2008b}. A more generic and simpler composition paradigm for was then shown by Dinur and Harsha~\cite{DinurH2009}. For want of time, we will skip the details of the composition and conclude on the note that applying the composition theorem of Dinur and Harsha to the robust PCP constructed in \lref[Theorem]{thm:polylog}, one can obtain the constant query PCP with arbitrarily small error as claimed before (see~\cite{DinurH2009} for the details of this construction).



\cleardoublepage

\lecture{7}{21 July, 2009}{Dev Desai}{Subhash
  Khot}{\texorpdfstring{\hastad's $3$-Bit PCP}{Hastad's 3-bit PCP}}


Previously, we saw the proof of the PCP theorem and its connection to proving inapproximability results. The PCPs that we have seen so far have constant number of queries, but over a large alphabet. Now we are interested in designing useful PCPs while keeping the number of query bits low. The purpose of this lecture is to present such a PCP construction, which leads to optimal inapproximability results for various problems such as $\maxsat$ and $\maxlin$.

\section{Introduction}

The PCP theorem and Raz's parallel repetition theorem~\cite{Raz1998} give the $\np$-hardness of a problem called $\labelcover$ (which we will define shortly). This problem is the canonical starting point for reductions that prove inapproximability results. Such reductions can be broadly categorized into two:
\begin{description}
	\item[Direct reductions.] These have been successful in proving inapproximability for network problems, lattice based problems, etc.
	\item[Long code based reductions.] Salient examples of such results can be found in the paper of Bellare, Goldreich, and Sudan~\cite{BellareGS1998} and \hastad~\cite{Hastad2001}.
\end{description}

Long code based reductions have been successful for many important problems like $\maxcut$, albeit with a catch: many of these results depend on conjectures like the Unique Games Conjecture~\cite{Khot2002}, which will be the subject of the next lecture. The focus of this lecture is to prove the powerful result of \hastad's, which can be stated as follows:
\begin{theorem}[\hastad's $3$-bit PCP~\cite{Hastad2001}] \label{th:hastad}
 For every $\eps,\eta > 0$, $\np$ has a PCP verifier that uses $O(\log n)$ random bits, queries exactly $3$ bits to the proof, evaluates a linear predicate on these $3$ bits, and has completeness $1-\eps$ and soundness $1/2+\eta$.
\end{theorem}
We can view the bits in this PCP proof as boolean variables. Then the test of the verifier can be interpreted as a system of linear equations, each equation corresponding to the triplet of bits tested for a given random string. Thus, we immediately obtain the hardness for the $\maxlin$ problem (where an instance of $\maxlin$ is a system of linear equations modulo $2$ with at most $3$ variables per equation, and we are interested in maximizing the fraction of equations that can simultaneously be satisfied).
\begin{corollary}
For every $\eps,\eta > 0$, given an instance of $\maxlin$, it is $\np$-hard to tell if $1-\eps$ fraction of the equations are satisfiable by some assignment or that no assignment satisfies more than $1/2+\eta$ fraction of the equations.
\end{corollary}
In other words, getting an efficient algorithm with approximation guarantee better than $(1/2+\eta)/(1-\eps)$, which is $\approx 1/2$ (since we can choose $\eps$ and $\eta$ to be arbitrarily small), for $\maxlin$ is impossible unless $\p=\np$. This result is tight since a random assignment satisfies half of the equations in expectation. Note that the imperfect completeness in the result is essential if $\p\neq\np$, since Gaussian elimination can be used to efficiently check whether any system of linear equations can be completely satisfied.

We now go on to the proof of \lref[Theorem]{th:hastad}. Here, the concepts of Proof composition, Long codes and Fourier analysis play a pivotal role. We will start with a high-level picture of the PCP and work our way down to the actual $3$-bit test.

\section{Proof Composition}

The standard method to construct a Long code based PCP is by composing an Outer PCP with an Inner PCP. These two concepts are explained below.

\subsection*{The Outer PCP}

The Outer PCP is based on a hard instance of the $\labelcover$ problem.
\begin{definition}
A $\labelcover$ problem $\lc(G(V,W,E),[m],[n],\{\pi_{vw}|(v,w)\in E\})$ consists of:
\begin{enumerate}
 \item A  bipartite graph $G(V,W,E)$ with bipartition $V$, $W$.
 \item Every vertex in $V$ is supposed to get a label from a set $[m]$ and every vertex in $W$ is supposed to get a label from a set $[n]$ ($n\geq m$).
 \item Every edge $(v,w)\in E$ is associated with a projection $\pi_{vw}:[n]\mapsto[m]$.
\end{enumerate}
We say that a labeling $\phi:V\mapsto[m]$, $\phi:W\mapsto[n]$ satisfies an edge $(v,w)$ if $\pi_{vw}(\phi(w))=\phi(v)$. The goal is to find a labeling that maximizes the number of satisfied edges.
\end{definition}
Let us define $\opt(\lc)$ to be the maximum fraction of edges that are satisfied by any labeling. As mentioned earlier, the hardness of $\labelcover$ is obtained by combining the PCP theorem~\cite{FeigeGLSS1996, AroraS1998, AroraLMSS1998} with Raz's parallel repetition theorem~\cite{Raz1998}.
\begin{theorem} \label{th:lcover}
 For every $\delta>0$, there exist $m$ and $n$ such that given a $\labelcover$ instance $\lc(G,[m],[n],\{\pi_{vw}\})$, it is $\np$-hard to tell if $\opt(\lc)=1$ or $\opt(\lc)\leq\delta$.
\end{theorem}
It is useful to think of the above theorem as a PCP that makes $2$ queries over the constant (but large) alphabets of size $m$ and $n$. The verifier just picks an edge at random, queries the labels of the endpoints of this edge and accepts if and only if these labels satisfy the edge. This PCP, which is based on a hard instance of $\labelcover$ forms our Outer PCP.

\subsection*{The Inner PCP}

The Outer PCP verifier expects some labels as the answers to its two queries. We will \textit{compose} this verifier with an Inner PCP verifier, which expects the proof to contain some encoding (in our case, the Long Code) of the labels, rather than the labels themselves. We choose to have an encoding of the labels so that we can check just a few bits of the proof and tell with a reasonable guarantee whether the labeling is valid.

Thus, the Inner verifier expects large bit strings (supposed to be encodings) for each vertex in $G$. Let the edge picked by the Outer verifier be $(v,w)$. Then the Inner verifier checks $1$ bit from $g_v$ (the supposed encoding of the label of $v$) and $2$ bits from $f_w$ (the supposed encoding of the label of $w$). Note that the Inner verifier is trying to simulate the Outer PCP. It needs to check the following two things in one shot:
\begin{description}
 \item[Codeword Test] The strings $f_w$ and $g_v$ are correct encodings of \textit{some} $j\in[n]$ and $i\in[m]$.
 \item[Consistency Test] These $i$ and $j$ satisfy $\pi(j)=i$.
\end{description}
We can therefore convert a PCP which asked $2$ large queries to a PCP which asks $3$ `bit' queries. We now need to show the following two implications:
\begin{enumerate}
 \item (\textit{Completeness}) $\opt(\lc)=1\implies\exists\ \text{Proof}\quad\Pr[\acc]\geq 1-\eps$.
 \item (\textit{Soundness}) $\opt(\lc)\leq\delta\implies\forall\ \text{Proofs}\quad\Pr[\acc]< \frac{1}{2}+\eta$.
\end{enumerate}
The completeness follows by the design of the PCP and should be regarded as a sanity check on the construction. The soundness will be proved by contraposition. We will assume that $\Pr[\acc]\geq 1/2+\eta$ and then decode the labels to satisfy a lot of edges.

\section{The Long Code and its Test}

\hastad's Inner PCP verifier does the codeword test and consistency test in one shot. For clarity, let us first analyze just the codeword test. We will see how to incorporate the consistency test in the next section. For the codeword test, we need to look at the particular encoding that the Inner PCP will use: the Long Code. It is defined below.
\begin{definition}
 The Long Code encoding of $j\in[n]$ is defined to be the truth table of the boolean dictatorship function on the $j$th coordinate, $f:\bits^n\mapsto\bits$ such that $f(x_1,\ldots,x_n)=x_j$.
\end{definition}
Some observations are in order. Note that we are representing bits by $\bits$ and not $\{0,1\}$. This is done just for the sake of clarity, since the calculations done with $\bits$ are less messy. Also note that the Long Code is huge. An element $j\in[n]$ will require $\log n$ bits to represent, but the Long Code of $j$ requires $2^n$ bits, a doubly-exponential blowup! We can get away with this because $n$, the alphabet size, is a constant.

Now for a short aside on Fourier analysis. Recall that for each $S\subseteq[n]$, the Fourier character $\chi_S:\bits^n\mapsto\bits$ is defined as
\[
\chi_S(x)=\prod_{i\in S}x_i.
\]
These characters form an \textit{orthonormal basis} for the set of Boolean functions $f:\bits^n\mapsto\bits$. All such functions can therefore be written in terms of this basis, called the \textit{Fourier expansion}, as
\[
f(x)=\sum_{S\subseteq [n]}\fc{f}(S)\chi_S(x).
\]
where $\fc{f}(S)$ is called the \textit{Fourier coefficient} of set $S$. For Boolean functions, these coefficients satisfy Parseval's identity, namely $\sum_{S\subseteq [n]}\fc{f}(S)^2=1$.

Back to the Long Code. In terms of Fourier expansion, the Long Code of element $j$ is the same as the function $\chi_{\{j\}}$. Thus it is simple to write down, since the only non-zero coefficient is $\fc{f}(\{j\}) = 1$. The Long Code then fits into a general class of functions which have \textit{high Fourier coefficients of low order}. This fact will be useful in the soundness analysis of the test.

On to the codeword test. This will essentially be a linearity test, that is, we will check whether 
\[
f(x+y)=f(x)+f(y). 
\]
Since we are in the $\bits$ domain, this linearity translates to checking whether
\[
f(xy)=f(x)f(y).
\]
We will also introduce a small randomized perturbation in the linearity test. This is done to improve the overall soundness of the test. 
\begin{definition}
An $\eps$-perturbation vector is a string of $\pm 1$ bits, where each bit is independently set to $-1$ with probability $\eps$ and $1$ with probability $1-\eps$.
\end{definition}
The final codeword test is described below. It is a randomized $3$-bit linear test that checks whether the input function is close to a Long Code.

\noindent {\bf Long Code Test:}{\em \mbox{}

Input: Function $f:\bits^n\mapsto\bits$ and error parameter $\eps$.

Test: Pick $x,y\in\bits^n$ at random. Pick an $\eps$-perturbation vector $\mu\in\bits^n$ and let $z=xy\mu$. Accept if and only if
\[
 f(z)=f(x)f(y).
\]}

Let us analyze the completeness and soundness of this test. We have the following theorem.
\begin{theorem} \label{th:lc}
Given a truth table of a function $f:\bits^n\mapsto\bits$ and $\eps >0$, the following are true for the Long Code Test:
\begin{enumerate}
 \item If $f=\chi_{\{j\}}$ for some $j$, then $\Pr[\acc]=1-\eps$.
 \item If $\Pr[\acc]\geq 1/2+\eta$, then $f$ ``resembles'' a Long Code in the following sense: $\exists\ S\subseteq[n]$ such that $|\fc{f}(S)|\geq\eta$ and $|S|\leq \bigO((1/\eps)\log(1/\eta))$, in other words, there exists a large Fourier coefficient of low order.
\end{enumerate}
\end{theorem}
\begin{proof}
To prove the completeness part, assume that $f=\chi_{\{j\}}$ for some $j\in[n]$, that is, it is a Long Code. Then the test will ``Accept'' if and only if
\[
 z_j=x_jy_j\Longleftrightarrow x_jy_j\mu_j=x_jy_j\Longleftrightarrow\mu_j=1
\]
which happens with probability $1-\eps$.

For the soundness analysis, assume that $\Pr[\acc]\geq 1/2+\eta$. We can write this probability in terms of the test as
\[
 \Pr[\acc]=\expect_{x,y,\mu}\left[\frac{1+f(z)f(x)f(y)}{2}\right].
\]
This is a standard PCP trick that is used often in analyzing such tests. Substituting for $\Pr[\acc]$ and using the Fourier expansion of $f$, 
\begin{align*}
 \frac{1}{2}+\eta &\leq \frac{1}{2}+\frac{1}{2}\expect_{x,y,\mu}[f(xy\mu)f(x)f(y)] \\
 2\eta &\leq \expect_{x,y,\mu}\left[ \left( \sum_{S\subseteq[n]}\fc{f}(S)\chi_S(xy\mu) \right) \left( \sum_{T\subseteq[n]}\fc{f}(T)\chi_T(x) \right)\left( \sum_{U\subseteq[n]}\fc{f}(U)\chi_U(y) \right)\right] \\
  &= \sum_{S,T,U}\fc{f}(S)\fc{f}(T)\fc{f}(U)\expect_{x,y,\mu}\left[ \chi_S(xy\mu)\chi_T(x)\chi_U(y) \right] \\
  &= \sum_{S,T,U}\fc{f}(S)\fc{f}(T)\fc{f}(U)\expect_{x,y,\mu}\left[ \chi_S(x)\chi_S(y)\chi_S(\mu)\chi_T(x)\chi_U(y) \right] \\
  &= \sum_{S,T,U}\fc{f}(S)\fc{f}(T)\fc{f}(U)\expect_x\left[ \chi_S(x)\chi_T(x) \right]\expect_y\left[ \chi_S(y)\chi_U(y) \right]\expect_\mu\left[ \chi_S(\mu) \right].
\end{align*}
We can simplify the last expression by using orthonormality of the $\chi$'s to argue that
\begin{equation*}
\expect_x[\chi_S(x)\chi_T(x)] = \expect_x\left[ \prod_{i\in S}x_i\prod_{j\in T}x_j \right] = \prod_{i\in S\Delta T}\expect_x[x_i]=\begin{cases}
1 \quad\text{if }S=T, \\
0 \quad\text{otherwise}
\end{cases}\end{equation*}
where $S\Delta T$ stands for the symmetric difference between sets $S$ and $T$. Thus the terms in the summation will vanish unless $S=T=U$. We then get
\begin{equation}
2\eta\leq\sum_S\fc{f}(S)^3\expect_\mu[\chi_S(\mu)].
\label{eq:fc} 
\end{equation}
As an aside, the Long Code Test without perturbations was analyzed long before \hastad~by Blum, Luby and Rubinfeld~\cite{BlumLR1993}. In that case, we just get
\[
 \sum_S\fc{f}(S)^3\geq 2\eta \implies |\fc{f}_{\max}| \sum_S\fc{f}(S)^2\geq 2\eta \implies \exists\ \text{large }|\fc{f}| \quad(\text{since }\sum\fc{f}^2=1).
\]
Now, with perturbation, we have
\[
 \expect_\mu\left[ \chi_S(\mu) \right]=\expect_\mu\left[ \prod_{i\in S}\mu_i \right]=[1(1-\eps)+(-1)\eps]^{|S|}=(1-2\eps)^{|S|}.
\]
Substituting this in Inequality~\eqref{eq:fc}, we get
\[
 2\eta\leq\sum_S\fc{f}(S)^3(1-2\eps)^{|S|}.
\]
We can again think of the above sum as a convex combination (since $\sum\fc{f}^2=1$). This implies that there exists an $S$ such that $\fc{f}(S)(1-2\eps)^{|S|}\geq 2\eta$. Thus we have
\[
 \Pr[\acc]\geq\frac{1}{2}+\eta\implies\exists\ S:\quad |\fc{f}(S)|\geq 2\eta\quad\text{and}\quad |S|\leq\bigO\left(\frac{1}{\eps}\log\frac{1}{\eta}\right). \qedhere
\]
\end{proof}
The Long Code test can be thought of as an analog of the concept of \textit{gadget} in $\np$ reductions. In particular, \lref[Theorem]{th:lc} is very important and is the crux of the PCP. We will prove similar results in the analysis of later tests. 

\section{Incorporating Consistency}

Let us restate what we want from our $3$-bit test. Recall that the inputs to the Inner PCP verifier are two supposed Long codes, $g$ and $f$, of two vertices $v$ and $w$, and the projection $\pi$ between them. We have to check two things in one shot:
\begin{enumerate}
	\item $g$ and $f$ are Long codes of some $i\in[m]$ and $j\in[n]$.
	\item $\pi(j)=i$.
\end{enumerate}
We are going to do this by reading $1$ bit from $g$ and $2$ bits from $f$ and applying a $3$-bit linear test similar to the Long code test. 

\noindent {\bf Consistency Test:}{\em \mbox{}

Input: Functions $g:\bits^m\mapsto\bits$ and $f:\bits^n\mapsto\bits$, projection $\pi:[n]\mapsto[m]$, and error parameter $\eps$.

Test: Pick $x\in\bits^m$, $y\in\bits^n$ at random. Pick an $\eps$-perturbation vector $\mu:\bits^n$. Let $z=(x\circ\pi)y\mu$. Accept if and only if
\[
f(z)=g(x)f(y).
\]}

In the above test, the vector $(x\circ\pi)$ is defined as $(x\circ\pi)_j=x_{\pi(j)}\quad\forall\ 1\leq j\leq n$. Such a definition is needed because $x$ and $y$ are vectors of different sizes. 

The analysis of the above test is similar to that of the Long Code test. Hence, we will skip a rigorous proof and state only the important details. The completeness is simple to analyze. 

For the soundness analysis, we can imagine a restricted case where $n=m$, $\pi=\id$ (identity permutation), and $f=g$. Then the test is exactly the $3$-bit Long Code test that we saw in the previous section and we get
\[
\Pr[\acc]\geq \frac{1}{2}+\eta\ \xrightarrow{\text{\lref[Theorem]{th:lc}}}\sum_{\substack{S\subseteq[n]\\|S|\leq\bigO(\frac{1}{\eps}\log\frac{1}{\eta})}} \fc{f}(S)^3 \geq \eta\ \xrightarrow{\text{Cauchy-Schwarz}}\sum_{\substack{S\subseteq[n]\\|S|\leq\bigO(\frac{1}{\eps}\log\frac{1}{\eta})}} \fc{f}(S)^4 \geq \eta^2.
\]
Now if we have different $f$ and $g$, one can verify by analysis similar to that in \lref[Theorem]{th:lc} that instead of the above inequality, we would get
\[
\sum_{\substack{S\subseteq[n]\\|S|\leq\bigO(\frac{1}{\eps}\log\frac{1}{\eta})}} \fc{g}(S)^2\fc{f}(S)^2 \geq \eta^2.
\]
Further, now if $\pi\neq\id$ and $n\neq m$, then we would end up with the inequality in the following theorem.
\begin{theorem} \label{th:ct}
The following are true for Consistency Test($g,f,\pi,\eps$):
\begin{enumerate}
	\item If $f=\chi_{\{j\}}$, $g=\chi_{\{i\}}$ and $\pi(j)=i$, then $\Pr[\acc]=1-\eps$.
  \item If $\Pr[\acc]\geq 1/2+\eta$, then $f$ and $g$ are correlated in the following sense:
  \[
  \sum_{\substack{S\subseteq[m],T\subseteq[n]\\|S|,|T|\leq\bigO(1/\eps\log(1/\eta))\\S,T \text{correlated\ by\ }\pi}} \fc{g}(S)^2\fc{f}(T)^2 \geq \eta^2
  \]
  where ``$S,T \text{correlated\ by\ }\pi$'' means that there exist $i\in S$, $j\in T$ such that $\pi(j)=i$.
\end{enumerate}
\end{theorem}
We are now ready to describe \hastad's composed PCP verifier.

\noindent {\bf \hastad's 3-bit PCP:}{\em \mbox{}

Input: Hard instance of $\labelcover$, $\lc(G(V,W,E),[m],[n],\{\pi_{vw}\})$ (given by \lref[Theorem]{th:lcover}) and error parameter $\eps$.

Verifier: Pick edge $(v,w)\in E$ at random. Let $g_v$ and $f_w$ be the supposed Long codes of the two vertices. Run Consistency Test($g_v,f_w,\pi_{vw},\eps$). }

The following theorem gives the completeness and soundness of the verifier.
\begin{theorem} \label{th:final}
Given a hard instance of $\labelcover$, $\lc$, \hastad's PCP guarantees
\begin{enumerate}
	\item (Completeness) If $\opt(\lc)=1$, then there exists a proof for which $\Pr[\acc]\geq 1-\eps$.
	\item (Soundness) If $\Pr[\acc]\geq 1/2+2\eta$, then $\opt(\lc)\geq\eps^2\eta^3/\log^2(1/\eta)$, that is, there is a labeling to $\lc$ that satisfies at least $\eps^2\eta^3/\log^2(1/\eta)$ fraction of edges.
\end{enumerate}
\end{theorem}
\begin{proof}
For the completeness part, we can assume that the proof contains correct encodings on correct labels. Then by \lref[Theorem]{th:ct}, the Verifier accepts with probability $1-\eps$ on every choice of edge. Therefore, the overall acceptance probability also remains $1-\eps$.

Now for the soundness part. If $\Pr[\acc]\geq 1/2+2\eta$, then by an averaging argument, for at least $\eta$ fraction of the edges $(v,w)\in E$, Consistency Test($g_v,f_w,\pi_{vw},\eps$) accepts with probability at least $1/2+\eta$.

Fix any such \textit{good} edge. Then by \lref[Theorem]{th:ct}, we have
\begin{equation} \label{eq:goodsets}
\sum_{\substack{|S|,|T|\leq\bigO(1/\eps\log(1/\eta))\\\exists\ i\in S,j\in T\ :\ \pi(j)=i}}\fc{g}_v(S)^2\fc{f}_w(T)^2\geq\eta^2.
\end{equation}
We will now define labels for vertices $v$ and $w$. First, we pick sets $S\subseteq[m]$ and $T\subseteq[n]$ with probability $\fc{g}_v(S)^2$ and $\fc{f}_w(T)^2$ respectively. Next, we pick labels $i\in S$ and $j\in T$ at random. This is a randomized labeling and by the probabilistic method, the argument goes through. In expectation, at least
\[
\underbrace{\eta}_{\substack{\Pr[\text{Pick}\\\text{good edge}]}}\cdot\underbrace{\eta^2}_{\substack{\Pr[\text{Pick correlated }S,T]\\\text{by \lref[Inequality]{eq:goodsets}}}}\cdot\underbrace{\frac{\eps^2}{\log^2(1/\eta)}}_{\substack{\Pr[\text{Pick }i\in S,j\in T\\ \text{s.t. }\pi(j)=i]}}
\]
fraction of label-cover edges are satisfied.
\end{proof}
Finally, observe that if we set $\delta=c\eps^2\eta^3/\log^2(1/\eta)$, then \lref[Theorem]{th:lcover} and \lref[Theorem]{th:final} together prove \lref[Theorem]{th:hastad}.

\section{Concluding Remarks}

In this lecture, we have seen and analyzed \hastad's powerful $3$-bit PCP. There are a few subtleties in the construction of this PCP that are worth mentioning. The first subtlety, that we have already seen in the Introduction, is that we have to sacrifice perfect completeness if we want our Verifier to have linear predicates.

The second issue is that any linear test can be satisfied if everything is $0$ (or $+1$ in our case). Moreover, in the soundness analysis of \lref[Theorem]{th:final}, the sets $S$ and $T$ that are chosen by the probability distributions $\{\fc{g}_v(S)^2\}$ and $\{\fc{f}_w(T)^2\}$ respectively should not be empty. This problem is taken care by an operation called \textit{folding}. For more information, the reader can look at the references below.

Some good references for more information on \hastad's PCP are Chapter~$22$ in Arora and Barak's textbook~\cite{AroraB2009}, Khot's article on Long code based PCPs~\cite{Khot2005a} and \hastad's original paper~\cite{Hastad2001}.

\cleardoublepage

\lecture{8}{21 July, 2009}{Alantha Newman}{Moses
Charikar}{Semidefinite Programming and Unique Games}

\section{Unique Games}

The topic of Unique Games has generated much interest in the past few
years.  The {\em Unique Games Conjecture} was posed by
Khot~\cite{Khot2002}.  We will discuss the associated optimization problem
and the algorithmic intuition and insight into the conjecture, as well
as the limits of these algorithmic techniques.  Finally, we mention
the amazing consequences implied for many optimization problems if the
problem is really as hard as conjectured.

We now define the Unique Games problem.  The input is a set of
variables $V$ and a set of $k$ labels, $L$, where $k$ is the size of
the {\em domain}.  Our goal is to compute a mapping, $\ell: V
\rightarrow L$, satisfying certain constraints that we now describe.
Let $E$ denote a set of pairs of variables, $\{(u,v)\} \subset V
\times V$.  For each $(u,v) \in E$, there is an associated constraint
represented by $\pi_{uv}$, indicating that $\ell(v)$ should be equal
to $\pi_{uv}(\ell(u))$; we assume that the constraint $\pi_{vu}$ is
the inverse of the constraint $\pi_{vu}$ i.e, $\pi_{uv} =
\pi_{vu}^{-1}$.  Thus, our goal is to compute the aforementioned
mapping, $\ell: V \rightarrow L$, so as to {\em maximize the number of
  satisfied constraints}.

Each constraint, $\pi_{uv}$, can be viewed as a permutation on $L$.
Note that this permutation may be different for each pair $(u,v) \in
E$.  For a pair $(u,v) \in E$, if $v$ is given a particular label from
$L$, say $\ell(v)$, then there is only one label for $u$ that will
satisfy the constraint $\pi_{uv}$.  Specifically, $\ell(u)$ should
equal $\pi_{vu}(\ell(v))$.  Hence, the ``unique'' in Unique Games.
The practice of calling this optimization problem a unique ``game''
stems from the connection of this problem to 2-prover 1-round
games~\cite{FeigeL1992}.  The Unique Games problem is a
special case of Label Cover (discussed in other lectures in the
workshop), in which each constraint forms a bijection from $L$ to $L$.
Having such a bijection turns out to be useful for hardness results.

\section{Examples}

We will refer to $E$ as a set of edges, since we can view an instance
of Unique Games as a graph $G=(V,E)$ in which each edge $(u,v) \in E$
is labeled with a constraint $\pi_{uv}$.  We now give some specific
examples of optimization problems that are special cases of Unique
Games.

\subsection{\texorpdfstring{Linear Equations Mod $p$}{Linear Equations
  Mod p}}

We are given a set of equations in the form $x_i - x_j \equiv c_{ij}~
(\bmod ~p)$.  The goal is to assign each variable in $V = \{x_i\}$ a label
from the set $L = [0, 1, \dots p-1]$ so as to maximize the number of
satisfied equations.  Note that each constraint is a bijection.

\subsection{MAXCUT}

Given an undirected graph $G=(V,E)$, the Max Cut problem is to find a
bipartition of the vertices that maximizes the weight of the edges
with endpoints on opposite sides of the partition.  

We can represent this problem as a special case of Linear Equations
$\bmod ~p$ and therefore as a special case of Unique Games.  For each
edge $(i,j) \in E$, we write the equation $x_i-x_j \equiv 1 ~(\bmod
~2)$.  Note that the domain size is two, since there are two possible
labels, $0$ and $1$.

\section{Satisfiable vs Almost Satisfiable Instances}

If an instance of Unique Games is satisfiable, it is easy to find an
assignment that satisfies all of the constraints.  Can you see why?
Essentially, the uniqueness property says that if you know the correct label
of one variable, then you know the labels of all the neighboring
variables.  So we can just guess all possible labels for a variable;
at some point your guess is correct and this propagates correct labels
to all neighbors, and to their neighbors, and so on.  This is a
generalization of saying that if a graph is bipartite (e.g. all
equations in the Max Cut problem are simultaneously satisfiable), then
such a bipartition can be found efficiently.  So when all constraints
in an instance of Unique Games are satisfiable, this is an ``easy''
problem.

In contrast, the following problem has been conjectured to be
``hard'': If 99$\%$ of the constraints are satisfiable, can we satisfy
1$\%$ of the constraints?  The precise form of the conjecture is known
as the Unique Games Conjecture~\cite{Khot2002}: For all small constants
$\eps, \delta > 0$, given an instance of Unique Games where $1-\eps$
of the constraints are satisfied, it is hard to satisfy a $\delta$
fraction of satisfiable constraints, for some $k > f(\eps, \delta)$,
where $k$ is the size of the domain and $f$ is some function of $\eps$
and $\delta$.

How does $f$ grow as a function of $\eps$ and $\delta$?  We claim that
$f(\eps, \delta) > 1/\delta$.  This is because it can easily
be shown that we can satisfy a $1/k$ fraction of the constraints:
Randomly assigning a label to each variable achieves this guarantee.
Thus, in words, the conjecture is that for a sufficiently large domain
size, it is hard to distinguish between almost satisfiable and close
to unsatisfiable instances.

\subsection{Almost Satisfiable Instances of MAXCUT}

We can also consider the Max Cut problem from the viewpoint of
distinguishing between almost satisfiable and close to unsatisfiable
instances.  However, for this problem, a conjecture as strong as that
stated above for general Unique Games is clearly false.  This is
because we can always satisfy at least half of the equations.  (See
Sanjeev's lecture.)  We now consider the problem of satisfying the
maximum number of constraints given that a $(1-\eps)$ fraction of the
constraints are satisfiable.  We write the standard semidefinite
programming (SDP) relaxation in which each vertex $u$ (with a slight abuse of notation) is represented
by a unit vector, $u$.
\begin{eqnarray*}
& \max &  \sum_{(u,v) \in E} \frac{1-u \cdot v}{2} \\
u\cdot u & = & 1 \quad \forall u \in V\\
u & \in & \R^n \quad \forall u \in V.
\end{eqnarray*}

For a fixed instance of the Max Cut problem, let $OPT$ denote the fraction
of constraints satisfied by an optimal solution, and let $OPT_{SDP}$
denote the value of the objective function of the above SDP on this
instance.  If $OPT \geq (1-\eps)|E|$, then $OPT_{SDP} \geq
(1-\eps)|E|$, since $OPT_{SDP} \geq OPT$.  In Lecture 1 (Sanjeev's
lecture), it was shown that using the random hyperplane rounding of
Goemans-Williamson~\cite{GoemansW1995}, we can obtain
a $.878$-approximation algorithm for this problem.  We will now try to
analyze this algorithm for the case when $OPT$ is large, e.g. at least
$(1-\eps)|E|$.  From a solution to the above SDP, we obtain a collection
of $n$-dimensional unit vectors, where $n = |V|$.  We choose a random
hyperplane, represented by a vector $r \in N(0,1)^n$ (i.e. each
coordinate is chosen according to the normal distribution with mean 0
and variance 1).  Each vector $u \in V$ has either a positive or a
negative dot product with the vector $r$, i.e $r \cdot u > 0$ or $r
\cdot u < 0$.  Let us now analyze what guarantee we can obtain for the
algorithm in terms of $\eps$.

As previously stated, we have the following inequality for the SDP
objective function:
\begin{eqnarray*}
\sum_{(u,v) \in E} \frac{(1-u\cdot v)}{2}  & \geq & (1-\eps)|E|.\
\end{eqnarray*}
Let $\theta_{uv}'$ represent the angle between vectors $u$ and $v$,
i.e. $\arccos(u \cdot v)$.
Let $\theta_{uv}$ denote the angle $(\pi - \theta_{uv}')$.  Then we
can rewrite the objective function of the SDP as:
\begin{eqnarray*}
\sum_{(u,v) \in E} \frac{1 + \cos(\theta_{uv})}{2}.
\end{eqnarray*}
Further rewriting of the objective function results in the following:
\begin{eqnarray*}
\sum_{(u,v) \in E} \frac{1+\cos(\theta_{uv})}{2} & = & \sum_{(u,v) \in
  E} 1 - \frac{1-\cos{(\theta_{uv})}}{2}\\
& = & |E| - \sum_{(u,v) \in E} \frac{1-\cos{(\theta_{uv})}}{2}\\
& = & |E| - \sum_{(u,v) \in E} \sin^2{(\frac{\theta_{uv}}{2})}\\
& \geq & |E| - \eps|E|.
\end{eqnarray*}
We say that vertices $u$ and $v$ are ``cut'' if they fall on opposite
sides of the bipartition after rounding.
\begin{eqnarray*}
\Pr[u \text{ and } v \text{ cut}] ~ = ~ \frac{\theta_{uv}'}{\pi} ~ = ~
1 - \frac{\theta_{uv}}{\pi}.
\end{eqnarray*}
The expected size of $S$---the number of edges cut in a solution---is:
\begin{eqnarray*}
\text{E}[S] & = & \sum_{(u,v) \in E} 1 - \frac{\theta_{uv}}{\pi}\\
& = & |E| - \sum_{(u,v) \in E} \frac{\theta_{uv}}{\pi}.
\end{eqnarray*}
Assume for all $(u,v)\in E$ that $\sin^2{(\frac{\theta_{uv}}{2})} = \eps$.
Then $\sin{(\frac{\theta_{uv}}{2})} = \sqrt{\eps}$.  For small
$\theta$, we have that $\sin{(\theta)} \approx \theta$.  Therefore,
$\theta_{uv}/2 \approx \sqrt{\eps}$.  

Thus, the expected value $\text{E}[S] \geq |E|(1-c \sqrt{\eps})$ for
some constant $c$.  In other words, if we are given a Max Cut instance
with objective value $(1-\eps)|E|$, we can find a solution of size
$(1-c\sqrt{\eps})|E|$.  In other words, an almost satisfiable instance can be
given an almost satisfying assignment, although the assignment has a
weaker guarantee.

\section{General Unique Games}

What happens for a large domain?  How do we write an SDP for this
problem?  Before we had just one vector per vertex.  Now for each
variable, we have $k$ values.  So we have a vector for each variable
and for each value that it can be assigned.  First, we will write a $\{0,1\}$
integer program for Unique Games and then we relax this to obtain an SDP relaxation.  

\subsection{Integer Program for Unique Games}

Recall that $L$ is a set of $k$ labels.  For each variable $u$ and
each label $i \in L$, let $u_i$ be an indicator variable that is 1 if
$u$ is assigned label $i$ and $0$ otherwise.  
Note that the expression in the
objective function is 1 exactly when a constraint $\pi_{uv}$ is satisfied.
\begin{eqnarray*}
& \max & \sum_{(u,v) \in E} \sum_{i \in L} u_i \cdot v_{\pi_{uv}(i)}\\
\sum_{i \in L} u_i  & = & 1 \quad \forall u \in V.
\end{eqnarray*}
Now we move to a vector program.  The objective function stays the
same, but we can add some more equalities and inequalities to the
relaxation that are valid for an integer program.  Below, we write
quadratic constraints since our goal is ultimately to obtain a quadratic program.
\begin{eqnarray*}
\sum_{i \in L} u_i \cdot u_i & = & 1 \quad \forall u \in V, ~i \in L,\\
u_i \cdot u_j & = & 0 \quad \forall u \in V, ~ i\neq j \in L.
\end{eqnarray*}
Additionally, we can also add triangle-inequality constraints on 
triples of vectors, $\{u_i, v_j, w_h\}$ for $u,v,w \in V$ and $i,j,h
\in L$:
\begin{eqnarray}
||u_i - w_h||^2 & \leq & ||u_i - v_j||^2 ~ + ~ ||v_j-w_h||^2,\label{tri_ineq_1}\\
||u_i-v_j||^2 & \geq & ||u_i||^2 - ||v_j||^2.\label{tri_ineq_2}
\end{eqnarray}
These constraints are easy to verify for 0/1 variables, i.e. for
integer solutions.  Note that these constraints are not necessary for
the integer program, but they make the SDP relaxation stronger.

\subsection{Trevisan's Algorithm}

We now look at an algorithm due to
Trevisan~\cite{Trevisan2008}.  Recall that if we know
that every constraint in a given instance is satisfiable, then we
can just propagate the labels and obtain a satisfiable assignment.
The algorithm that we discuss is roughly based on this idea.

How can we use a solution to the SDP relaxation to obtain a solution that
satisfies many constraints?  Suppose that $OPT$ is $|E|$ and consider two
vertices $u$ and $v$ connected by an edge.  In this case, the set of $k$
vectors corresponding to $u$ is the same constellation of $k$ vectors
corresponding to vertex $v$, possibly with a different labeling.  If $OPT$
is $(1-\eps)|E|$, then although these two constellations may no longer be
identical, they should be ``close''.  The correlation of the vectors
corresponds to the distance, i.e. high correlation corresponds to
small distance.  Thus, we want to show that the vector corresponding
to the label of the root vertex $r$ is ``close'' to other vectors,
indicating which labels to assign the other vertices.

\subsubsection{An Algorithm for Simplified Instances}\label{simplified}

Consider the following ``simplified instance''.  Recall that the
constraint graph consists of a vertex for each variable and has an
edge between two variables if there is a constraint between these two
variables.  Suppose the constraint graph has radius $d$: there exists
a vertex $r$ such that every variable is a distance at most $d$ from 
vertex $r$.  The following lemma can be proved using the ideas discussed above.
\begin{lemma}\label{simple_instance} If every edge contributes $1-
  \eps/8(d+1)$ to the SDP objective value, then it is possible to
  efficiently find an assignment satisfying a $(1-\eps)$-fraction of
  the constraints.
\end{lemma}
\noindent
We now give the steps of the rounding algorithm.

\vspace{3mm}

\fbox{\parbox{14cm}{

{\bf Rounding the SDP}
\begin{itemize}

\item[(i)] Find root vertex, $r$, such that every other vertex is
  reachable from $r$ by a path of length at most $d$.

\item[(ii)] Assign label $i$ to $r$ with probability $||r_i||^2$.  

\item[(iii)] For each $u \in V$, assign $u$ label $j$, where $j$ is the
  label that minimizes the quantity $||u_j - r_i||^2$.

\end{itemize}
}}

\vspace{3mm}

As mentioned earlier, the intuition for this label assignment is that
$u_j$ is the vector that is ``closest'' to $r_i$.  We now prove the
following key claim: For each edge $(u,v)$, the probability that
constraint $\pi_{uv}$ is satisfied is at least $1-\eps$.  In
particular, recall that edge $(u,v)$ is mislabeled if
$\ell(v) \neq \pi_{uv}(\ell(u))$.   Thus, we want to show that the
probability that edge $(u,v)$ is mislabeled is at most $\eps$.  

Since $r$ is at most a distance $d$ from all other vertices, a BFS
tree with root $r$ has the property that each $u$ has a path to $r$ on
the tree of distance at most $d$.  Fix a BFS tree and consider the
path from $r$ to $u$: $r = u^0, u^1, u^2, \dots, u^{t-1}, u^t = u$,
where $t \leq d$.  Let $\pi_{u^1}$ denote the permutation $\pi_{u^0,
  u^1}$, and recursively define $\pi_{u^k}$ as the composition of
permutations $(\pi_{u^k, u^{k-1}}) \cdot (\pi_{u^{k-1}})$.  Let $\pi_v
= (\pi_{uv})\cdot (\pi_{u})$.  We now compute the probability that
vertex $u$ is assigned label $\pi_{u(i)}$ and that vertex $v$ is
assigned label $\pi_{v(i)}$, given that $r$ is assigned label $i$.
Note that if both these assignments occur, then edge $(u,v)$ is
satisfied.  (Since edge $(u,v)$ may also be satisfied with another
assignment, we can think of our calculation as possibly being an {\em
  underestimate} on the probability that edge $(u,v)$ is satisfied.)

Let $A(u)$ denote the label assigned to vertex $u$ by the rounding
algorithm.  We will show:
\begin{eqnarray*}
\Pr[A(u) = \pi_{u}(i)] ~ \geq ~ 1-\frac{\eps}{2}\quad {\text{ and }}\quad
\Pr[A(v) = \pi_{v}(i)] ~ \geq ~ 1-\frac{\eps}{2}.
\end{eqnarray*}
This implies that the probability that constraint $\pi_{uv}$ is
satisfied is at least $1-\eps$.  Now we compute the probability that
$A(u) \neq \pi_{u}(i)$.  Suppose that $u_j$ for $j \neq \pi_u(i)$ is
closer to vector $r_i$ than $u_{\pi_u(i)}$ is.  In other words, suppose:
\begin{eqnarray}
||u_j-r_i||^2 ~ \leq ~ ||u_{\pi_u(i)} - r_i||^2.\label{key_ineq}
\end{eqnarray}
Let $B_u$ be the set of labels such that if $r$ is assigned label $i \in
B_u$, then $u$ is not assigned label $\pi_u(i)$.  Note that label $j$
belongs to $B_u$ iff inequality \eqref{key_ineq} holds for $j$.
Thus, the probability
that $u$ is not labeled with $\pi_u(i)$ is exactly:
\begin{eqnarray*}
\Pr[A(u) \neq \pi_u(i)] ~ = ~ \sum_{i \in B_u} ||r_i||^2.
\end{eqnarray*}
One can verify that if there is some label $j$ such that
inequality \eqref{key_ineq} holds, then the quantity $||r_i||^2$ is at most $2 ||r_i - u_{\pi_u(i)}||^2$.  This proof makes use of
inequalities from the SDP,
\eqref{tri_ineq_1} and \eqref{tri_ineq_2}, as well as inequality
\eqref{key_ineq}.  (See \lref[Lemma]{lucas_lemma} from \cite{Trevisan2008}, which we
include in the Appendix.)  Recall that each edge in
the graph (and thus each edge on the path from $r$ to $u$ in the BFS
tree) contributes at most $1-\eps/8(d+1)$ to the objective value.  By
triangle inequality, this implies that $\sum_{i \in L} ||r_i -
u_{\pi_u(i)}||^2 \leq \eps/4$.  Thus, we conclude:
\begin{eqnarray*}
\Pr[A(u) \neq \pi_u(i)] & = & \sum_{i \in B_u} ||r_i||^2 \\
& \leq & 2 \sum_{i \in B_u} ||r_i - u_{\pi_u(i)}||^2 \\
& \leq & 2 \sum_{i \in L} ||r_i - u_{\pi_u(i)}||^2 \\
& \leq & \frac{\eps}{2}.
\end{eqnarray*} 
Similarly, we conclude that $\Pr[A(v) \neq \pi_v(i)] \leq \eps/2$,
which implies that the probability that constraint $\pi_{uv}$ is not
satisfied is at most $\eps$.

\subsubsection{Shift Invariant Instances}\label{si}

In the case of Linear Equations $\bmod ~p$, we can add more
constraints to the SDP relaxation, which allow for a simplified
analysis of the rounding algorithm.  For any assignment of labels, we
can shift each of the labels by the same fixed amount, i.e, by adding
a value $k \in L$ to each label, and obtain an assignment with the
same objective value.  This property of a solution has been referred
to as {\em shift invariance}.  In these instances, the following are
valid constraints.  Note that $p = |L|$.
\begin{eqnarray*}
||u_i||^2 & = & \frac{1}{p} \quad u \in V, ~ i \in L, \\
u_i \cdot v_j & = & u_{i+k} \cdot v_{j+k} \quad u,v \in V, ~ i,j, k \in L. 
\end{eqnarray*}
In this case, we obtain a stronger version of \lref[Lemma]{simple_instance}.  
\begin{lemma}\label{si_instance} In a shift invariant instance in
  which every edge contributes more than $1- 1/2(d+1)$ to the SDP objective
  value, it is possible to efficiently find an assignment that
  satisfies all of the constraints.
\end{lemma}

We will show that in this case, the vector $r_i$ is closer to 
vector $u_{\pi_u(i)}$ than to vector $u_j$ for any label $j \neq
\pi_u(i)$.  In other words, $r_i \cdot u_{\pi_u(i)} > r_i \cdot u_j$
for all $j \in L$.
If each edge contributes more than $1-1/2(d+1)$ to the objective value,
then $||r_i - u_{\pi_u(i)}||^2 < 1/p$.  This implies that
$r_i \cdot u_{\pi_u(i)} > 1/2p$.
By triangle inequality, we have:
\begin{eqnarray*}
||u_j-u_{\pi_u(i)}||^2 & \leq & ||u_j-r_i||^2 ~ + ~ ||r_i-u_{\pi_u(i)}||^2 \\
\frac{2}{p} & \leq & \frac{2}{p} ~ - ~ 2 r_i \cdot u_j ~ + ~
\frac{1}{p} \quad \Rightarrow \\
r_i \cdot u_j & \leq & \frac{1}{2p}.
\end{eqnarray*}
Assuming that vector $u_j$ is closer to $r_i$ than vector
$u_{\pi_u(i)}$, we obtain the following contradiction:
\begin{eqnarray*}
\frac{1}{2p} & < & r_i \cdot u_{\pi_u(i)} 
~ \leq ~ r_i \cdot u_j  ~ \leq ~  \frac{1}{2p}.
\end{eqnarray*}

Note that in the case of shift invariance, $r$ is assigned each label
from $L$ with equal probability.  Because of shift invariance, it does
not actually matter which label $r$ is assigned.  
Thus, we can just assign $r$ a label $i$ arbitrarily (we
no longer need randomization) and then proceed with the rest of the
SDP rounding algorithm.  

\subsubsection{Extension to General Instances}

Applying this SDP rounding to general graphs may not yield such good
results as in Lemmas \ref{simple_instance} and \ref{si_instance},
since the radius of an arbitrary graph can be large, and the objective
values of the SDP relaxation would therefore have to be very high for
the lemmas to be applicable.  In order to apply these lemmas, we break
the graph into pieces, each with a radius of no more than
$O(\log{n}/\eps)$.  Doing this requires throwing out no more than an
$\eps$-fraction of the constraints.  
The following lemma is originally due to Leighton and
Rao \cite{LeightonR1999} and can also be found in \cite{Trevisan2008}.

\begin{lemma}
For a given graph $G=(V,E)$ and for all $\eps > 0$, there is a
polynomial time algorithm to find a subset of edges $E' \subseteq E$
such that $|E'| > (1-\eps)|E|$, and every connected connected
component of $E'$ has diameter $O(\log{|E|}/\eps)$.
\end{lemma}

Using this lemma, we obtain the following guarantee for general
instances: Given an instance for which OPT is at least $(1 - c
\eps^3/\log{n})|E|$, we can efficiently find a labeling satisfying a
$1-\eps$ fraction of the constraints.  Note that $c$ is an absolute
constant.  For shift invariant instances, we can satisfy $(1-\eps)|E|$
of the constraints for an instance where OPT is at least $(1 - c
\eps^2/\log{n})|E|$.

Given a graph, we remove the $\frac{\eps}{3}$ fraction of constraints
that contribute the least to the objective value.  This leaves us with
at least $(1 - \eps/3)|E|$ constraints that each contributes at least
$1 - 3 c \eps^2/\log{n}$ (or $1- 3 c \eps/\log{n}$ for shift invariant
instances) to the objective value.  We can apply
\lref[Lemma]{simple_instance} (or \lref[Lemma]{si_instance}) with $d =
\log{n}/\eps$, satisfying at least $(1-2\eps/3)|E|$ constraints (or
$(1-\eps/3)|E|$ constraints).

\section{Improving the Approximation Ratio}

Algorithms with improved approximation guarantees for Unique Games
have been presented in \cite{GuptaT2006, CharikarMM2006}.  The latter
work gives an algorithm with the following guarantee: Given an
instance of Unique Games with a domain size $k$ for which OPT is at
least $(1-\eps)|E|$, the algorithm produces a solution that satisfies
at least $\max \{1- \sqrt{\eps \log{k}}, ~k^{-\eps/(2-\eps)} \}$
fraction of the constraints.  Furthermore, it has been shown that the
existence of an efficient algorithm that can distinguish between
instances in which $(1-\eps)|E|$ constraints can be satisfied and
those at which less than $k^{-\eps/2}$ constraints can be satisfiable,
would disprove the Unique Games Conjecture~\cite{KhotKMO2007}.
Moreover, it is sufficient to refute the conjecture if this algorithm
works only for the special case of Linear Equations $\bmod ~p$.  Thus,
focusing on shift invariant instances is a reasonable approach.

Additionally, the Unique Games problem has been studied for cases in
which the constraint graph is an expander; in an instance in which OPT
is at least $(1-\eps)|E|$, one can efficiently find a solution
satisfying at least $1 - O(\frac{\eps}{\lambda})$ fraction of the
constraints, where $\lambda$ is a function of the expansion of the
graph~\cite{AroraKKSTV2008, MakarychevM2009}.

\section{Consequences}

The interest in the Unique Games Conjecture has grown due to the many
strong, negative consequences that have been proved for various
optimization problems.  Assuming the Unique Games Conjecture, it has
been shown that the Goemans-Williamson algorithm for Max Cut
(presented in Sanjeev's lecture) achieves the optimal approximation
ratio~\cite{KhotKMO2007}.  More surprisingly,
there are many other NP-complete optimization problems for which the
best-known approximation guarantees are obtained via extremely simple
algorithms.  Nevertheless, no one has been able to find algorithms
with improved approximation guarantees, even when resorting to
sophisticated techniques such as linear and semidefinite programming.
Such optimization problems include the Minimum Vertex Cover problem
and the Maximum Acyclic Subgraph problem, for which the best-known
approximation factors are $1/2$ and $2$, respectively.  If the Unique
Games Conjecture is true, then these approximation ratios are
tight~\cite{KhotR2008,GuruswamiMR2008}.
This phenomena has been investigated for several other optimization
problems as well.  A recent result shows that for a whole class of
constraint satisfaction problems, which can be modeled using a
particular integer program, the integrality gap of a particular SDP
relaxation is exactly equal to its approximability threshold under the
Unique Games Conjecture~\cite{Raghavendra2008}.



\section*{Appendix}\label{append}

We include the following lemma from \cite{Trevisan2008} and its proof:
\begin{lemma}\label{lucas_lemma}
Let $\rvector, \vectoru, \vectorv$ be vectors such that: (i) $\vectoru\cdot \vectorv = 0$, (ii) $||\rvector -
\vectoru||^2 \geq ||\rvector - \vectorv||^2$, and $(iii)$ the vectors $\rvector,\vectoru,\vectorv$ satisfy the
triangle inequality constraints from the SDP.  Then $||\rvector - \vectoru||^2 \geq
\frac{1}{2} ||\rvector||^2$.
\end{lemma}

\begin{proof}
There are three cases:
\begin{itemize}
\item[1.] {If $||\vectoru||^2 ~\leq~ \frac{1}{2} ||\rvector||^2$, then by
  \eqref{tri_ineq_2}, we have: $$||\rvector-\vectoru||^2
  ~\geq~ ||\rvector||^2 - ||\vectoru||^2 ~\geq~ \frac{1}{2} ||\rvector||^2.$$}

\item[2.] {If $||\vectorv||^2 ~ \leq ~ \frac{1}{2}||\rvector||^2$, then by
  \eqref{tri_ineq_1}, and subsequently \eqref{tri_ineq_2}, we have: 
$$||\rvector - \vectoru||^2 ~ \geq ~ ||\rvector-\vectorv||^2 ~ \geq ~ ||\rvector||^2 - ||\vectorv||^2 ~
  \geq ~ \frac{1}{2} ||\rvector||^2.$$
}

\item[3.] {If $||\vectoru||^2, ||\vectorv||^2 \geq \frac{1}{2}||\rvector||^2$, then from
  \eqref{tri_ineq_1} and assumption (ii), we have: $$||\vectorv - \vectoru||^2 ~
  \leq ~ ||\vectorv - \rvector||^2 + ||\rvector - \vectoru||^2 ~\leq ~ 2 ||\rvector - \vectoru||^2.$$
By Pythagoras theorem and by orthogonality of $\vectoru$ and $\vectorv$
(assumption (i)), we have:  $$||\vectorv - \vectoru||^2 ~ = ~ ||\vectorv||^2 +
||\vectoru||^2.$$  Finally, we have: $$||\rvector- \vectoru||^2 ~ \geq ~ \frac{1}{2}
||\vectorv - \vectoru||^2 ~ = ~ \frac{1}{2} ||\vectorv||^2 + \frac{1}{2} ||\vectoru||^2 ~ \geq
~ \frac{1}{2} ||\rvector||^2.$$}

\end{itemize}

\end{proof}

\cleardoublepage

\lecture{9}{21 July, 2009}{Igor Gorodezky}{Subhash Khot}{Unique Games Hardness for MAXCUT}

\section{Introduction: MAXCUT and Unique Games}

In this lecture we sketch the proof of one of the more remarkable
consequences of the Unique Games Conjecture:
MAX-CUT is inapproximable to any constant
better than $\alpha$, where
\beq\label{alpha}
\alpha = \min_{-1 \leq \rho \leq 1} \frac{2}{\pi} \frac{\arccos(\rho)}{1-\rho} \approx .87856
\eeq
is the approximation ratio of the Goemans-Williamson algorithm. 

\subsection{The Goemans-Williamson algorithm}

Recall that 
given a graph $G=(V,E,w)$ with edge weights $w_{ij} \geq 0$,
the MAX-CUT problem asks for $S \subseteq V$ that 
maximizes $\sum_{i \in S, j \notin S} w_{ij}$ (we call this 
the \emph{weight} of the cut induced by $S$). 
We will write $\mc(G)$ for the maximum weight
of a cut in $G$.  

MAX-CUT is NP-hard. The best known approximation 
algorithm for MAX-CUT is due to Goemans and Williamson \cite{GoemansW1995}
and is as follows: given $G=(V,E,w)$, first solve the MAX-CUT SDP
\begin{alignat}{2}\label{sdp}
 \text{max }  & \sum_{ij} w_{ij} \frac{1-v_i \cdot v_j}{2} & &\\
&  ||v_i||^2 =1 ,& & i =1, \dots, n \notag \\
&  v_i \in \mathbb{R}^n ,& & i =1, \dots, n \notag 
\end{alignat}
to get a set of unit vectors $v_1, \dots, v_n$ in $\mathbb R^n$. 
Then, uniformly sample a hyperplane through the origin
and define $S \subseteq V$ to be the set of $i$ such that
$v_i$ lies ``above" this hyperplane. 

We saw in Sanjeev's lecture (Lecture 1) that the Goemans-Williamson
algorithm gives, in expectation, an $(\alpha-\epsilon)$-approximation 
for any $\epsilon>0$ (this additive error of $\epsilon$
stems from the fact that semidefinite programs must be
solved to within some fixed, arbitrarily small accuracy). 
The algorithm can be derandomized (see \cite{MahajanH1999}) to yield
a deterministic $(\alpha-\epsilon)$-approximation to MAX-CUT. 

A series of subsequent hardness results culminated in
H\aa stad's PCP-based proof in \cite{Hastad2001} that MAX-CUT is NP-hard to approximate
to within $16/17 \approx .941$. Then,
roughly a decade after the publication of the Goemans-Williamson
algorithm, Khot, Kindler, Mossel and O'Donnell proved in \cite{KhotKMO2007}
that, assuming the Unique Games Conjecture, it is NP-hard
to approximate MAX-CUT to within any factor greater than $\alpha$.
This suggests, as Khot et al.~note, that the geometric nature
of the Goemans-Williamson algorithm is intrinsic to the MAX-CUT
problem. 

\subsection{Label Cover and Unique Games}

The inapproximability of MAX-CUT is conditional
on the Unique Games Conjecture, which we state in this section. 

A \emph{unique game} $\uginst$ is a bipartite
graph with left-side vertex set $V,$ right-side vertex set $W,$
edge set $E$, and a set of labels of size $M$. Each
edge $(v,w)$ has an associated constraint function $\pi_{v,w}$
which is a permutation of $[M]$ (i.e.~a bijection from the set of
labels to itself). We will sometimes refer to a unique game 
with these parameters in the longhand $\uginst(V,W,E,[M],\{\pi_{v,w}\})$. 

A \emph{labeling} of a unique game $\uginst$ is an assignment of a label
from $[M]$ to each vertex of $\uginst$. A labeling \emph{satisfies}
the edge $(v,w)$ if $\pi_{v,w}$ maps the label of $w$ to the
label of $v$. We define
\[
\opt(\uginst) = \max\{ \ r \mid \text{there exists a labeling of }\uginst\text{ that satisfies } r|E| \text{ edges}\}.
\]
The \emph{unique Label Cover problem with parameter $\delta$}
is the problem of deciding, given a unique game
$\uginst(V,W,E,[M],\{\pi_{v,w}\})$, whether $\opt(\uginst) \geq 1-\delta$
or $\opt(\uginst) \leq \delta$. That is, given $\uginst$, we
are asked to decide whether there exists a labeling that
satisfies nearly all edge constraints, or whether no labeling
can satisfy more than a tiny fraction of them. 
We will write this decision problem as ULC($\delta$).

Intuition tells us that computing $\opt(\uginst)$
should be a hard problem, but what about ULC($\delta$)?
That is, if we are asked not to compute $\opt(\uginst)$ but
simply to decide whether it is very large or very small, does
the problem become easier? 
The Unique Games Conjecture claims that the answer is no. \\

\begin{figure}[h] 
   \centering
   \includegraphics[width=3.5in]{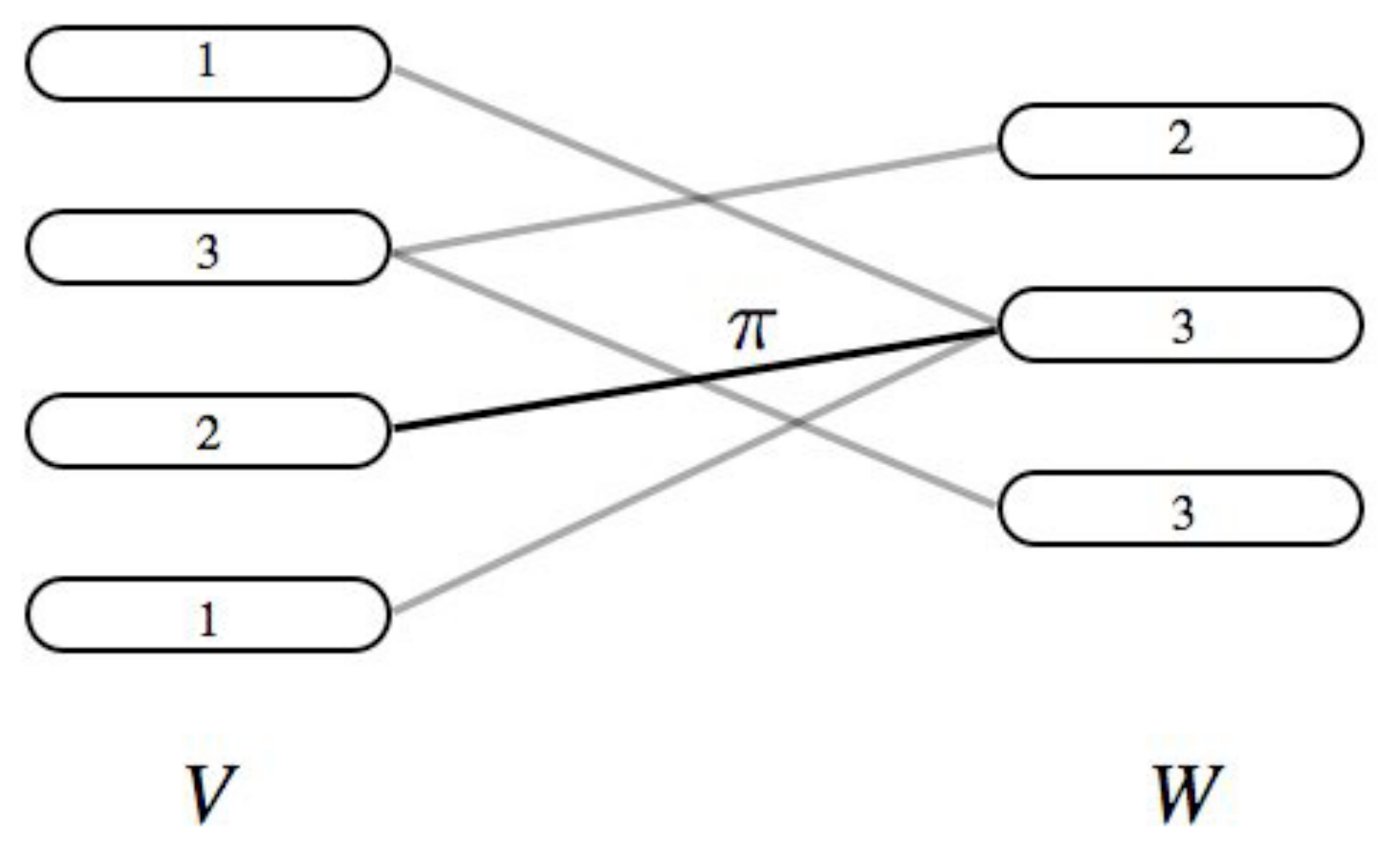} 
   \caption[A unique game with $M=3$]{A unique game with $M=3$. The constraint $\pi$ is associated
   with the highlighted edge. The edge is satisfied by the labeling in the figure
   if $\pi(3)=2$. }
   \label{fig:ug}
\end{figure}

The following is (a slightly weakened form of) the conjecture, 
first stated by Khot in \cite{Khot2002}.\\

\noindent {\bf Unique Games Conjecture (\cite{Khot2002}) : }{\em For any
  $\delta > 0$ there exists a constant $M$ such that it is NP-hard to
  decide ULC($\delta$) on instances with a label set of size $M$. }

\subsection{The Main Result}

We are ready to present the main result of \cite{KhotKMO2007}. 
This result can be stated in the form of a PCP for 
the ULC($\delta$) problem, by which we mean a 
probabilistically checkable proof system such that
given a unique game $\uginst$ and a proof, 
if $\opt(\uginst) \geq 1-\delta$ then the verifier accepts with high
probability $c$ (completeness), while if $\opt(\uginst) \leq \delta$
then the verifier accepts with low probability $s$ (soundness). 
As usual, the verifier only uses $O(\log n)$ random bits
on an instance of size $n$. 

\begin{theorem}\label{thm_main}
For every $\rho \in (-1,0)$ and $\epsilon > 0$ there exists
$\delta > 0$ such that
there is a PCP for ULC($\delta$)
in which the verifier reads two bits from the proof and accepts
iff they are unequal, and which has completeness 
\[
c \geq \frac{1-\rho}{2} - \epsilon
\]
and soundness 
\[
s \leq \frac{1}{\pi} \arccos(\rho)+\epsilon . 
\]
\end{theorem}

Before sketching the proof of this theorem in \lref[Section]{sec_pcp}, 
let us see how
it implies the inapproximability of MAX-CUT. Recall that
$\mc(G)$ is the maximum weight of a cut in $G$. 

\begin{corollary}\label{cor_reduc}
For every $\rho \in (-1,0)$ and $\epsilon > 0$ there exists $\delta > 0$
and a polynomial-time
reduction from an instance $\uginst$ of ULC($\delta$) to 
an instance $G=(V,E,w)$ of MAX-CUT such that
\begin{alignat*}{2}
&\opt(\uginst) \geq 1-\delta \quad & &\Longrightarrow \quad \mc(G) \geq  \frac{1-\rho}{2} - \epsilon \\
&\opt(\uginst) \leq \delta \quad & & \Longrightarrow \quad \mc(G) \leq \frac{1}{\pi} \arccos(\rho)+\epsilon.
\end{alignat*}
\end{corollary}
\begin{proof}
Given $\rho$ and $\epsilon$, let $\delta$ be the same as in
\lref[Theorem]{thm_main}. Given an instance $\uginst$ of ULC($\delta$),
consider the PCP given by that theorem. Define the graph $G$
to have the bits of the proof as vertices,\footnote{As we will later see, 
if $\uginst$ is of the form
$\uginst(V,W,E,[M],\{ \pi_{v,w}\})$ then there will be $|W|2^M$ vertices in $G$.}
and create an edge between two bits if there is a non-zero probability
of that pair of bits being sampled by the verifier. Finally, set $w$ to be 
the trivial weight function that is 1 on all edges. 

Observe that a proof, which is an assignment of a value in $\{-1,1\}$ to the bits,
corresponds to a cut in $G$, and the number of edges crossing this 
cut is precisely the probability that this proof is accepted by the verifier. 
The claim now follows from the completeness and
soundness of the PCP. 
\end{proof}

Assuming the Unique Games Conjecture, for any $\delta > 0$
there is some constant $M$ such that it is NP-hard to decide
ULC($\delta$) on instances with a label set of size $M$. 
Now, by standard arguments, \lref[Corollary]{cor_reduc} implies that it is 
NP-hard to approximate MAX-CUT to within
\[
\frac{\arccos(\rho)/\pi+\epsilon}{(1-\rho)/2-\epsilon} > \frac{\arccos(\rho)/\pi}{(1-\rho)/2} 
\]
for any $\rho \in (-1,0)$. Therefore, MAX-CUT
is hard to approximate to within any constant larger than
\[
\min_{-1 \leq \rho \leq 0} \frac{2}{\pi} \frac{\arccos(\rho)}{1-\rho} =
\min_{-1 \leq \rho \leq 1} \frac{2}{\pi} \frac{\arccos(\rho)}{1-\rho} = \alpha
\]
which is the promised inapproximability result. 

Let us turn our attention, then, to proving \lref[Theorem]{thm_main}.
The proof will rely on a highly nontrivial result in
boolean Fourier analysis that we state in the next section.

\section{Majority is Stablest}\label{sec_mis}

The proof of \lref[Theorem]{thm_main} makes crucial use of
the Majority is Stablest (MIS) theorem, which is an extremal result
in boolean Fourier analysis. In this section we
state this theorem after defining the necessary concepts. 

We will use the common convention that bits take value in $\{-1,1\}$
rather than $\{0,1\}$ (in particular, we identify $x \in \{0,1\}$ with
$y \in \{-1,1\}$ using the bijection $y=(-1)^x$). 
Thus, a \emph{boolean function} is a map $f: \{-1,1\}^n \to \{-1,1\}$. 
In what follows we will assume familiarity with the basic
concepts of boolean Fourier analysis, as we lack the space for 
a thorough review of the subject; 
a reader seeking such a review is directed to the survey \cite{Odonnell2008}. \\

We begin with several definitions. 
Given $f:\{-1,1\}^n \to \{-1,1\}$, define the \emph{influence} of $x_i$
on $f$ to be the probability over all $n$-bit strings that $f$ changes value when
the $i$th bit is flipped:
\[
\infl(f) = \prob{x \in \{-1,1\}^n}{f(x) \neq f(x_1, \dots, x_{i-1}, -x_i, x_{i+1}, \dots, x_n)}.
\]
It is not hard to show that
\beq\label{def_infl}
\infl(f) = \sum_{S \mid i \in S} \hat f(S)^2
\eeq
and indeed, equation \eqref{def_infl} can be used to \emph{define}
the influence of a variable on a non-boolean function $f:\{-1,1\}^n \to \R$. 

Given a bit string $x$ and some $\rho \in (-1,1)$, let us define 
a distribution over $y \in \{-1,1\}^n$ by setting 
\[
y_i=
\begin{cases}
x_i \text{ with probability } \frac{1+\rho}{2} \\
-x_i \text{ with probability } \frac{1-\rho}{2} 
\end{cases}
\]
We write $y \sim_{\rho} x$ to mean a $y$ sampled from such a
distribution. 
Now, given a boolean function $f$ and $\rho \in (-1,1)$, define the
\emph{noise sensitivity} of $f$ at rate $\rho$ to be
\[
\ns(f) = \prob{x, \, y \sim_{\rho} x}{f(x) \neq f(y)}
\]
where $x$ is sampled uniformly and $y \sim_{\rho} x$. 
It can be shown that
\beq\label{def_ns}
\ns(f)=\frac{1}{2}-\frac{1}{2} \sum_S \hat f(S)^2 \rho^{|S|}.
\eeq
As before, equation \eqref{def_ns} serves as the \emph{definition}
of noise sensitivity for non-boolean functions. 

We observe that if $f$ is a \emph{dictator function}, i.e.~$f(x_1,\dots,x_n)=x_i$
for some $i$, then $\ns(f)$ is exactly $ (1-\rho)/2$ since $\hat f(S)=0$
when $S \neq \{x_i\}$ and is 1 otherwise. If $f$ is the majority
function (the boolean function whose value on a bit string 
is equal to the value of the majority of the bits), then it can be shown
(see \cite{KhotKMO2007} for references) that 
\[
\ns(f) = \frac{1}{\pi}\arccos(\rho)+o(1).
\]
Observe that when $\rho \in [0,1)$,
the noise sensitivity of a dictator function 
is lower than that of the majority function. 
The MIS theorem (proven in \cite{MosselOO2005}) tells us that if
we disqualify dictators by restricting 
our attention to functions in which no 
coordinate has large influence, then the majority function
achieves the smallest possible noise sensitivity. 

\begin{theorem}[Majority is Stablest, \cite{MosselOO2005}]\label{thm_mis}
For every $\rho \in [0,1), \epsilon > 0$ there exists $\delta$
such that if $f: \{-1,1\}^n \to [-1,1]$ with $\Exp{}{f}=0$ and $\infl(f) \leq \delta \ \forall i$,
then 
\[
\ns(f) \geq \frac{1}{\pi} \arccos \rho - \epsilon. 
\]
\end{theorem}

Note the additional requirement that $\Exp{}{f}=0$; such functions 
are called \emph{balanced}. 
Our application requires the following corollary from \cite{KhotKMO2007}.
It states that if we choose a negative rather than positive $\rho$,
the majority function becomes the \emph{least} stable. 

\begin{corollary}\label{cor_mius}
For every $\rho \in (-1,0), \epsilon > 0$ there exists $\delta$
such that if $f: \{-1,1\}^n \to [-1,1]$ with $\infl(f) \leq \delta \ \forall i$,
then 
\[
\ns(f) \leq \frac{1}{\pi} \arccos \rho + \epsilon. 
\]
\end{corollary}
Observe that $f$ is no longer required to be balanced. 

Looking ahead to the proof of \lref[Theorem]{thm_main},
the test that the verifier will perform on the PCP in
the theorem will be, in a way, a noise-sensitivity
test on a boolean function. Thus, we will be able to 
bound the soundness
of this test by appealing to the MIS theorem.

\section{\texorpdfstring{Proving {\lref[Theorem]{thm_main}}}{Proving Theorem}}\label{sec_pcp}

In this section we sketch the construction
of the PCP whose existence is claimed in 
\lref[Theorem]{thm_main}. Recall that 
the PCP is for the problem ULC($\delta$), 
so the verifier is given a unique game
$\uginst(V,W,E,[M],\{ \pi_{v,w}\})$ and a proof
that is accepted with high probability if $\opt(\uginst) \geq 1-\delta$
and accepted with low probability if $\opt(\uginst) \leq \delta$. 
The PCP will be parameterized by $\rho \in (-1,0)$ and $\epsilon > 0$. 

It follows from the results of \cite{KhotR2008} that given $\uginst$,
we may assume with no loss of generality
that all $v \in V$ have the same degree. Thus, uniformly sampling
$v \in V$ and then uniformly sampling a neighbor $w \in W$ of $v$
yields a uniformly random edge $(v,w)$. Therefore, 
if we define a proof to be a labeling of $\uginst$ that
maximizes the proportion of satisfied edge constraints,
and define the verifier's test to be uniformly sampling 
$(v,w)$ and checking if this labeling satisfies $\pi_{v,w}$,
then we would have a proof with completeness $1-\delta$
and soundness $\delta$. 
However, such a test involves sampling $\Omega(\log M)$ bits, 
and we require a test that samples only 2.

We therefore look for a way to encode elements
of the label set $[M]$ in a way that will allow such a test. 
To this end, we will encode labels using the Long Code, which
we first saw in Subhash's lecture (Lecture 7) on H\aa stad's 3-bit PCP. 

\subsection{Motivation: the Long Code}

Recall that in the Long Code, the codeword encoding $i \in [M]$ is the
truth-table for the dictator function $f(x_1, \dots, x_n)=x_i$. 
In our PCP, the proof will be a labeling of $\uginst$ with each
label encoded using the Long Code. 
It remains to design a test for the verifier with the properties
specified in \lref[Theorem]{thm_main}. Before explicitly
stating the test in the next section, we use this section
to motivate its construction. 

Given a boolean function $f$, let us say that $f$ is \emph{far from
a dictator} if all coordinates have negligible influence.
It is not hard to design a 2-bit test that, given a truth-table for a function 
$f$, accepts with high probability if $f$ is 
a dictator and with low probability if $f$ is far from a dictator.
The test required for our PCP must clearly do more than this,
but for the moment let us warm up with this simpler problem. 

Consider the following noise-sensitivity test:
sample $x \in \{-1,1\}^n$ uniformly, sample $y\sim_{\rho} x$
as in \lref[Section]{sec_mis}, and accept iff $f(x) \neq f(y)$. 
By definition, the probability of accepting is $\ns(f)$, the
noise sensitivity of $f$. 

The completeness of this test is thus exactly the noise-sensitivity
of a dictator function, which is $(1-\rho)/2$. On the other hand,
we can use \lref[Corollary]{cor_mius} to bound the soundness: if $f$
is far from a dictator its noise-sensitivity is at most $\arccos(\rho)/\pi+\epsilon$. 

Returning to our PCP, we require a 2-bit noise-sensitivity test with
(almost) exactly these parameters that instead of testing whether
a boolean function is a dictator or far from it, tests whether an
encoded labeling of $\uginst$ (which consists of many boolean
functions) satisfies many edge constraints, or far from it. 

\subsection{The Test}

In this section we describe our PCP's verifier test. First, some notation.
For $x \in \{-1,1\}^M$ and a bijection $\pi : [M] \to [M]$, let
$x \circ \pi$ denote the string $(x_{\pi(1)}, \dots, x_{\pi(M)})$. 
Given a unique game $\uginstfull$ and the associated PCP
proof (which the verifier expects is the Long Code encoding of a
labeling), let $f_v$ be the Long Code encoding of the label
given to $v \in V$, and define $f_w$ for $w \in W$ analogously. 

\begin{center}
\framebox{\parbox{14cm}{
\begin{center}The 2-bit verifier test: \end{center}
\begin{enumerate}
\item Given $\uginstfull$ and a proof, sample $v \in V$ uniformly, then

sample two of its neighbors $w,w' \in W$ uniformly and independently. 
\item Let $\pi=\pi_{v,w}$ and $\pi'=\pi_{v,w'}$ be the constraints for edges 
$(v,w)$ and $(v,w')$. 
\item Sample $x \in \{-1,1\}^M$ uniformly and sample $y \sim_{\rho} x$. 
\item Accept if $f_w(x \circ \pi) \neq f_{w'}(y \circ \pi')$. 
\end{enumerate}
\vspace{.2cm}
}}
\end{center}

\textbf{Completeness}. Assume that $\opt(\uginst) \geq 1-\delta$ and that
the proof given to the verifier encodes all labels correctly (i.e.~as dictator
functions). Let the labels of $v,w,w'$ be $i,j,j' \in [M]$, respectively. 
With probability at least $1-2\delta$, 
both $(v,w)$ and $(v,w')$ are satisfied by the labeling, which
implies $\pi(j)=\pi'(j')=i$. Conditioning on this event, we have
\[
f_w(x \circ \pi) = x_{\pi(j)}=x_i \quad \text{ and } \quad f_{w'}(y \circ \pi') = y_{\pi'(j')}=y_i. 
\]
Since $x_i=y_i$ with probability $(1-\rho)/2$, the test accepts with the
same probability. We conclude that the completeness is at least
$(1-2\delta)(1-\rho)/2$. Tweaking our choice of $\rho$, we 
conclude that completeness is at least $(1-\rho)/2-\epsilon$, as desired. \\

\begin{figure}[h] 
   \centering
   \includegraphics[width=4in]{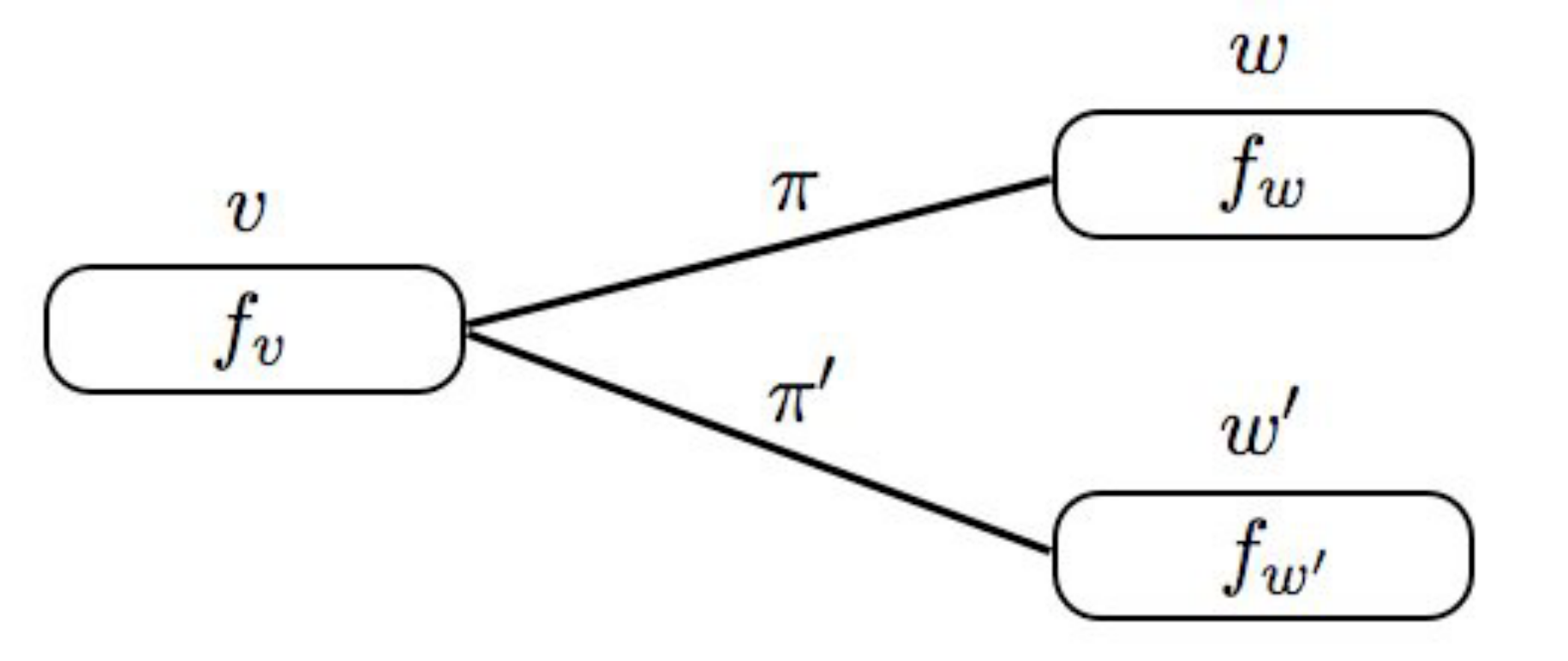} 
   \caption[The 2-bit verifier test]{The 2-bit verifier test samples $v \in V$, two neighbors $w,w' \in W$, then
   compares $f_{w}$ and $f_{w'}$ on certain inputs.}
   \label{fig:test}
\end{figure}

\textbf{Soundness}. As usual, bounding the soundness is the 
difficult part; we only sketch the argument. The proof is in
the contrapositive direction: assuming that the test accepts
with probability greater than $\arccos(\rho)/\pi+\epsilon$,
we prove the existence of a labeling that satisfies many edges
by exploiting the resulting Fourier-analytic properties
of the boolean functions encoded in the proof (which is
where \lref[Corollary]{cor_mius} comes in). 

The proof is as follows. Given $v \in V$, let $p_v$ be
the probability that the test accepts after choosing $v$ from $V$. Then
\begin{align*}
p_v &= \Exp{w,w',x,y\sim_{\rho}x}{\frac{1-f_w(x \circ \pi)f_{w'}(y \circ \pi')}{2}}\\
& = \frac{1}{2} - \frac{1}{2} \Exp{w,w'}{\Exp{x,y\sim_{\rho} x}{f_w(x \circ \pi)f_{w'}(y \circ \pi')}}.
\end{align*}
Standard Fourier-analytic arguments can be used to show that
\[
\Exp{x,y\sim_{\rho} x}{f_w(x \circ \pi)f_{w'}(y \circ \pi')} = \sum_S \hat f_w(S) \hat f_{w'}(S) \rho^{|S|}
\]
from which it follows that
\begin{align*}
p_v &  =  \frac{1}{2} - \frac{1}{2}  \Exp{w,w'}{\sum_S \hat f_w(S) \hat f_{w'}(S) \rho^{|S|}} \\
& =  \frac{1}{2} - \frac{1}{2}\sum_S  \Exp{w,w'}{\hat f_w(S) \hat f_{w'}(S) }\rho^{|S|} \\
& =  \frac{1}{2} - \frac{1}{2}\sum_S  \Exp{w \sim v}{\hat f_w(S)}\Exp{w' \sim v}{ \hat f_{w'}(S) }\rho^{|S|} 
\end{align*}
where the last equality follows from the independence of $w$ and $w'$
(and $w \sim v$ means that $w$ is a neighbor of $v$).  
If we define a function $g_v : \{-1,1\}^n \to [-1,1]$ by
\[
g_v(z) = \Exp{w \sim v}{f_w(z \circ \pi_{v,w})}
\]
then it is not hard to show that
\[
\hat g_v(S) = \Exp{w \sim v}{\hat f_w(S)}.
\]
Therefore, returning to $p_v$, we have
\begin{align*}
p_v &  =  \frac{1}{2} - \frac{1}{2}\sum_S  \Exp{w \sim v}{\hat f_w(S)}\Exp{w' \sim v}{ \hat f_{w'}(S) }\rho^{|S|} \\
& =  \frac{1}{2} - \frac{1}{2}\sum_S\hat g_v(S)^2 \rho^{|S|} \\
& = \ns(g_v)
\end{align*}
where the last equality is by equation \eqref{def_ns}. 

Recall that we assumed that the test is accepted
with probability at least $\arccos(\rho)/\pi+\epsilon$. Standard averaging arguments
tell us that $p_v \geq \arccos(\rho)/\pi+\epsilon/2$ for at least an $\epsilon/2$
fraction of $v \in V$. By the above, we have
\[
\ns(g_v) \geq  \arccos(\rho)/\pi+\epsilon/2
\]
for such $v$. Now we conclude by \lref[Corollary]{cor_mius} (having tweaked
$\epsilon$ as necessary) that for such a $v$, $g_v$ has
an influential coordinate. This fact can be used to show that 
for a constant fraction of neighbors $w$ of $v$, $f_w$ has a 
small set of influential coordinates. These various
influential coordinates can be used to define
a labeling that satisfies a large fraction of constraints; we
direct the reader to \cite{KhotKMO2007} for the gory details. 

A final caveat: technically, if we want to prove that 
$f_w$ has a small set of influential coordinates
for a constant fraction of neighbors $w$ of $v$, it is not enough
to assume that $g_v$ has a coordinate with large influence.
What is required is for $g_v$ to have a coordinate with large
\emph{low-degree} influence, which is defined, in analogy
to equation \eqref{def_infl}, as
\[
\inflk(f) = \sum_{S \, \mid \, i \in S, |S| \leq k} \hat f(S)^2
\]
for some constant $k$. If we define low-order noise
sensitivity in obvious analogy to equation \eqref{def_ns},
it is possible to prove low-order analogues of
the Majority is Stablest theorem and \lref[Corollary]{cor_mius},
which can then be used to formalize the argument that we have sketched. 

This completes the description of the PCP test and the proof
of \lref[Theorem]{thm_main}.

\section{ The Big Picture}

The past few years have seen a flurry of powerful inapproximability 
results conditional on the Unique Games Conjecture. 
In \cite{Raghavendra2008}, Raghavendra exhibits a canonical
semidefinite programming relaxation of an arbitrary CSP
whose integrality gap, assuming UGC, is precisely equal 
to the best possible approximation ratio for that CSP. 
In FOCS 2009, Raghavendra and Steurer presented an efficient
rounding scheme for these SDPs that achieves the integrality gap. 

UGC has been used to prove that Vertex Cover
is conditionally inapproximable to within $2-\epsilon$ (see \cite{KhotR2008}). 
This proof utilizes a theorem on the influence of boolean functions
due to Friedgut. In addition, it was independetly shown by Khot and Vishnoi \cite{KhotV2005}
and Chawla et al.~\cite{ChawlaKKRS2006} that UGC implies the hardness
of approximating Sparsest Cut to within any constant. These
proofs, as expected, use theorems on the influence of boolean
functions due to Kahn-Kalai-Linial and Bourgain. 

Underlying these results are surprising and fruitful
connections between unique game reductions,
semidefinite programming relaxations of CSPs,
extremal problems in Fourier analysis, and isoperimetric
problems in geometry. The reader is directed to
Section 5 of \cite{KhotKMO2007} for an insightful high-level
discussion of how 
these connections are manifested in the 
particular case of MAX-CUT.

\cleardoublepage

{\footnotesize
\bibliographystyle{prahladhurl}
\phantomsection
\addcontentsline{toc}{chapter}{Bibliography}
\bibliography{../jrnl-names-abb,../prahladhbib,../crossref,ref}

\newcommand{\etalchar}[1]{$^{#1}$}
\begin{thebibliography}{KKMO07}

\bibitem[AA97]{AwerbuchA1997}
\textsc{Baruch Awerbuch} and \textsc{Yossi Azar}.
\newblock \emph{Buy-at-bulk network design}.
\newblock In \emph{Proc.\ $38$rd IEEE Symp.\ on Foundations of Comp.\ Science
  (FOCS)}, pages 542--547. IEEE, 1997.
\newblock
  \href{http://dx.doi.org/10.1109/SFCS.1997.646143}{\path{doi:10.1109/SFCS.199%
7.646143}}.

\bibitem[AB09]{AroraB2009}
\textsc{Sanjeev Arora} and \textsc{Boaz Barak}.
\newblock \emph{Computational Complexity: A Modern Approach}.
\newblock Cambridge University Press, 2009.

\bibitem[ACG{\etalchar{+}}07]{AndrewsCGKTZ2007}
\textsc{Matthew Andrews}, \textsc{Julia Chuzhoy}, \textsc{Venkatesan
  Guruswami}, \textsc{Sanjeev Khanna}, \textsc{Kunal Talwar}, and \textsc{Lisa
  Zhang}.
\newblock \emph{Inapproximability of edge-disjoint paths and low congestion
  routing on undirected graphs}.
\newblock Technical Report TR07-113, Electronic Colloquium on Computational
  Complexity, 2007.
\newblock \href{http://eccc.hpi-web.de/report/2007/113}{\path{eccc:TR07-113}}.

\bibitem[AKK{\etalchar{+}}08]{AroraKKSTV2008}
\textsc{Sanjeev Arora}, \textsc{Subhash Khot}, \textsc{Alexandra Kolla},
  \textsc{David Steurer}, \textsc{Madhur Tulsiani}, and \textsc{Nisheeth~K.
  Vishnoi}.
\newblock \emph{Unique games on expanding constraint graphs are easy: extended
  abstract}.
\newblock In \emph{Proc.\ $40$th ACM Symp.\ on Theory of Computing (STOC)},
  pages 21--28. ACM, 2008.
\newblock
  \href{http://dx.doi.org/10.1145/1374376.1374380}{\path{doi:10.1145/1374376.1%
374380}}.

\bibitem[AKR95]{AgrawalKR1995}
\textsc{Ajit Agrawal}, \textsc{Philip~N. Klein}, and \textsc{R.~Ravi}.
\newblock \emph{When trees collide: An approximation algorithm for the
  generalized steiner problem on networks}.
\newblock SIAM J. Computing, 24(3):440--456, 1995.
\newblock (Preliminary version in {\em 23rd STOC}, 1991).
\newblock
  \href{http://dx.doi.org/10.1137/S0097539792236237}{\path{doi:10.1137/S009753%
9792236237}}.

\bibitem[ALM{\etalchar{+}}98]{AroraLMSS1998}
\textsc{Sanjeev Arora}, \textsc{Carsten Lund}, \textsc{Rajeev Motwani},
  \textsc{Madhu Sudan}, and \textsc{Mario Szegedy}.
\newblock \emph{Proof verification and the hardness of approximation problems}.
\newblock J. ACM, 45(3):501--555, May 1998.
\newblock (Preliminary Version in {\em 33rd FOCS}, 1992).
\newblock \href{http://eccc.hpi-web.de/report/1998/008}{\path{eccc:TR98-008}},
  \href{http://dx.doi.org/10.1145/278298.278306}{\path{doi:10.1145/278298.2783%
06}}.

\bibitem[And04]{Andrews2004}
\textsc{Matthew Andrews}.
\newblock \emph{Hardness of buy-at-bulk network design}.
\newblock In \emph{Proc.\ $45$th IEEE Symp.\ on Foundations of Comp.\ Science
  (FOCS)}, pages 115--124. IEEE, 2004.
\newblock
  \href{http://dx.doi.org/10.1109/FOCS.2004.32}{\path{doi:10.1109/FOCS.2004.32%
}}.

\bibitem[AR06]{AzarR2006}
\textsc{Yossi Azar} and \textsc{Oded Regev}.
\newblock \emph{Combinatorial algorithms for the unsplittable flow problem}.
\newblock Algorithmica, 44(1):49--66, 2006.
\newblock (Preliminary version in {\em IPCO}, 2001).
\newblock
  \href{http://dx.doi.org/10.1007/s00453-005-1172-z}{\path{doi:10.1007/s00453-%
005-1172-z}}.

\bibitem[Aro03]{Arora2003}
\textsc{Sanjeev Arora}.
\newblock \emph{Approximation schemes for {NP}-hard geometric optimization
  problems: a survey}.
\newblock Mathematical Programming, 97(1--2):43--69, 2003.
\newblock
  \href{http://dx.doi.org/10.1007/s10107-003-0438-y}{\path{doi:10.1007/s10107-%
003-0438-y}}.

\bibitem[AS98]{AroraS1998}
\textsc{Sanjeev Arora} and \textsc{Shmuel Safra}.
\newblock \emph{Probabilistic checking of proofs: A new characterization
  of~{NP}}.
\newblock J. ACM, 45(1):70--122, January 1998.
\newblock (Preliminary Version in {\em 33rd FOCS}, 1992).
\newblock
  \href{http://dx.doi.org/10.1145/273865.273901}{\path{doi:10.1145/273865.2739%
01}}.

\bibitem[AS03]{AroraS2003}
\textsc{Sanjeev Arora} and \textsc{Madhu Sudan}.
\newblock \emph{Improved low-degree testing and its applications}.
\newblock Combinatorica, 23(3):365--426, 2003.
\newblock (Preliminary Version in {\em 29th STOC}, 1997).
\newblock \href{http://eccc.hpi-web.de/report/1997/003}{\path{eccc:TR97-003}},
  \href{http://dx.doi.org/10.1007/s00493-003-0025-0}{\path{doi:10.1007/s00493-%
003-0025-0}}.

\bibitem[AZ06]{AndrewsZ2006}
\textsc{Matthew Andrews} and \textsc{Lisa Zhang}.
\newblock \emph{Logarithmic hardness of the undirected edge-disjoint paths
  problem}.
\newblock J. ACM, 53(5):745--761, 2006.
\newblock (Preliminary version in {\em 37th STOC}, 2005).
\newblock
  \href{http://dx.doi.org/10.1145/1183907.1183910}{\path{doi:10.1145/1183907.1%
183910}}.

\bibitem[AZ07]{AndrewsZ2007}
---{}---{}---.
\newblock \emph{Hardness of the undirected congestion minimization problem}.
\newblock SIAM J. Computing, 37(1):112--131, 2007.
\newblock (Preliminary version in {\em 37th STOC}, 2005).
\newblock
  \href{http://dx.doi.org/10.1137/050636899}{\path{doi:10.1137/050636899}}.

\bibitem[AZ08]{AndrewsZ2008}
---{}---{}---.
\newblock \emph{Almost-tight hardness of directed congestion minimization}.
\newblock J. ACM, 55(6), 2008.
\newblock (Preliminary version in {\em 38th STOC}, 2006).
\newblock
  \href{http://dx.doi.org/10.1145/1455248.1455251}{\path{doi:10.1145/1455248.1%
455251}}.

\bibitem[AZ09]{AndrewsZ2009}
---{}---{}---.
\newblock \emph{Complexity of wavelength assignment in optical network
  optimization}.
\newblock IEEE/ACM Trans. Netw., 17(2):646--657, 2009.
\newblock (Preliminary version in {\em 25th INFOCOM}, 2006).
\newblock
  \href{http://dx.doi.org/10.1145/1552193.1552215}{\path{doi:10.1145/1552193.1%
552215}}.

\bibitem[Bar98]{Bartal1998}
\textsc{Yair Bartal}.
\newblock \emph{On approximating arbitrary metrices by tree metrics}.
\newblock In \emph{Proc.\ $30$th ACM Symp.\ on Theory of Computing (STOC)},
  pages 161--168. ACM, 1998.
\newblock
  \href{http://dx.doi.org/10.1145/276698.276725}{\path{doi:10.1145/276698.2767%
25}}.

\bibitem[BGH{\etalchar{+}}06]{BenSassonGHSV2006}
\textsc{Eli {Ben-Sasson}}, \textsc{Oded Goldreich}, \textsc{Prahladh Harsha},
  \textsc{Madhu Sudan}, and \textsc{Salil Vadhan}.
\newblock \emph{Robust {PCP}s of proximity, shorter {PCP}s and applications to
  coding}.
\newblock SIAM J. Computing, 36(4):889--974, 2006.
\newblock (Preliminary Version in {\em 36th STOC}, 2004).
\newblock \href{http://eccc.hpi-web.de/report/2004/021}{\path{eccc:TR04-021}},
  \href{http://dx.doi.org/10.1137/S0097539705446810}{\path{doi:10.1137/S009753%
9705446810}}.

\bibitem[BGS98]{BellareGS1998}
\textsc{Mihir Bellare}, \textsc{Oded Goldreich}, and \textsc{Madhu Sudan}.
\newblock \emph{Free bits, {PCPs}, and nonapproximability---towards tight
  results}.
\newblock SIAM J. Computing, 27(3):804--915, June 1998.
\newblock (Preliminary Version in {\em 36th FOCS}, 1995).
\newblock \href{http://eccc.hpi-web.de/report/1995/024}{\path{eccc:TR95-024}},
  \href{http://dx.doi.org/10.1137/S0097539796302531}{\path{doi:10.1137/S009753%
9796302531}}.

\bibitem[BLR93]{BlumLR1993}
\textsc{Manuel Blum}, \textsc{Michael Luby}, and \textsc{Ronitt Rubinfeld}.
\newblock \emph{Self-testing/correcting with applications to numerical
  problems}.
\newblock J. Computer and System Sciences, 47(3):549--595, December 1993.
\newblock (Preliminary Version in {\em 22nd STOC}, 1990).
\newblock
  \href{http://dx.doi.org/10.1016/0022-0000(93)90044-W}{\path{doi:10.1016/0022%
-0000(93)90044-W}}.

\bibitem[BS00]{BavejaS2000}
\textsc{Alok Baveja} and \textsc{Aravind Srinivasan}.
\newblock \emph{Approximation algorithms for disjoint paths and related routing
  and packing problems}.
\newblock Math. Oper. Res., 25(2):255--280, 2000.
\newblock
  \href{http://dx.doi.org/10.1287/moor.25.2.255.12228}{\path{doi:10.1287/moor.%
25.2.255.12228}}.

\bibitem[CGKT07]{ChuzhoyGKT2007}
\textsc{Julia Chuzhoy}, \textsc{Venkatesan Guruswami}, \textsc{Sanjeev Khanna},
  and \textsc{Kunal Talwar}.
\newblock \emph{Hardness of routing with congestion in directed graphs}.
\newblock In \emph{Proc.\ $39$th ACM Symp.\ on Theory of Computing (STOC)},
  pages 165--178. ACM, 2007.
\newblock
  \href{http://dx.doi.org/10.1145/1250790.1250816}{\path{doi:10.1145/1250790.1%
250816}}.

\bibitem[CHKS06]{ChekuriHKS2006}
\textsc{Chandra Chekuri}, \textsc{Mohammad~Taghi Hajiaghayi}, \textsc{Guy
  Kortsarz}, and \textsc{Mohammad~R. Salavatipour}.
\newblock \emph{Approximation algorithms for non-uniform buy-at-bulk network
  design}.
\newblock In \emph{Proc.\ $47$th IEEE Symp.\ on Foundations of Comp.\ Science
  (FOCS)}, pages 677--686. IEEE, 2006.
\newblock
  \href{http://dx.doi.org/10.1109/FOCS.2006.15}{\path{doi:10.1109/FOCS.2006.15%
}}.

\bibitem[Chr76]{Christofides1976}
\textsc{Nicos Christofides}.
\newblock \emph{Worst-case analysis of a new heuristic for the travelling
  salesman problem}.
\newblock Technical Report 388, Graduate School of Industrial Administration,
  Carnegie-Mellon University, 1976.

\bibitem[CK06]{ChuzhoyK2006}
\textsc{Julia Chuzhoy} and \textsc{Sanjeev Khanna}.
\newblock \emph{Hardness of directed routing with congestion}.
\newblock Technical Report TR06-109, Electronic Colloquium on Computational
  Complexity, 2006.
\newblock \href{http://eccc.hpi-web.de/report/2006/109}{\path{eccc:TR06-109}}.

\bibitem[CKK{\etalchar{+}}06]{ChawlaKKRS2006}
\textsc{Shuchi Chawla}, \textsc{Robert Krauthgamer}, \textsc{Ravi Kumar},
  \textsc{Yuval Rabani}, and \textsc{D.~Sivakumar}.
\newblock \emph{On the hardness of approximating multicut and sparsest-cut}.
\newblock Computational Complexity, 15(2):94--114, 2006.
\newblock (Preliminary version in {\em 20th IEEE Conference on Computational
  Complexity}, 2005).
\newblock
  \href{http://dx.doi.org/10.1007/s00037-006-0210-9}{\path{doi:10.1007/s00037-%
006-0210-9}}.

\bibitem[CKS06]{ChekuriKS2006}
\textsc{Chandra Chekuri}, \textsc{Sanjeev Khanna}, and \textsc{F.~Bruce
  Shepherd}.
\newblock \emph{An $o(sqrt(n))$ approximation and integrality gap for disjoint
  paths and unsplittable flow}.
\newblock Theory of Computing, 2(1):137--146, 2006.
\newblock
  \href{http://dx.doi.org/10.4086/toc.2006.v002a007}{\path{doi:10.4086/toc.200%
6.v002a007}}.

\bibitem[CMM06]{CharikarMM2006}
\textsc{Moses Charikar}, \textsc{Konstantin Makarychev}, and \textsc{Yury
  Makarychev}.
\newblock \emph{Near-optimal algorithms for unique games}.
\newblock In \emph{Proc.\ $38$th ACM Symp.\ on Theory of Computing (STOC)},
  pages 205--214. ACM, 2006.
\newblock
  \href{http://dx.doi.org/10.1145/1132516.1132547}{\path{doi:10.1145/1132516.1%
132547}}.

\bibitem[DH09]{DinurH2009}
\textsc{Irit Dinur} and \textsc{Prahladh Harsha}.
\newblock \emph{Composition of low-error 2-query {PCP}s using decodable
  {PCP}s}.
\newblock In \emph{Proc.\ $50$th IEEE Symp.\ on Foundations of Comp.\ Science
  (FOCS)}, pages 472--481. IEEE, 2009.
\newblock \href{http://eccc.hpi-web.de/report/2009/042}{\path{eccc:TR09-042}},
  \href{http://dx.doi.org/10.1109/FOCS.2009.8}{\path{doi:10.1109/FOCS.2009.8}}.

\bibitem[Din07]{Dinur2007}
\textsc{Irit Dinur}.
\newblock \emph{The {PCP} theorem by gap amplification}.
\newblock J. ACM, 54(3):12, 2007.
\newblock (Preliminary Version in {\em 38th STOC}, 2006).
\newblock \href{http://eccc.hpi-web.de/report/2005/046/}{\path{eccc:TR05-046}},
  \href{http://dx.doi.org/10.1145/1236457.1236459}{\path{doi:10.1145/1236457.1%
236459}}.

\bibitem[EH03]{EngebretsenH2003}
\textsc{Lars Engebretsen} and \textsc{Jonas Holmerin}.
\newblock \emph{Towards optimal lower bounds for clique and chromatic number}.
\newblock Theoretical Comp. Science, 299(1--3):537--584, 2003.
\newblock
  \href{http://dx.doi.org/10.1016/S0304-3975(02)00535-2}{\path{doi:10.1016/S03%
04-3975(02)00535-2}}.

\bibitem[Erl06]{Erlebach2006}
\textsc{Thomas Erlebach}.
\newblock \emph{Approximation algorithms for edge-disjoint paths and
  unsplittable flow}.
\newblock In \textsc{Evripidis Bampis}, \textsc{Klaus Jansen}, and
  \textsc{Claire Kenyon}, eds., \emph{Efficient Approximation and Online
  Algorithms}, LNCS, chapter~4, pages 97--134. Springer, 2006.
\newblock
  \href{http://dx.doi.org/10.1007/11671541_4}{\path{doi:10.1007/11671541_4}}.

\bibitem[FGL{\etalchar{+}}96]{FeigeGLSS1996}
\textsc{Uriel Feige}, \textsc{Shafi Goldwasser}, \textsc{L{\'a}szl{\'o}
  Lov{\'a}sz}, \textsc{Shmuel Safra}, and \textsc{Mario Szegedy}.
\newblock \emph{Interactive proofs and the hardness of approximating cliques}.
\newblock J. ACM, 43(2):268--292, March 1996.
\newblock (Preliminary version in {\em 32nd FOCS}, 1991).
\newblock
  \href{http://dx.doi.org/10.1145/226643.226652}{\path{doi:10.1145/226643.2266%
52}}.

\bibitem[FHW80]{FortuneHW1980}
\textsc{Steven Fortune}, \textsc{John~E. Hopcroft}, and \textsc{James Wyllie}.
\newblock \emph{The directed subgraph homeomorphism problem}.
\newblock Theoretical Comp. Science, 10(2):111--121, 1980.
\newblock
  \href{http://dx.doi.org/10.1016/0304-3975(80)90009-2}{\path{doi:10.1016/0304%
-3975(80)90009-2}}.

\bibitem[FL92]{FeigeL1992}
\textsc{Uriel Feige} and \textsc{L{\'a}szl{\'o} Lov{\'a}sz}.
\newblock \emph{Two-prover one-round proof systems: Their power and their
  problems (extended abstract)}.
\newblock In \emph{Proc.\ $24$th ACM Symp.\ on Theory of Computing (STOC)},
  pages 733--744. ACM, 1992.
\newblock
  \href{http://dx.doi.org/10.1145/129712.129783}{\path{doi:10.1145/129712.1297%
83}}.

\bibitem[FRT04]{FakcharoenpholRT2004}
\textsc{Jittat Fakcharoenphol}, \textsc{Satish Rao}, and \textsc{Kunal Talwar}.
\newblock \emph{A tight bound on approximating arbitrary metrics by tree
  metrics}.
\newblock J. Computer and System Sciences, 69(3):485--497, 2004.
\newblock (Preliminary version in {\em 35th STOC}, 2003).
\newblock
  \href{http://dx.doi.org/10.1016/j.jcss.2004.04.011}{\path{doi:10.1016/j.jcss%
.2004.04.011}}.

\bibitem[GKR{\etalchar{+}}03]{GuruswamiKRSY2003}
\textsc{Venkatesan Guruswami}, \textsc{Sanjeev Khanna}, \textsc{Rajmohan
  Rajaraman}, \textsc{F.~Bruce Shepherd}, and \textsc{Mihalis Yannakakis}.
\newblock \emph{Near-optimal hardness results and approximation algorithms for
  edge-disjoint paths and related problems}.
\newblock J. Computer and System Sciences, 67(3):473--496, 2003.
\newblock (Preliminary version in {\em 31st STOC}, 1999).
\newblock
  \href{http://dx.doi.org/10.1016/S0022-0000(03)00066-7}{\path{doi:10.1016/S00%
22-0000(03)00066-7}}.

\bibitem[GMR08]{GuruswamiMR2008}
\textsc{Venkatesan Guruswami}, \textsc{Rajsekar Manokaran}, and \textsc{Prasad
  Raghavendra}.
\newblock \emph{Beating the random ordering is hard: Inapproximability of
  maximum acyclic subgraph}.
\newblock In \emph{Proc.\ $49$th IEEE Symp.\ on Foundations of Comp.\ Science
  (FOCS)}, pages 573--582. IEEE, 2008.
\newblock
  \href{http://dx.doi.org/10.1109/FOCS.2008.51}{\path{doi:10.1109/FOCS.2008.51%
}}.

\bibitem[GT06a]{GuptaT2006}
\textsc{Anupam Gupta} and \textsc{Kunal Talwar}.
\newblock \emph{Approximating unique games}.
\newblock In \emph{Proc.\ $17$th Annual {ACM}-{SIAM} Symposium on Discrete
  Algorithms (SODA)}, pages 99--106. SIAM, 2006.
\newblock
  \href{http://dx.doi.org/10.1145/1109557.1109569}{\path{doi:10.1145/1109557.1%
109569}}.

\bibitem[GT06b]{GuruswamiT2006}
\textsc{Venkatesan Guruswami} and \textsc{Kunal Talwar}.
\newblock \emph{Hardness of low congestion routing in directed graphs}.
\newblock Technical Report TR06-141, Electronic Colloquium on Computational
  Complexity, 2006.
\newblock \href{http://eccc.hpi-web.de/report/2006/141}{\path{eccc:TR06-141}}.

\bibitem[GW95]{GoemansW1995}
\textsc{Michel~X. Goemans} and \textsc{David~P. Williamson}.
\newblock \emph{Improved approximation algorithms for maximum cut and
  satisfiability problems using semidefinite programming}.
\newblock J. ACM, 42(6):1115--1145, 1995.
\newblock (Preliminary version in {\em 26th STOC}, 1994).
\newblock
  \href{http://dx.doi.org/10.1145/227683.227684}{\path{doi:10.1145/227683.2276%
84}}.

\bibitem[Har04]{Harsha2004}
\textsc{Prahladh Harsha}.
\newblock \href{http://hdl.handle.net/1721.1/26720} {\emph{Robust {PCP}s of
  Proximity and Shorter {PCP}s}}.
\newblock Ph.D. thesis, Massachusetts Institute of Technology, September 2004.

\bibitem[H{\aa}s01]{Hastad2001}
\textsc{Johan H{\aa}stad}.
\newblock \emph{Some optimal inapproximability results}.
\newblock J. ACM, 48(4):798--859, July 2001.
\newblock (Preliminary Version in {\em 29th STOC}, 1997).
\newblock
  \href{http://dx.doi.org/10.1145/502090.502098}{\path{doi:10.1145/502090.5020%
98}}.

\bibitem[Kho02]{Khot2002}
\textsc{Subhash Khot}.
\newblock \emph{On the power of unique 2-prover 1-round games}.
\newblock In \emph{Proc.\ $34$th ACM Symp.\ on Theory of Computing (STOC)},
  pages 767--775. ACM, 2002.
\newblock
  \href{http://dx.doi.org/10.1145/509907.510017}{\path{doi:10.1145/509907.5100%
17}}.

\bibitem[Kho05]{Khot2005a}
---{}---{}---.
\newblock \emph{Guest column: inapproximability results via long code based
  {PCP}s}.
\newblock SIGACT News, 36(2):25--42, 2005.
\newblock
  \href{http://dx.doi.org/10.1145/1067309.1067318}{\path{doi:10.1145/1067309.1%
067318}}.

\bibitem[KKMO07]{KhotKMO2007}
\textsc{Subhash Khot}, \textsc{Guy Kindler}, \textsc{Elchanan Mossel}, and
  \textsc{Ryan O'Donnell}.
\newblock \emph{Optimal inapproximability results for {MAX-CUT} and other
  2-variable {CSP}s?}
\newblock SIAM J. Computing, 37(1):319--357, 2007.
\newblock (Preliminary version in {\em 45th FOCS}, 2004).
\newblock \href{http://eccc.hpi-web.de/report/2005/101}{\path{eccc:TR05-101}},
  \href{http://dx.doi.org/10.1137/S0097539705447372}{\path{doi:10.1137/S009753%
9705447372}}.

\bibitem[Kle96]{Kleinberg1996}
\textsc{Jon~M Kleinberg}.
\newblock \href{http://hdl.handle.net/1721.1/11013} {\emph{Approximation
  algorithms for disjoint paths problems}}.
\newblock Ph.D. thesis, Massachusetts Institute of Technology, May 1996.

\bibitem[KR08]{KhotR2008}
\textsc{Subhash Khot} and \textsc{Oded Regev}.
\newblock \emph{Vertex cover might be hard to approximate to within
  2-$\epsilon$}.
\newblock J. Computer and System Sciences, 74(3):335--349, 2008.
\newblock (Preliminary Version in {\em 18th IEEE Conference on Computational
  Complexity}, 2003).
\newblock
  \href{http://dx.doi.org/10.1016/j.jcss.2007.06.019}{\path{doi:10.1016/j.jcss%
.2007.06.019}}.

\bibitem[KS01]{KolliopoulosS2001}
\textsc{Stavros~G. Kolliopoulos} and \textsc{Clifford Stein}.
\newblock \emph{Approximation algorithms for single-source unsplittable flow}.
\newblock SIAM J. Computing, 31(3):919--946, 2001.
\newblock (Preliminary version in {\em 38th FOCS}, 1997).
\newblock
  \href{http://dx.doi.org/10.1137/S0097539799355314}{\path{doi:10.1137/S009753%
9799355314}}.

\bibitem[KV05]{KhotV2005}
\textsc{Subhash Khot} and \textsc{Nisheeth~K. Vishnoi}.
\newblock \emph{The unique games conjecture, integrality gap for cut problems
  and embeddability of negative type metrics into l$_{\mbox{1}}$}.
\newblock In \emph{Proc.\ $46$th IEEE Symp.\ on Foundations of Comp.\ Science
  (FOCS)}, pages 53--62. IEEE, 2005.
\newblock
  \href{http://dx.doi.org/10.1109/SFCS.2005.74}{\path{doi:10.1109/SFCS.2005.74%
}}.

\bibitem[LR99]{LeightonR1999}
\textsc{Frank~Thomson Leighton} and \textsc{Satish Rao}.
\newblock \emph{Multicommodity max-flow min-cut theorems and their use in
  designing approximation algorithms}.
\newblock J. ACM, 46(6):787--832, 1999.
\newblock
  \href{http://dx.doi.org/10.1145/331524.331526}{\path{doi:10.1145/331524.3315%
26}}.

\bibitem[MH99]{MahajanH1999}
\textsc{Sanjeev Mahajan} and \textsc{Ramesh Hariharan}.
\newblock \emph{Derandomizing approximation algorithms based on semidefinite
  programming}.
\newblock SIAM J. Computing, 28(5):1641--1663, 1999.
\newblock (Preliminary version in {\em 36th FOCS}, 1995).
\newblock
  \href{http://dx.doi.org/10.1137/S0097539796309326}{\path{doi:10.1137/S009753%
9796309326}}.

\bibitem[MM09]{MakarychevM2009}
\textsc{Konstantin Makarychev} and \textsc{Yury Makarychev}.
\newblock \emph{How to play unique games on expanders}, 2009.
\newblock \href{http://arxiv.org/abs/0903.0367}{\path{arXiv:0903.0367}}.

\bibitem[MOO05]{MosselOO2005}
\textsc{Elchanan Mossel}, \textsc{Ryan O'Donnell}, and \textsc{Krzysztof
  Oleszkiewicz}.
\newblock \emph{Noise stability of functions with low in.uences invariance and
  optimality}.
\newblock In \emph{Proc.\ $46$th IEEE Symp.\ on Foundations of Comp.\ Science
  (FOCS)}, pages 21--30. IEEE, 2005.
\newblock \href{http://arxiv.org/abs/math/0503503}{\path{arXiv:math/0503503}},
  \href{http://dx.doi.org/10.1109/SFCS.2005.53}{\path{doi:10.1109/SFCS.2005.53%
}}.

\bibitem[MR08]{MoshkovitzR2008b}
\textsc{Dana Moshkovitz} and \textsc{Ran Raz}.
\newblock \emph{Two query {PCP} with sub-constant error}.
\newblock In \emph{Proc.\ $49$th IEEE Symp.\ on Foundations of Comp.\ Science
  (FOCS)}, pages 314--323. IEEE, 2008.
\newblock \href{http://eccc.hpi-web.de/report/2008/071}{\path{eccc:TR08-071}},
  \href{http://dx.doi.org/10.1109/FOCS.2008.60}{\path{doi:10.1109/FOCS.2008.60%
}}.

\bibitem[O'D08]{Odonnell2008}
\textsc{Ryan O'Donnell}.
\newblock \emph{Some topics in analysis of {B}oolean functions}.
\newblock In \emph{Proc.\ $40$th ACM Symp.\ on Theory of Computing (STOC)},
  pages 569--578. ACM, 2008.
\newblock \href{http://eccc.hpi-web.de/report/2008/055}{\path{eccc:TR08-055}},
  \href{http://dx.doi.org/10.1145/1374376.1374458}{\path{doi:10.1145/1374376.1%
374458}}.

\bibitem[Rag08]{Raghavendra2008}
\textsc{Prasad Raghavendra}.
\newblock \emph{Optimal algorithms and inapproximability results for every
  {CSP}?}
\newblock In \emph{Proc.\ $40$th ACM Symp.\ on Theory of Computing (STOC)},
  pages 245--254. ACM, 2008.
\newblock
  \href{http://dx.doi.org/10.1145/1374376.1374414}{\path{doi:10.1145/1374376.1%
374414}}.

\bibitem[Raz98]{Raz1998}
\textsc{Ran Raz}.
\newblock \emph{A parallel repetition theorem}.
\newblock SIAM J. Computing, 27(3):763--803, June 1998.
\newblock (Preliminary Version in {\em 27th STOC}, 1995).
\newblock
  \href{http://dx.doi.org/10.1137/S0097539795280895}{\path{doi:10.1137/S009753%
9795280895}}.

\bibitem[RS95]{RobertsonS1995}
\textsc{Neal Robertson} and \textsc{Paul~D. Seymour}.
\newblock \emph{Graph minors. {XIII}. {T}he disjoint paths problem}.
\newblock Journal of Combinatorial Theory, Series B, 63(1):65--110, 1995.
\newblock
  \href{http://dx.doi.org/10.1006/jctb.1995.1006}{\path{doi:10.1006/jctb.1995.%
1006}}.

\bibitem[RS96]{RubinfeldS1996}
\textsc{Ronitt Rubinfeld} and \textsc{Madhu Sudan}.
\newblock \emph{Robust characterizations of polynomials with applications to
  program testing}.
\newblock SIAM J. Computing, 25(2):252--271, April 1996.
\newblock (Preliminary Version in {\em 23rd STOC}, 1991 and {\em 3rd SODA},
  1992).
\newblock
  \href{http://dx.doi.org/10.1137/S0097539793255151}{\path{doi:10.1137/S009753%
9793255151}}.

\bibitem[RS97]{RazS1997}
\textsc{Ran Raz} and \textsc{Shmuel Safra}.
\newblock \emph{A sub-constant error-probability low-degree test, and a
  sub-constant error-probability {PCP} characterization of~{NP}}.
\newblock In \emph{Proc.\ $29$th ACM Symp.\ on Theory of Computing (STOC)},
  pages 475--484. ACM, 1997.
\newblock
  \href{http://dx.doi.org/10.1145/258533.258641}{\path{doi:10.1145/258533.2586%
41}}.

\bibitem[RT87]{RaghavanT1987}
\textsc{Prabhakar Raghavan} and \textsc{Clark~D. Thompson}.
\newblock \emph{Randomized rounding: a technique for provably good algorithms
  and algorithmic proofs}.
\newblock Combinatorica, 7(4):365--374, 1987.
\newblock
  \href{http://dx.doi.org/10.1007/BF02579324}{\path{doi:10.1007/BF02579324}}.

\bibitem[RZ]{RaoZunpub}
\textsc{Satish Rao} and \textsc{Shuheng Zhou}.
\newblock (unpublished).

\bibitem[ST00]{SamorodnitskyT2000}
\textsc{Alex Samorodnitsky} and \textsc{Luca Trevisan}.
\newblock \emph{A {PCP} characterization of {NP} with optimal amortized query
  complexity}.
\newblock In \emph{Proc.\ $32$nd ACM Symp.\ on Theory of Computing (STOC)},
  pages 191--199. ACM, 2000.
\newblock
  \href{http://dx.doi.org/10.1145/335305.335329}{\path{doi:10.1145/335305.3353%
29}}.

\bibitem[Tre01]{Trevisan2001}
\textsc{Luca Trevisan}.
\newblock \emph{Non-approximability results for optimization problems on
  bounded degree instances}.
\newblock In \emph{Proc.\ $33$rd ACM Symp.\ on Theory of Computing (STOC)},
  pages 453--461. ACM, 2001.
\newblock
  \href{http://dx.doi.org/10.1145/380752.380839}{\path{doi:10.1145/380752.3808%
39}}.

\bibitem[Tre08]{Trevisan2008}
---{}---{}---.
\newblock \emph{Approximation algorithms for unique games}.
\newblock Theory of Computing, 4(1):111--128, 2008.
\newblock (Preliminary version in {\em 46th FOCS}, 2005).
\newblock \href{http://eccc.hpi-web.de/report/2005/034}{\path{eccc:TR05-034}},
  \href{http://dx.doi.org/10.4086/toc.2008.v004a005}{\path{doi:10.4086/toc.200%
8.v004a005}}.

\bibitem[VV04]{VaradarajanV2004}
\textsc{Kasturi~R. Varadarajan} and \textsc{Ganesh Venkataraman}.
\newblock \emph{Graph decomposition and a greedy algorithm for edge-disjoint
  paths}.
\newblock In \emph{Proc.\ $15$th Annual {ACM}-{SIAM} Symposium on Discrete
  Algorithms (SODA)}, pages 379--380. SIAM, 2004.
\newblock
  \href{http://dx.doi.org/10.1145/982792.982846}{\path{doi:10.1145/982792.9828%
46}}.

\end{thebibliography}

}
\end{document}